%% file: Dissertation.tex
\documentclass[12pt,a4paper,DIV=11]{scrreprt} 
\usepackage[german,english]{babel}
\usepackage[utf8]{inputenc}
\usepackage{graphicx}

\input{\mypath mathcommands.tex}

\usepackage[numbers]{natbib} 
\usepackage[format=plain,hangindent=.3 cm]{caption}
\usepackage{xfrac}
\usepackage[table]{xcolor}

\newtheorem{definition}{Definition}[chapter]
\newtheorem{theorem}[definition]{Theorem}
\newtheorem{lemma}[definition]{Lemma}
\newtheorem{corollary}[definition]{Corollary}

\newtheorem{assumption}{Assumption}

\theoremstyle{remark}
\newtheorem{example}{Example}
\newtheorem*{remark}{Remark}

\begin{document}
\input{deckblatt.tex}

\begin{abstract} \begin{center} \begin{minipage}{0.8 \textwidth}
\section*{\centering Deutschsprachige Kurzfassung}
\vspace{1 em}
Diese Dissertation beschäftigt sich mit affinen Prozessen und deren Anwendung im Bereich der Finanzmathematik. Im ersten Teil betrachten wir die Theorie zeitinhomogener affiner Prozesse auf allgemeinen Zustandsräumen. Zuerst werden zeitinhomogene Markov Prozesse genau definiert. Für stochastisch stetige affine Prozesse zeigen wir, dass immer eine càdlàg Modifikation existiert. Anschließend betrachten wir die Regularität und die Semimartingal Eigenschaft von affinen Prozessen. Im Gegensatz zum zeithomogenen Fall sind zeitinhomogene affine Prozesse im Allgemeinen weder regulär noch Semimartingale. Der zeitinhomogene Fall führt daher zu neuen und interessanten Fragestellungen. Nimmt man an, dass ein affiner Prozess auch ein Semimartingal ist, so kann man auch ohne Regularität zeigen, dass die Parameterfunktionen verallgemeinerte Riccati Integralgleichungen lösen. Diese Aussage verallgemeinert ein wichtiges Resultat für zeithomogene affine Prozesse. Wir zeigen weiters, dass stochastisch stetige affine Semimartingale im Wesentlichen durch deterministische Zeittransformationen aus absolut stetig affinen Semimartingalen entstehen. 
Diese Prozesse sind eine Verallgemeinerung der zeithomogenen regulär affinen Prozesse. 

Im zweiten Teil betrachten wir die Klasse der affinen LIBOR Marktmodelle. Wir modifizieren das ursprüngliche Modell, sodass neben nichtnegativen affinen Prozessen auch reellwertige affine Prozesse verwendet werden können. Numerische Beispiele zeigen, dass dadurch flexiblere Volatilitätsoberflächen erzeugt werden können. 
Weiters führen wir die Klasse der affinen Inflationsmarktmodelle ein. Diese Modelle erweitern affine LIBOR Marktmodelle auf Inflationsmärkte.  Ein Kalibrierungsbeispiel zeigt, dass beobachtete Marktpreise von Inflationsderivaten sehr gut wiedergegeben werden können. 
\end{minipage} \end{center} \end{abstract}

\selectlanguage{english}

\begin{abstract} \begin{center} \begin{minipage}{0.8 \textwidth}
\section*{\centering Abstract}
\vspace{1 em}
This thesis is devoted to the study of affine processes and their applications in financial mathematics. In the first part we consider the theory of time-inhomogeneous affine processes on general state spaces. We present a concise setup for time-inhomogeneous Markov processes. 
For stochastically continuous affine processes we show that there always exists a càdlàg modification. Afterwards we consider the regularity and the semimartingale property of affine processes. Contrary to the time-homogeneous case, time-inhomogeneous affine processes are in general neither regular nor semimartingales and the time-inhomogeneous case raises many new and interesting questions. Assuming that an affine process is a semimartingale, we show that even without regularity the parameter functions satisfy generalized Riccati integral equations. This generalizes an important result for time-homogeneous affine processes. We also show that stochastically continuous affine semimartingales are essentially generated by deterministic time-changes of what we call absolutely continuously affine semimartingales. These processes generalize time-homogeneous regular affine processes.

In the second part we consider the class of affine LIBOR market models. We contribute to this class of models in two ways. First, we modify the original setup of the affine LIBOR market models in such a way that next to nonnegative affine processes real-valued affine processes can also be used. Numerical examples show that this allows for more flexible implied volatility surfaces. Second, we introduce the class of affine inflation market models, an extension of the affine LIBOR market models. A calibration example shows that these models perform very well in fitting market-observed prices of inflation derivatives.
\end{minipage} \end{center} \end{abstract}

\begin{abstract} \begin{center} \begin{minipage}{0.8 \textwidth}
\section*{\centering Acknowledgements}
\vspace{1 em}
First and foremost I want to thank my advisor, Wolfgang Müller, who encouraged me to pursue a PhD in this area. His door was always open and I am very grateful for all the time he spent discussing and answering my many questions. I would also like to thank the whole staff of the Institute of statistics, who filled the daily coffee breaks with many interesting discussions. It was a very familiar environment and I was very happy to spend my time there. Special thanks go to Ernst Stadlober, who was very generous with my requests to attend conferences and scientific meetings. 

I also want to thank Christa Cuchiero, Philipp Harms, Josef Teichmann and Thorsten Schmidt, with whom I had many interesting and helpful discussions on affine processes. I am especially grateful to Josef Teichmann for inviting me to ETH Zürich and pointing me in some interesting directions. 

Special thanks go to my family for supporting all my decision and giving me the opportunity to get to this point. Finally, and perhaps most importantly, I want to thank Steffi Klinger, who encouraged me during the whole time of my PhD, unquestioningly put up with all the hours spent on writing this thesis and who always listened to my doubts and problems.
\end{minipage} \end{center} \end{abstract}

\addtocontents{toc}{\protect\thispagestyle{empty}}
\tableofcontents

\chapter*{Introduction}
\addcontentsline{toc}{chapter}{Introduction}
\setcounter{page}{1}

This dissertation is centered around the class of affine processes and their applications to financial markets. Time-inhomogeneous affine processes are informally described as $E$-valued Markov processes, where the logarithm of the conditional moment generating function is affine in the current state of the process, i.e.
$$\EV{ \e^{\scalarprod{u, X_{t}}} \vert \F_s}  = \ex{\phi_{s,t}(u) + \scalarprod{\psi_{s,t}(u), X_s}}, \quad u \in \i \R^d, x \in E,$$
with deterministic functions $\phi$ and $\psi$. This means that the transition operator of the Markov process maps exponentially affine functions into exponentially affine functions. 

In the first part of this thesis we introduce time-inhomogeneous affine processes on general state spaces. While time-homogeneous affine processes ($\phi$ and $\psi$ depend only on $t-s$) are extensively studied, time-inhomo\-geneous affine processes did not get the same amount of attention. 
We focus on continuously affine transition functions (referring to the continuity of $\phi$ and $\psi$), which basically correspond to stochastically continuous affine processes. For affine processes with a continuously affine transition function we are able to show that there always exists a càdlàg modification of the process. This extends a similar result for time-homogeneous affine processes. 

We then consider the question of regularity and the semimartingale property of time-inhomogeneous affine processes. Regularity is related to differentiability of $\phi$ and $\psi$ with respect to the time parameters. For time-homogeneous continuously affine processes regularity always holds. It is used to infer that the parameter functions are solutions of generalized Riccati differential equations and that there exist differentiated semimartingale characteristics which are affine in the current state of the process. We give examples that in the time-inhomogeneous case neither regularity nor the semimartingale property holds automatically. However, for time-inhomogeneous continuously affine processes which are simultaneously semimartingales we can extend many results of the time-homogeneous case. 
We are able to show that in this case the parameter functions are solutions to generalized Riccati integral equations, where the integral is a Lebesgue-Stieltjes integral with respect to a deterministic function $G$. Furthermore, the semimartingale characteristics are absolutely continuous with respect to $G$ and depend in an affine way on the current state of the process. If  the functions $\phi$ and $\psi$ are even absolutely continuous, we show that one can choose $G(t) = t$. We say that such affine semimartingales have an absolutely continuously affine transition function and propose that absolutely continuously affine processes are a suitable generalization of time-homogeneous regular affine processes. The deterministic function $G$ can be interpreted as a time-change. It follows that continuously affine semimartingales are essentially equal to deterministically time-changed absolutely continuously affine semimartingales.

In the second part of this dissertation we consider practical applications of affine processes based on the class of affine LIBOR market models as introduced in \citet{KPT11}. The class of affine LIBOR market models uses the analytical tractability of affine processes to describe an analytically tractable class of LIBOR market models. In order to guarantee nonnegative interest rates only nonnegative affine processes can be used. In chapter \ref{cha:cosh} (which corresponds to \citet{MW14}) we extend the affine LIBOR market models in such a way that real-valued affine processes can also be used without destroying the nonnegativity of interest rates. This has the advantage that the resulting models are more flexible in producing market-observed implied volatility smiles. We illustrate this by numerical examples. 

Chapter \ref{cha:infl} extends affine LIBOR market models to inflation markets. In inflation markets there exist two different types of liquidly traded inflation swaps. Market models considered so far are not able to give closed formulas for both types of swaps. Here the analytical tractability of affine inflation market models can be used to get closed formulas for both swaps. Furthermore, calls and puts on these swap rates can be priced by Fourier methods. This covers pricing for the most liquidly traded derivatives in inflation markets. Using these formulas, we are able to calibrate an affine inflation market model to market data. This chapter corresponds to \citet{WA15}.

\chapter{Time-inhomogeneous affine processes} \label{cha:affineproc}
Time-homogeneous affine processes are time-homogeneous Markov processes with a state space $E$, such that the transition probabilities $P_{t}$ satisfy
$$ \int_E \e^{\scalarprod{u,\xi}} P_{t}(x,\dd{\xi}) = \Phi_{t}(u) \e^{\scalarprod{\psi_{t}(u),x}} $$
for $x \in E$, $t \geq 0$ and all $u$ such that $x \mapsto \e^{\scalarprod{u,x}}$ is bounded in $E$.

The class of affine processes includes many well-known and widely used stochastic processes. In particular, Lévy processes, Ornstein-Uhlenbeck processes (\citet{SA99}, section 17), CIR processes (\citet{CIR85}) and 
Wishardt-processes (\citet{BR91}) belong to this class. 
Affine processes go back to \citet{KW71}, where continuous-time limits of Galton-Watson branching processes were studied. Later affine processes with values in the canonical state space $\R^m_{\geq 0} \times \R^n$ were introduced and applied in a series of papers, compare \citet{DU96}, \citet{DS00}, \citet{DPS00}. An important break-through was presented in \citet{DFS03}, where a complete characterization of regular affine processes taking values in the canonical state space was given. Here regularity refers to the existence of the time-derivatives of $\Phi$ and $\psi$. It was shown that $\Phi$ and $\psi$ are solutions of generalized Riccati differential equations. Furthermore, it was proved that regular affine processes are semimartingales whose differentiated semimartingale characteristics are affine in the current state of the process.
It was not until recently, that in another series of papers it was shown that the assumption of regularity is automatically satisfied for time-homogeneous stochastically continuous affine processes. For the canonical state space this was done in \citet{KST11}, and for general state spaces in \citet{KST11b} and \citet{CT13}. 

Time-inhomogeneous affine processes are analogously defined with the functions $\Phi$ and $\psi$ depending on two time parameters. They were first introduced in \citet{FI05}, where, assuming stronger forms of regularity, the results in \citet{DFS03} were extended to the time-inhomogeneous case. In this chapter we explore how the theory of time-homogeneous affine processes with a general state space extends to the time-inhomogeneous case without such strong regularity assumptions. We show that the time-inhomogeneous case raises new and interesting questions, which do not appear in the time-homogeneous case. 

We start by defining time-inhomo\-geneous Markov processes in a way which fits especially well with semimartingale theory. 
In section \ref{sec:affineproctheory} we introduce time-inhomo\-geneous affine processes on general state spaces and show that stochastic continuity of the Markov process is essentially equivalent to the continuity of $\Phi$ and $\psi$. This leads to the definition of a continuously affine transition function.  In section \ref{sec:cadlagversion} we extend the results of \citet{CT13} to the time-inhomogeneous case and proof that continuously affine processes have a càdlàg modification which possesses the strong Markov property. 

The last part of chapter \ref{cha:affineproc} focuses on regularity questions and the semimartingale property. While in the time-homogeneous case regularity is automatically satisfied, this is not true for time-inhomogeneous affine processes. Furthermore, the semimartingale property no longer holds in general, as examples show. 
Under mild technical conditions we are able to show that for Markov semimartingales with a continuously affine transition function the semimartingale characteristics are absolutely continuous with respect to a deterministic increasing function $G$. Additionally, they depend in an affine way on the current state of the process. 
From this we get that the functions $\Phi$ and $\psi$ are solutions of generalized Riccati integral equations, where the integrals are Lebesgue-Stieltjes integrals with respect to the function $G$. Hence even without regularity many of the results for time-homogeneous affine processes generalize.

The function $G$ can be used to time-change the process. Fixing the probability measure, the time-changed process turns out to be an affine process with a new affine transition function, where the new $\Phi$ and $\psi$ satisfy generalized Riccati integral equations with respect to the Lebesgue measure, i.e. they are even absolutely continuous. This leads to the definition of an absolutely continuously affine transition function. The time-changed process then has differentiated semimartingale characteristics, which depend on time and affinely on the current state of the process. We consider such processes as a suitable generalization of time-homogeneous regular  affine processes. 
The class of  Markov semimartingales with continuously affine transition functions is therefore essentially equal to the class of deterministically time-changed Markov semimartingales with an absolutely continuously (`regular') affine transition function. 

\section{Markov processes} \label{sec:markov}

\begin{definition}
A transition kernel $N$ on a measurable space $(E,\mathcal{E})$ is a map
\begin{align*}
N: E \times \mathcal{E} & \rightarrow [0,\infty], 
\end{align*}
such that
\begin{itemize}
\item for each $x \in E$ the map $A \mapsto N(x,A)$ is a measure,
\item for each $A \in \mathcal{E}$ the map $x \mapsto N(x,A)$ is measurable.
\end{itemize}
If furthermore $N(x,E) = 1$ for all $x \in E$, then $N$ is a transition probability. If $N(x,E) < 1$ we can add a cemetery point $\Delta$ to $E$ and set $N(x,\{ \Delta \}) = 1 - N(x,E)$ and $N(\Delta,\{ \Delta \}) = 1$, so that $N$ is a transition probability on $(E_\Delta, \mathcal{E}_\Delta)$, where $E_\Delta = E \cup \{ \Delta \}$ and $ \mathcal{E}_\Delta = \sigma(\mathcal{E},\{ \Delta \})$. 
\end{definition}
Functions $f$ mapping from $E$ to $\R$ or $\C$ are extended to $E_\Delta$ by setting $f(\Delta) = 0$. To shorten notation for the rest of this section we use $E$ for the state space, even if there exists a cemetery point. A transition probability introduces a transition operator on the set of bounded (or alternatively nonnegative) measurable functions $f: E \rightarrow \R$ by
\begin{equation*}
N f(x) := \int_E f(y) N(x,\dd{y}).
\end{equation*}
Note that we use the same letter to denote the transition probability and the corresponding operator. For two kernels $M, N$ it is possible to define $M N$ by 
$$M N (x,A) := \int_E N(y,A) M(x,\dd{y}),$$
which is again a transition kernel. Note that $x \mapsto Nf(x)$ is measurable and that $(M N) f (x) = M (N f) (x)$.
\begin{definition}
A transition function on $(E,\mathcal{E})$ is a collection $\{P_{s,t}\}_{0 \leq s \leq t}$ of transition probabilities on $(E,\mathcal{E})$ such that $$P_{s,u} P_{u,t} = P_{s,t}$$ for all $0 \leq s < u < t$ and $P_{s,s}(x,\cdot) = \delta_x(\cdot)$, where $\delta_x(\cdot)$ denotes the Dirac measure. The transition function is called time-homogeneous if $P_{s,t} = P_{t-s}$ for a collection $\{P_t\}_{t \geq 0}$ of transition probabilities.
\end{definition}
\begin{definition} \label{def:markov1}
A process $X$ with state space $(E,\mathcal{E})$ on the filtered probability space $(\Omega,\mathcal{A},\F,\PM)$ with $\F := (\F_t)_{t \geq 0}$, such that $X_s = \Delta$ implies $X_t = \Delta$ for all $t > s$, is a time-inhomogeneous Markov process with transition function (operator) $\{P_{s,t}\}$ on $(E,\mathcal{E})$, if it is $\F$-adapted and 
\begin{equation} \EV{f(X_t) \vert \F_s} = P_{s,t} f(X_s) \label{eq:markovprop1} \end{equation}
 for all times $s<t$ and bounded measurable $f: E \rightarrow \R$. It is a time-homogeneous Markov process if  the transition function is time-homogeneous.
\end{definition}
As with the analogy between transition functions and transition operators one can equivalently formulate the Markov property by requiring that for $B \in \mathcal{E}$, $$\PM(X_{t} \in B \vert \F_s) = P_{s,t}(X_s,B).$$

One can always find a Markov process with a given transition function. This is the statement of the following lemma (see e.g. I.3 in \citet{RY10}).
\begin{lemma} \label{lem:markovconstr}
Consider the coordinate process $X$ on the filtered space $(E^{\R_{\geq 0}},\F^X_\infty,\F^X)$ with $\F^X_\infty := \sigma((X_t)_{t \geq 0})$ and $\F_t^X = \sigma((X_t)_{0 \leq s \leq t})$. For a given transition function $\{ P_{s,t} \}$ and start distribution $\mu$ on $E$ there is a unique probability measure $\PM$ on $(E^{\R_{\geq 0}},\F^X_\infty)$, such that $X$ is a Markov process
on $(E^{\R_{\geq 0}},\F^X_\infty,\F^X,\PM)$ and $X_0$ has distribution $\mu$. 
\end{lemma}
This motivates an alternative definition of a Markov process, not with respect to a fixed probability measure $\PM$, but with a whole family of probability measures. We demonstrate this in the time-homogeneous case first. Consider
\begin{enumerate}[i) ]
\item an $\F$-adapted $(E,\mathcal{E})$-valued stochastic process $X$ on $(\Omega,\mathcal{A},\F)$. If $E$ contains a cemetery point $\Delta$ it has to hold that $X_s = \Delta$ implies $X_t = \Delta$ for all $t > s$. 
\item Assume that for every $x \in E$ there exists a measure $\PM^{x}$ on $(\Omega,\mathcal{A})$, such that $\PM^x(X_0 = x) = 1$. Additionally, assume that for all $t \geq 0$ and $B \in \mathcal{E}$ the map $x \mapsto \PM^x(X_t \in B)$ is measurable. In other words $P_t(x,B) := \PM^x(X_t \in B)$ is a transition probability on $(E,\mathcal{E})$.
\end{enumerate}
Denote by $\mathbb{E}^x$ the expectation value with respect to the measure $\PM^x$. 
\begin{definition}\label{def:markov2hom}
The combination $(X,\F,\{\PM^x\})$ is a time-homogeneous Mar\-kov process on $(\Omega,\mathcal{A})$, if i) and ii) are satisfied and for each bounded measurable function $f: E \rightarrow \R$ and $x \in E$, $s,t \geq 0$
\begin{equation} \label{eq:markov2hom}
 \EV[x]{f(X_{s+t}) \vert \F_s} = P_{t}f(X_s).
\end{equation}
\end{definition}
Note that the tower property for conditional expectations gives that $\{P_t\}$ is a time-homogeneous transition function. With this definition $X$ is a homogeneous Markov process in the sense of Definition \ref{def:markov1} on $(\Omega,\mathcal{A},\F,\PM^x)$ for each $x \in E$. Contrary, for the transition function $\{P_t\}$ of a time-homogeneous Markov process in the sense of Definition \ref{def:markov1} we can construct probability measures $\PM^{x}$ using Lemma \ref{lem:markovconstr} with starting distribution $\mu = \delta_x$.  Then $(X,\F^X,\{\PM^{x}\})$ is a Markov process on $(\Omega,\F_\infty^X)$ in the sense of Definition \ref{def:markov2hom}. This is then called the canonical realization of $\{ P_t \}$. 

Sometimes the Markov property is written as
$$ \EV[x]{f(X_{t+s}) \vert \F_s} = \EV[X_s]{f(X_t)}.$$
Here the right hand side is defined as $\EV[X_s]{f(X_t)} = g(X_s)$, where $$g(x) = \EV[x]{f(X_{t})} = P_t f(x).$$ Note that $g$ is measurable and that $\EV[X_s]{f(X_t)}$ is a $\sigma(X_s)$-measurable random variable.

Next we generalize Definition \ref{def:markov2hom} to time-inhomogeneous Markov processes. For this we replace ii) by
\begin{enumerate}
\item[ii')] Assume that for each $r \geq 0$ and $x \in E$ there exists a probability measure $\PM^{(r,x)}$ on $(\Omega,\mathcal{A})$ such that $\PM^{(r,x)}(X_0 = x) = 1$. Additionally, assume that for all $s,t \geq 0$, $B \in \mathcal{E}$ the map $x \mapsto \PM^{(s,x)}(X_t \in B)$ is measurable.  This implies that $P_{s,s+t}(x,B) := \PM^{(s,x)}(X_t \in B)$ is a transition probability. 
\end{enumerate}

\begin{definition}\label{def:markov2inhom}
The combination $(X,\F,\{\PM^{(r,x)}\})$ is a time-inhomogeneous Mar\-kov process on $(\Omega,\mathcal{A})$, if i) and ii') are satisfied and for each bounded measurable function $f: E \rightarrow \R$ and $x \in E$, $r,s,t \geq 0$
\begin{equation} \label{eq:markov2inho} \EV[(r,x)]{f(X_{t+s}) \vert \F_s} = P_{r+s,r+s+t} f(X_s) =\EV[(r+s,X_s)]{f(X_t)}. \end{equation}
\end{definition}

In this case the tower property gives that $\{ P_{s,t} \}$ is a (time-inhomogeneous) transition function. 

\begin{remark}
Definiton \ref{def:markov2inhom} is not standard. For instance in \citet{GS75} a Markov process is defined as a system $(X,\F_s^r,\{\PM^{(r,x)}\})$, where 
\begin{itemize}
\item $(\mathcal{F}^r_s)_{0 \leq r \leq s \leq \infty}$ is a family of $\sigma$-algebras, such that $\F_s^r \subset \F_t^v$ for $0 \leq v \leq r \leq s \leq t$,
\item $\PM^{(r,x)}$ are probability measures on $(\Omega,\F_\infty^r)$ for $x \in E$, such that $x \mapsto \PM^{(r,x)}(X_t \in B)$ is measurable and 
\item $X$ is a process such that $X_t$ is $\F_t^t$-measurable and for each bounded measurable $f$ and $0 \leq r \leq s \leq t < \infty$
$$ \EV[(r,x)]{f(X_{t}) \vert \F^r_s} = \EV[(s,X_s)]{f(X_t)}.$$
\end{itemize}
The canonical choice for the family of $\sigma$-algebras is $\F_s^r = \sigma(X_t, r \leq t \leq s)$. We believe that Definiton \ref{def:markov2inhom} is more convenient, especially if one considers Markov processes which are semimartingales, as done in section \ref{sec:regularity}. 
\end{remark}
With Definition \ref{def:markov2inhom} $X$ is a time-inhomogeneous Markov process in the sense of Definition \ref{def:markov1} on $(\Omega,\mathcal{A},\F,\PM^{(r,x)})$ with transition function $\{P_{r+s,r+t}\}_{0 \leq s \leq t}$ for each $(r,x)$. On the other hand given a time-inhomogeneous transition function $\{ P_{s,t} \}$ we can construct probability measures $\PM^{(r,x)}$ by applying Lemma \ref{lem:markovconstr} using the transition function $\{P_{r+s,r+t}\}_{0 \leq s \leq t}$ and starting measure $\mu = \delta_x$, such that the coordinate process $X$ is a Markov process on $(E^{\R_{\geq 0}}, \F_\infty^X, \F^X,\PM^{(r,x)})$. Then $(X,\F^X,\{\PM^{(r,x)}\})$ is a time-inhomogeneous Markov process in the sense of Definition \ref{def:markov2inhom}. We call  $(X,\F^X,\{\PM^{(r,x)}\})$ the canonical realization of $\{P_{s,t} \}$. 

For the canonical realization of a transition function one can define the shift operators $\theta_t:  E^{\R_{\geq 0}} \rightarrow E^{\R_{\geq 0}}, \theta_t(\omega)(s) := \omega(s+t),$
so that \begin{equation} X_s(\theta_t(\omega)) = X_{s+t}(\omega). \label{eq:Xshiftprop} \end{equation}
Alternatively one could assume the existence of operators $\theta_t: \Omega \rightarrow \Omega$, such that \eqref{eq:Xshiftprop} holds. 
In this case the Markov property can be extended. Note that this is sometimes also used to define a Markov process (see Proposition III.1.7 in \citet{RY10} for the homogeneous case). Remember that $\F_\infty^X = \sigma(X_t, t \geq 0)$ is the $\sigma$-algebra generated by the process $X$.
\begin{lemma}
Let $(X,\F,\{\PM^{(r,x)}\})$ be a Markov process and 
assume that shift operators $\theta_t$ exist. Then for every bounded (or nonnegative) $\F_\infty^X$-measurable random variable $Z$ and $r,t \geq 0$, $x \in E$
\begin{equation} \label{eq:markov2inhomalternative}
\EV[(r,x)]{Z \circ \theta_t \vert \F_t} = \EV[(r+t,X_t)]{Z}, \quad \text{ a.s. on the set $\{X_t \neq \Delta \}$}.
\end{equation}
\end{lemma}
\begin{proof}
For the time-homogeneous case this is proved in \citet{RY10}, Proposition III.1.6 and III.1.7. The time-inhomogeneous case is similar. 
For the right hand side to be meaningful we need measurability of $x \mapsto \EV[(s,x)]{Z}$ for each $s \geq 0$. For $Z=\mathbb{I}_\Gamma$, where $\Gamma = \{X_0 \in A_0, X_{t_1} \in A_1, \dots, X_{t_n} \in A_n \}$ this follows from the Markov property and measurability of the transition function. The family of such sets $\Gamma$ is closed under finite intersections, an application of the monotone class theorem (\citet{RY10}, section 0.II) gives measurability for $\Gamma \in \F^X_\infty$. An approximation by simple functions extends measurability to $\F_\infty^X$-measurable $Z$. 

Now let $\Gamma = \{X_s \in A\}$. Then equation \eqref{eq:markov2inhomalternative} holds, i.e.
$$ \EV[(r,x)]{\mathbb{I}_\Gamma \circ \theta_t \vert \F_t} = \EV[(r,x)]{\1{X_{s+t} \in A} \vert \F_t} =
P_{r+t,r+t+s}(X_t,A) = \EV[(r+t,X_t)]{\Gamma}.$$
With the Markov property this also follows for $\Gamma = \{X_0 \in A_0, X_{t_1} \in A_1, \dots, X_{t_n} \in A_n \}$ and a monotone class argument as above extends this to sets $\Gamma \in \F_\infty^X$ and then to $\F_\infty^X$-measurable $Z$. 
\end{proof}
\begin{remark}
The measurability of $x \mapsto \EV[(r,x)]{Z}$ for bounded $\F_\infty^X$-measurable $Z$ allows us to define for each $r \geq 0$ and probability measure $\mu$ on $(E,\mathcal{E})$ a probability measure $\PM^{(r,\mu)}$ on $(\Omega,\F_\infty^X)$ (but not $(\Omega,\mathcal{A})$) by
$$\PM^{(r,\mu)}(\Lambda) = \int_E \PM^{(r,x)}(\Lambda) \mu(\dd{x}).$$
Denote by $\E^{(r,\mu)}$ the expectation under $\PM^{(s,\mu)}$. For each bounded $\F_\infty^X$-measurable $Z$
$$\EV[(r,\mu)]{Z} = \int_E \EV[{(r,x)}]{Z} \mu(\dd{x}).$$
Furthermore equation \eqref{eq:markov2inhomalternative} holds with $\mathbb{E}^{(r,x)}$ replaced by $\mathbb{E}^{(r,\mu)}$ and $\F_t$ replaced by $\F_t^X \subset \F_t$. 
\end{remark}

A time-inhomogeneous Markov process is time-homogeneous if the measures $\PM^{(r,x)}$ do not depend on $r$. The time-inhomogeneous case can be reduced to the time-homogeneous case by extending the state space and the probability space (see also \citet{BO14} or \citet{WE14}, section 8.5.5). Consider the following:
\begin{itemize}
\item The extended state space $\tilde E := \R_{\geq 0} \times E$ with the $\sigma$-algebra 
$$\mathcal{\tilde E} := \{ B \subset \tilde E:  B^s := \{x: (s,x) \in B \} \in \mathcal E \text{ for all } s \in \R_{\geq 0} \}.$$
\item The extended probability space $\tilde \Omega := \R_{\geq 0} \times \Omega$ with the $\sigma$-algebra 
$$ \mathcal{\tilde A} := \{A \subset \tilde \Omega: A^s := \{\omega: (s,\omega) \in A\} \in \mathcal{A}  \text{ for all } s \in 
\R_{\geq 0} \}.$$
\item The space-time process $\tilde X = (\Theta,X)$ on $(\tilde \Omega, \mathcal{\tilde A})$ defined by
$$\tilde X_t(s,\omega) := (s+t,X_{t}(\omega)).$$
\item The filtration $\mathcal{\tilde F}$ defined by $\mathcal{\tilde F}_t := \{ A \subset \mathcal{\tilde A}: A^s \in \F_t \text { for all } s \in \R_{\geq 0} \}.$
\item The probability measures $\tilde \PM^{(r,x)}$ on $(\tilde \Omega, \mathcal{\tilde A})$, which for $A \in \mathcal{\tilde A}$ are defined by
$$\tilde \PM^{(r,x)}(A) := \PM^{(r,x)}(A^r).$$
Note that the probability measure $\tilde \PM^{(r,x)}$ is concentrated on the set $\{r\} \times \Omega$, i.e. $\tilde \PM^{(r,x)}(\{r\} \times \Omega) = 1$. 
\item The transition function $\tilde P_t$ on $(\tilde E, \mathcal{\tilde E})$ defined by setting for $B \in \mathcal{\tilde E}$
$$\tilde P_t((s,x),B) := P_{s,s+t}(x,B^{s+t}) =  \PM^{(s,x)}(X_t \in B^{s+t}).$$
\item If there exist shift operators $\theta_t$ they are extend by $\tilde \theta_t(s,\omega) := (s+t,\theta_t(\omega))$. 
\end{itemize}
\begin{lemma} \label{lem:timehomtransformation}
Let $(X,\F,\{\PM^{(r,x)} \})$ be a time-inhomogeneous Markov process on $(\Omega,\mathcal{A})$ with state space $(E,\mathcal{E})$. 
\begin{itemize}
\item $(\tilde X,\mathcal{\tilde F},\{ \tilde \PM^{(r,x)} \})$ is a time-homogeneous Markov process on $(\tilde \Omega, \mathcal{\tilde A})$ with state space $(\tilde E,\mathcal{\tilde E})$ and transition function $\{ \tilde P_t \}$.
\item $(X,\mathcal{\tilde F},\{ \tilde \PM^{(r,x)} \})$ is a time-inhomogeneous Markov process on $(\tilde \Omega, \mathcal{\tilde A})$ with state space $(E,\mathcal{E})$ and transition function $\{ P_{s,t} \}$.
\end{itemize}
\end{lemma}
\begin{proof}
That $\{ \tilde P_t \}$ is a transition function can be easily checked. The crucial point is measurability of $\tilde P_t$. For fixed $B \in \mathcal{\tilde E}$ and any real Borel set $C$
$$\Lambda := \tilde P_t^{-1}(\cdot, B)(C) = \bigcup_{s \in \R_{\geq 0}} \{s\} \times P_{s,s+t}^{-1}(\cdot,B^{s+t})(C)$$
satisfies $\Lambda^s = P_{s,s+t}^{-1}(\cdot,B^{s+t})(C) \in \mathcal{E}$ and therefore $\Lambda \in \mathcal{\tilde E}$. 
Similar considerations show that the process $\tilde X$ is $\mathcal{\tilde F}$-adapted. 
To show the Markov property of $\tilde X$ one has to check that for $B \in \mathcal{\tilde E}$
$$\mathbb{\tilde E}^{(r,x)}[\mathbb I \{ \tilde X_{s+t} \in B \} \vert \mathcal{\tilde F}_s] =  \tilde \PM^{(r,x)}(\tilde X_{s+t} \in B \vert \mathcal{\tilde F}_s) = 
\tilde P_t(\tilde X_s, B).$$
$A \in \mathcal{\tilde F}_s$ implies $A^r \in \F_s$ and $\tilde \PM^{(r,x)}(A) =  \tilde \PM^{(r,x)}(\{r\} \times A^r) = \PM^{(r,x)}(A^r)$. Together with the Markov property of $X$ one obtains
\begin{align*}
\int_A \tilde P_t(\tilde X_s(\tilde \omega), B) \tilde \PM^{(r,x)}(\dd{\tilde \omega}) & = \int_{A^r} P_{r+s,r+s+t}(X_s(\omega),B^{r+s+t}) \PM^{(r,x)}(\dd{\omega}) \\
& = \int_{A^r} \PM^{(r,x)}(X_{s+t} \in B^{r+s+t} \vert \F_s) \PM^{(r,x)}(\dd{\omega}) \\
& =  \int_{A^r} \1{ X_{s+t}(\omega) \in B^{r+s+t}} \PM^{(r,x)}(\dd{\omega}) \\
& = \int_A  \1{\tilde X_{s+t}(\tilde \omega) \in B} \tilde \PM^{(r,x)}(\dd{\tilde \omega}).
\end{align*}
This shows that $(\tilde X,\mathcal{\tilde F},\{ \tilde \PM^{(r,x)} \})$ is a time-homogeneous Markov-process. It is easy to check that $(X,\mathcal{\tilde F},\{ \tilde \PM^{(r,x)} \})$ is a time-inhomogeneous Markov process with transition function $\{P_{s,t}\}$. 
\end{proof}
To get measurability in $(\tilde E, \mathcal{\tilde E})$ of the map $(s,x) \mapsto \tilde P_t((s,x),B)$ for $B \in \mathcal{\tilde E}$ we have to consider the very large $\sigma$-algebra $\mathcal{\tilde E}$ on $\tilde E$ (and then also $\mathcal{\tilde A}$ on $\tilde \Omega$). If $(s,x) \mapsto P_{s,s+t}(x,B)$ is $\mathcal{B}(\R_{\geq 0}) \otimes \mathcal{E}$-measurable for all $B \in \mathcal{B}(\R_{\geq 0}) \otimes \mathcal{E}$ the space-time process $\tilde X$ can also be realized with respect to the smaller product $\sigma$-algebras $\mathcal{\tilde E}^\prime := \mathcal{B}(\R_{\geq 0}) \otimes \mathcal{E}$ on $\tilde E$ and $\mathcal{\tilde A}^\prime := \mathcal{B}(\R_{\geq 0}) \otimes \mathcal{A}$ on $\tilde \Omega$ (note that $\mathcal{\tilde E}^\prime \subset \mathcal{\tilde E}$ and $\mathcal{\tilde A}^\prime \subset \mathcal{\tilde A}$).
In this case $\tilde \PM^{(r,x)} = \delta_r \otimes \PM^{(r,x)}$ and for $B = [a,b] \times B^\prime$, $B^\prime \in \mathcal{E}$, the definition of the transition function reduces to
\begin{equation} \label{eq:spacetimeproducttransitionfunction}
\tilde P_t((s,x),B)  = \tilde P_t((s,x),[a,b]\times B^\prime) = P_{s,s+t}(x,B^\prime) \mathbb{I}_{[a,b]}(s+t).
\end{equation}

If we consider a Markov process in the sense of Definition \ref{def:markov1} the transformation can be simplified. In this case the $\tilde E$-valued process $\bar X_t(\omega) = (t,X_t(\omega))$ is a time-homogeneous Markov process on the original filtered probability space. The remarks regarding the measurability of the transition function are still true. 
\begin{definition}
\label{def:stochcont}
A transition function $\{P_{s,t}\}$ is called stochastically continuous if for all $x \in E$
$$P_{s,S}(x,\cdot) \overset{d}{\rightarrow} P_{t,T}(x,\cdot), \text{ whenever } (s,S) \rightarrow (t,T) \text{ with }0 \leq s \leq S \text{ and }0 \leq t \leq T.$$
\end{definition}
For a Markov process in the sense of Definition \ref{def:markov2inhom} there is a whole family of probability measures. Hence stochastic continuity for Markov processes is formulated in terms of weak convergence of the transition function $\{P_{s,t}\}$. 
Note that this definition implies convergence in probability in the following sense. Consider the continuous functions $f^{\epsilon}(x) = \frac{ \vert x \vert}{\epsilon} \land 1$. Under the fixed probability measure $\PM$ of Definition \ref{def:markov1} or $\PM$ being equal to any of the measures $\PM^{(r,x)}$ in case of Definition \ref{def:markov2inhom},
\begin{align*}
\lim_{s \rightarrow t} \PM( & \vert X_t - X_s \vert > \epsilon) \leq  \lim_{s \rightarrow t} \EV{ \EV{f^\epsilon(X_t - X_s) \vert \F_s }} \\ & =  \EV{ \lim_{s \rightarrow t} \int f^\epsilon(y-X_s) P_{s,t}(X_s,\dd{y}) } =  \EV{ f^\epsilon(0) } = 0.  
\end{align*}
It follows that $X$ satisfies the usual definition of stochastic continuity, i.e. that $X_s$ converges to $X_t$ in probability (under $\PM$) for $s \rightarrow t$.  

In later sections we will use completed filtered probability spaces (see \citet{RY10}, section 0.3) and modifications of stochastic processes (see Definitions I.1.6 and I.1.7 in \citet{RY10}). We want to clarify how this is to be understood.
\begin{definition} \label{def:completion}
Let $(A, \mathcal{A})$ be a measurable space with a probability measure $\PM$. The completion $\mathcal{A}^{\PM}$ is defined as
$$\mathcal{A}^{\PM} := \{B \subset A: \exists B_1, B_2 \in \mathcal{A}: B_1 \subset B \subset B_2, \PM(B_2 \setminus B_1) = 0 \}.$$
The probability measures are extended for $B \in \mathcal{A}^{\PM}$ by setting $\PM(B) := \PM(B_2) = \PM(B_1)$.  
A sub-$\sigma$-algebra $\F$ on $(A, \mathcal{A})$ is completed by
$$\F^{\PM} := \{A:  \exists B \in \F: 
\PM( A \setminus B \cup B \setminus A ) = 0 \}.$$
\end{definition}
The completed $\sigma$-algebra $\mathcal{A}^{\PM}$ contains all subsets of sets of $\PM$-measure zero and  the same is true for $\F^{\PM}$. 
Completing $\sigma$-algebras (and adding null sets in general) does not destroy the Markov property (see Lemma \ref{lem:MarkovCompletion} below). 

Let $X$ and $\tilde X$ be two $E$-valued stochastic process on $(\Omega,\mathcal{A},\PM)$. $X$ is a modification of $\tilde X$ if $\PM(X_t = \tilde X_t) = 1$ for all $t \geq 0$. If $(X,\F,\{\PM^{(s,x)}\})$ and $(\tilde X,\F,\{ \PM^{(s,x)}\})$ are Markov processes in the sense of Definition \ref{def:markov2inhom} they are modifications of each other, if $\PM^{(s,x)}(X_t = \tilde X_t) = 1$ for all $(s,x)$.

We are now prepared to formulate the following easy lemma. 
\begin{lemma} \label{lem:MarkovCompletion}
If $X$ is Markov process on $(\Omega,\mathcal{A}, \F ,\PM)$, it is also a Markov process on the completed space 
$(\Omega,\mathcal{A}^{\PM},\F^{\PM},\PM)$ with $\F^{\PM} = (\F_t^{\PM})_{t \geq 0}$. Also any modification $\tilde{X}$ of $X$ satisfying that $X_s = \Delta$ implies $X_t = \Delta$ for all $t > s$ is a Markov process on $(\Omega,\mathcal{A}^{\PM},\F^{\PM},\PM)$. Furthermore the transition functions are the same.
\end{lemma}
\begin{proof}
$X$ is still $\F^{\PM}$-adapted and conditional expectations are not influenced by adding sets of measure $0$. Since the filtration contains all sets of measure $0$, a modification is still adapted. Since $\PM(\{\tilde{X}_t \neq {X}_t \}) = 0$ it is still a Markov process. 
\end{proof}
For the rest of this thesis we write $(X,\F,\{ \PM^{(s,x)} \})$ if we explicitly consider a Markov process in the sense of Definition \ref{def:markov2inhom} while a Markov process $X$ can refer to either of the two definitions. If a Markov process $X$ is to be understood in the sense of Definition \ref{def:markov1}, we explicitly specify the underlying probability space $(\Omega, \mathcal{A},\F,\PM)$. 

\section{Affine processes} \label{sec:affineproctheory}
From now on consider a state space $E_{\Delta}$, where $E$ is a closed subset of the $d$-dimensional real vector space $V$ (think of $\R^d$, but it can also be a different space, like the space of positive definite matrices on $\R^{n \times n}$, which is isomorphic to $\R^{n(n+1)/2}$). 
\begin{assumption} \label{ass:affinespan}
$E$ contains an affine basis, i.e. $d+1$ elements $x^0,\dots,x^{d}$, such that
$x^1-x^0, \dots x^d-x^0$ are linearly independent. 
 
\end{assumption}
For $E$ consider the Borel $\sigma$-algebra $\mathcal{E} = \mathcal{B}(E)$. 
We use the convention that $\Vert \Delta \Vert = \infty$. 
Define
\begin{equation}
\mathcal{U} = \{ u \in V + \i V  \big \vert \  x \mapsto \e^{\scalarprod{u,x}} \text{ is a bounded function in }E\}
\end{equation}
and for $u \in \mathcal{U}$ the bounded functions $f_u(x) := \e^{\scalarprod{u,x}}$ for $x \in E$, $f_u(\Delta) := 0$. 
\begin{definition} \label{def:affineprocess}
A time-inhomogeneous transition function $\{ P_{s,t} \}$ is called affine if there exist functions $\Phi_{s,t}: \mathcal{U}  \rightarrow \C$ and $\psi_{s,t}: \mathcal{U} \rightarrow V+\i V$ for $0 \leq s \leq t$ such that for all $ u \in \mathcal{U}$, $x \in E$
\begin{equation} P_{s,t} f_u(x) =  \Phi_{s,t}(u) \e^{\scalarprod{\psi_{s,t}(u),x}}.
\label{eq:affineprocgeneral} \end{equation}
A time-inhomogeneous Markov process is affine if its transition function is affine. If the transition function is time-homogeneous (i.e. $P_{s,t} = P_{t-s}$), then
\begin{equation} P_{t} f_u(x) =  \Phi_{t}(u) \e^{\scalarprod{\psi_{t}(u),x}}
\label{eq:affineprocgeneralhom} \end{equation}
with $\Phi_t(u) := \Phi_{0,t}(u)$ and $\psi_t(u) := \psi_{0,t}(u)$.
\end{definition}
The affine property of a process is a property of the transition operator. Hence it does not depend on the used definition of a Markov process (Definition \ref{def:markov1} versus Definition \ref{def:markov2inhom}). 
The restrictions on the state space are not really restrictions. The affine property automatically extends to the closure of a set and one can always pass to a lower dimensional ambient vector space $V$, so that Assumption \ref{ass:affinespan} is fulfilled (see \citet{CT13}). 

As long as $\Phi_{s,t}(u) \neq 0$ we can find a function $\phi$ (which is unique only up to adding multiples of $2 \pi \i$), such that $\Phi_{s,t}(u) = \e^{\phi_{s,t}(u)}$. In this case equation \eqref{eq:affineprocgeneral} can be rewritten as
\begin{equation}
P_{s,t} f_u(x) = \e^{\phi_{s,t}(u) + \scalarprod{\psi_{s,t}(u),x}}.
\end{equation}
This exponential affine form is one motivation for the name affine process. Note that if $(X,\F, \{\PM^{(s,x)}\})$ is an affine process, then for $r,s,t \geq 0$
$$\EV[(r,x)]{\e^{\scalarprod{u,X_{t+s}}} \vert \F_s} = \Phi_{r+s,r+s+t}(u)  \e^{ \scalarprod{\psi_{r+s,r+s+t}(u),X_s}},$$
which for $s=0$ yields
$$\EV[(r,x)]{\e^{\scalarprod{u,X_{t}}}} = \Phi_{r,r+t}(u) \e^{ \scalarprod{\psi_{r,r+t}(u),x}}.$$
The exponential functions in the above equations are meaningfully defined with the convention $f(\Delta) = 0$, i.e. $\e^{\scalarprod{u,\Delta}} = 0$.
Equations like this will be used frequently throughout this chapter. 

Time-homogeneous affine process have been widely studied in the literature, see for example the pioneering work of \citet{DFS03} and the articles \citet{CT13, KST11b}, which are similar in spirit to this treatment. By Lemma \ref{lem:timehomtransformation} a time-inhomogeneous Markov process can be extended to the time-homogeneous space-time process. The affine property carries over to the space-time process only partially. 
\begin{lemma} \label{lem:affinespacetime}
Let $X$ be a Markov process with affine transition function. Then the time-homogeneous transition function $\{ \tilde P_t \}$ of the space-time process $\tilde{X} := (\Theta,X)$ from Lemma \ref{lem:timehomtransformation} satisfies for 
$u_0 \in \C_{\leq 0}$, $u \in \mathcal{U}$ and $f_{(u^0,u)}(s,x) := \e^{u^0 s} f_u(x) $
\begin{equation} \label{eq:affinehomtransform}
\tilde P_t f_{(u^0,u)}(s,x) = \e^{u^0 t} \Phi_{s,s+t}(u) \e^{u^0 s + \scalarprod{\psi_{s,s+t}(u),x}}.
\end{equation}
\end{lemma}
\begin{proof} 
Follows directly from the definition of the time-homogeneous transition function $\tilde P_t$ in Lemma \ref{lem:timehomtransformation}, i.e. $\tilde P_t f(s,x) := P_{s,s+t} f^{s+t}(x)$ with $f^{s+t}(x) = f(s+t,x)$. 
\end{proof}
\begin{remark}
$\tilde X$ is almost an affine process. The transition function $\{ \tilde P_t \}$ is exponentially affine in $x$, but not in the (basically deterministic) time component $s$, since $\Phi_{s,s+t}$ and $\psi_{s,s+t}$ depend on the state $(s,x)$ if $\{ P_{s,t} \}$ is not time-homogeneous. 

\end{remark}

So far the functions $\Phi$ and $\psi$ in Definition \ref{def:affineprocess} are not necessarily uniquely defined and can be arbitrary irregular. If we assume that the transition function is stochastically continuous, the following lemma yields continuity of $(s,t,u) \mapsto P_{s,t} f_u(x)$, which we later want to transfer to $\Phi$ and $\psi$. Note that stochastic continuity is sometimes part of the definition of an affine process (see e.g. \citet{CT13}). 
For an affine process consider 
\begin{equation} \label{eq:o(T,u)}
\begin{aligned}
\sigma(s,u) & := \inf\{t > s: \Phi_{s,t}(u) = 0\}, \\
o(T,u) & := \sup\{t < T: \Phi_{t,T}(u) = 0\} \lor 0,
\end{aligned}
\end{equation}
where we use the conventions $\inf \{ \emptyset \} = \infty$ and $\sup \{ \emptyset \} = - \infty$. 
$\sigma$ is used for going forward in time, while $o$ is needed for going backward in time. 
For $k \in \mathbb{N}$ let 
\begin{align*}
\mathcal{U}_k := & \{ u \in V + \i V: \sup_{x \in E} \vert \e^{\scalarprod{u,x}} \vert \leq k \}, \\
\mathcal{Q}_k  := & \{ (t,T,u): u \in \mathcal{U}_k, \; o(T,u) < t \leq T \} 
\\ = & \{(t,T,u): 0 \leq t \leq T, u \in \mathcal{U}_k, \Phi_{v,T}(u) \neq 0 \text{ for } t \leq v \leq T \}, 
\end{align*}
and $\mathcal{U} = \cup_k \, \mathcal{U}_k$ and $\mathcal{Q} := \cup_k \mathcal{Q}_k$.

\begin{lemma} \label{lem:affineproperties1}
Let $\{ P_{s,t} \}$ be stochastically continuous. Then
\begin{enumerate}[(i) ]
\item
$\Phi$ vanishes on $\{(t,T,u), 0 \leq t \leq T, u \in \mathcal{U} \} \setminus \mathcal{Q}$, so $$\mathcal{Q} = \{(t,T,u): 0 \leq t \leq T, u \in \mathcal{U}, \Phi_{t,T}(u) \neq 0 \}.$$
\item The function $\psi$ maps $\mathcal{Q}$ into $\mathcal{U}$.
\item For $x \in E$ the functions $f^x(t,T,u) := P_{t,T} f_u (x)$ are continuous on $$\{(t,T,u), 0 \leq t \leq T, u \in \mathcal{U}_k \}$$ for all $k \in \mathbb{N}$. If $0 \in E$, this is also true for $\Phi$.
\item $\mathcal{Q}$ is open and $\sigma(t,u) > t$, $o(T,u)<T$ for any $u \in \mathcal{U}$.
\end{enumerate}
\end{lemma}
\begin{remark}
Note that
\begin{equation} \label{eq:phi0}
\exists x \in E: f^x(t,T,u) = 0 \Leftrightarrow \Phi_{t,T}(u) = 0 \Leftrightarrow \forall x \in E_\Delta: f^x(t,T,u) = 0.
\end{equation}
\end{remark}
\begin{proof}
If $\Phi_{t,T}(u) = 0$ we get that for $s \leq t \leq T$, $x \in E$
\begin{align*} 
f^x(s,T,u)& = P_{s,T} f_u(x) = P_{s,t} P_{t,T} f_u(x) = 0
\end{align*}
so also $\Phi_{s,T}(u) = 0$. Then $\Phi_{s,T}(u) = 0$ for all $s < o(T,u)$. By stochastic continuity also $\Phi_{o(T,u),T}(u) = 0$ and (i) follows. 

To get (ii) note that for $u \in \mathcal{U}_k$ and $x \in E$
$$\vert \Phi_{t,T}(u) \e^{\scalarprod{\psi_{t,T}(u),x}} \vert = \vert P_{t,T} f_u (x) \vert \leq \int_E \vert f_u(\xi) \vert P_{t,T}(x,\dd{\xi}) \leq \int_E K P_{t,T}(x,\dd{\xi}) = K.$$
If $(t,T,u) \in \mathcal{Q}$ we can divide by $\vert \Phi_{t,T}(u) \vert$ and $\vert \e^{\scalarprod{\psi_{t,T}(u),x}} \vert \leq K{\vert \Phi_{t,T}(u) \vert}^{-1}$ for all $x \in E$. Hence $\psi_{t,T}(u) \in \mathcal{U}$. 

Let $x \in E$ and let $(t_n,T_n,u_n) \rightarrow (t,T,u)$ be a converging sequence, where $u_n, u \in \mathcal{U}_k$ and $0 \leq t_n \leq T_n$, $0 \leq t \leq T$. 
There exists a function $\rho: E \rightarrow [0,1]$ with compact support such that $$P_{t,T} (1-\rho)(x) < \epsilon.$$ Choose $N_1$ such that $$\vert f_{u_n}(\xi) - f_u(\xi) \vert = \vert \e^{\scalarprod{u_n,\xi}} -  \e^{\scalarprod{u,\xi}} \vert < \epsilon, \qquad \forall n \geq N_1, \xi \in \mathrm{supp}(\rho). $$
By stochastic continuity there is $N_2$ such that for all $n \geq N_2$
\begin{align*}
& P_{t_n,T_n}(1-\rho)(x)  < 2 \epsilon, \\
& \big \vert P_{t_n,T_n} f_u(x) - P_{t,T} f_u(x) \big \vert  < \epsilon .
\end{align*}
Then for all $n \geq \max\{N_1,N_2\}$
\begin{align*}
\big \vert P_{t_n,T_n} f_{u_n}(x) - P_{t,T} f_u(x) \big \vert  & \leq  \big \vert P_{t_n,T_n} (\rho (f_{u_n} - f_u)) (x) \big \vert  \\
 &\  + \big \vert P_{t_n,T_n} ((1-\rho) (f_{u_n} - f_u)) (x) \big \vert \\
 &\  +  \big \vert P_{t_n,T_n} f_u(x)  - P_{t,T} f_u(x) \big \vert \\ 
 & < \epsilon + 2 k \epsilon + \epsilon = 2 (k+1) \epsilon .
\end{align*}
Hence $f^x(t,T,u)$ is continuous on $\{(t,T,u): 0 \leq t \leq T, u \in \mathcal{U}_k\}$ for each $k \in \mathbb{N}$ and $x \in E$. If $0 \in E$, $\Phi_{t,T}(u) = f^0(t,T,u)$ is continuous and we get (iii). 
Since $\Phi_{t,t}(u) \neq 0$ the continuity of $f^x$ and \eqref{eq:phi0} implies (iv).
\end{proof}
We want to use the continuity in Lemma \ref{lem:affineproperties1} (ii) to get unique choices of to $\Phi$ and $\psi$. This is possible because of Assumption \ref{ass:affinespan}. 

\begin{lemma}  \label{lem:affineproperties}
Let $\{P_{s,t}\}$ be stochastically continuous. Then $\psi$ is uniquely defined on $\mathcal{Q}$ by requiring that $\psi_{t,T}(u)$ is continuous on $\mathcal{Q}_k$ for all $k \in \mathbb{N}$ and $\psi_{0,0}(0) = 0$. In this case also $\Phi$ is uniquely defined and continuous on $\mathcal{Q}_k$ 
 for all $k \in \mathbb{N}$. Furthermore, there is a unique function $\phi_{t,T}(u)$ on $\mathcal{Q}$, which is continuous on $\mathcal{Q}_k$ and satisfies $\phi_{0,0}(0) = 0$ and $\e^{\phi_{t,T}(u)} = \Phi_{t,T}(u)$. 
\end{lemma}
For the proof we borrow heavily from \citet{KST11b}.
\begin{proof}
Define the sets $K_n := \{(t,T,u): u \in \mathcal{U}_n, \Vert u \Vert \leq n, t \in [o(T,u) + \frac{1}{n},T] \}$, so that $K_n \subset \mathcal{Q}_n$ and $\lim_{n \rightarrow \infty} K_n = \mathcal{Q}$. We claim that every set $K_n$ is contractible to zero. 
Let $\gamma = (t(r),T(r),u(r))_{0 \leq r \leq 1}$ be a continuous curve in $K_n$. Then for $0 \leq \alpha \leq 1$ the curves $\gamma_\alpha = (t(r)+ (1 - \alpha) (T(r)-t(r)), T(r),u(r))_{0 \leq r \leq 1}$ depend continuously on $\alpha$ and stay in $K_n$ for each $\alpha$. Furthermore, $\gamma_1 = \gamma$ and $\gamma_0 = (T(r),T(r),u(r))_{0 \leq r \leq 1}$. So any continuous curve in $K_n$ is homotopically equivalent to a continuous curve in the convex set $\R_{\geq 0} \times \mathcal{U}$, where all continuous curves are contractible to zero. This proves that $K_n$ is contractible to zero. Let $H_n: [0,1] \times K_n \rightarrow K_n$ be such a continuous contraction and fix $x \in E$. Since $K_n$ is compact, we have 
$\lim_{\alpha \rightarrow \beta} \Vert f^x(H_n(\alpha, \cdot)) - f^x(H_n(\beta, \cdot)) \Vert_{\infty} = 0.$ Hence $f^x  \circ H_n$ is a continuous curve in $C_b(K_n)$ from $f^x\vert_{K_n}$ to the constant function $1$. Here $C_b(K_n)$ denotes the Banach space of bounded continuous functions on $K_n$. By Theorem 1.3 in \citet{DB91} there exists a continuous logarithm $g_n^x \in C_b(K_n)$, i.e. $f^x(t,T,u) = \e^{g^x_n(t,T,u)}$ for all $(t,T,u) \in K_n$.
We next show, that $g_n^x$ is uniquely defined by setting $g_n^x(0,0,0) = 0$. Assume that $\tilde g^x_n(t,T,u)$ is another such continuous logarithm, then $g^x_n(t,T,u) = \tilde g^x_n(t,T,u) + 2 \pi \i \alpha(t,T,u)$ with $\alpha(t,T,u) \in \mathbb{Z}$. Then $\alpha(t,T,u)$ is continuous on $K_n$, which is contractible to $0$ and hence also connected. Hence $\alpha(t,T,u)$ cannot depend on $(t,T,u)$. Inserting $(0,0,0)$ gives $\alpha(t,T,u) = 0$ and the uniqueness of $g_n^x$. 

For arbitrary $m \leq n$ we have 
$$g_m^x(t,T,u) = g_n^x(t,T,u) + 2 \pi \i \alpha^x(t,T,u) \qquad \text{ for all } (t,T,u) \in K_m,$$
where $\alpha^x(t,T,u)$ is a continuous function from $K_m$ to $\mathbb{Z}$ satisfying $\alpha^x(0,0,0)=0$. Hence $g_m^x(t,T,u) = g_n^x(t,T,u)$ for all $(t,T,u) \in K_m$. This shows that $g_n^x$ extends $g_m^x$, hence
 there is a unique function $g^x : \mathcal{Q} \rightarrow \C$ such that $g^x(0,0,0) = 0$ and \begin{equation}\e^{g^x(t,T,u)} = f^x(t,T,u) =  \Phi_{t,T}(u) \e^{\scalarprod{\psi_{t,T}(u),x}} . \label{eq:gx} \end{equation}
Note that like $f^x$, $g^x$ is then continuous on each $\mathcal{Q}^k$.

Fix a point $x^0$ of an affine basis $x^0, \dots, x^d$. Then on $\mathcal{Q}_k$
$$e^{\scalarprod{\psi_{t,T}(u) ,x - x^0}} = \e^{g^x(t,T,u) - g^{x^0}(t,T,u)},$$
so $\scalarprod{\psi_{t,T}(u)  , x - x^0} =  g^x(t,T,u) - g^{x^0}(t,T,u) + 2 \pi \i \alpha^x(t,T,u),$
where $\alpha^x(t,T,u) \in \mathbb{Z}$. $\alpha^x(0,0,0) = 0$. 
Continuity of $\psi$, $g^x$ and $g^{x^0}$ on $\mathcal{Q}_k$ gives continuity of $\alpha^x(t,T,u)$. Hence $\alpha^x=0$ and $\scalarprod{\psi_{t,T}(u) ,x - x^0}$ is uniquely defined on $\mathcal{Q}$ through the unique functions $g^x, x \in E$.
By Assumption \ref{ass:affinespan} this is also true for $\psi$. On
$\mathcal{Q}$ set $\phi_{t,T}(u) = g^x(t,T,u) - \scalarprod{\psi_{t,T}(u) ,x}$ which is equal for all $x \in E$ and then also uniquely defined. Then $\Phi_{t,T}(u) = \e^{\phi_{t,T}(u)}$ and $\Phi$ is also uniquely defined and continuous on $\mathcal{Q}_k$. 
\end{proof}
From now when talking about continuous functions $\Phi$ and $\psi$, we always mean that they are of the form described in Lemma \ref{lem:affineproperties}. It suffices to require weaker properties of $\psi$ to be sure that it is in fact a {continuous} version. These properties need to be sufficient to guarantee that $\alpha^x$ in the proof of Lemma \ref{lem:affineproperties} does not depend on $(t,T,u)$. E.g. one could relax this by requiring that for each $(t_0,T_0,u_0)$ and $x \in E$ there is some open neighborhood where $\vert  \scalarprod{\Im \psi_{t,T}(u),x} - \scalarprod{\Im \psi_{t_0,T_0}(u_0),x}\vert < \pi$ as suggested in \citet{CT13}. 
Note that as in the time-homogeneous case (see \citet{CT13, KST11b}) we do not get the existence of {continuous} versions, but only uniqueness.  
Under a strengthening of the assumptions on the state space, we get that the function $\psi$ is always continuous.
\begin{lemma} \label{lem:psicontexistence}
Assume that $E$ has a connected component, which contains an affine basis. If $\{ P_{s,t} \}$ is stochastically continuous, then any function $\psi$ satisfying \eqref{eq:affineprocgeneral} is continuous on $\mathcal{Q}_k$ with $\psi_{0,0}(0) = 0$. 
\end{lemma}

\begin{proof}
Let $x^0, \dots, x^d$ be an affine basis in a connected component. Define the function $$h(t,T,u,x) := \e^{\scalarprod{\psi_{t,T}(u),x-x_0}}, \qquad x \in V.$$
For $x \in E$ this can be written as $h(t,T,u,x) =  \frac{f^x(t,T,u)}{f^{x^0}(t,T,u)}$. By Lemma \ref{lem:affineproperties1} iii) $h$ is continuous in $(t,T,u) \in \mathcal{Q}_k$ for each $x \in E$. Furthermore it is continuous in $x$ for each $(t,T,u)$ and even jointly continuous on $\mathcal{Q}_k \times E$. Set $E_n = \{x \in E, \Vert x \Vert \leq n \}$. With $H_n$ as in the proof of Lemma \ref{lem:affineproperties}, $\alpha \rightarrow h(H_n(\alpha,(t,T,u)),x)$ is a continuous curve in $C_b(K_n \times E_n)$ from $h \vert_{K_n \times E_n}$ to the constant function $1$. By Theorem 1.3 in \citet{DB91} there is a continuous function $g^n(t,T,u,x)$ on $K_n \times E_n$, such that $h(t,T,u,x) = \e^{g^n(t,T,u,x)}$. Setting $g(0,0,0,x) = 0$ for all $x \in E_n$ uniquely defines $g^n$. As in the proof of Lemma \ref{lem:affineproperties} we can extend this to get a unique function $g(t,T,u,x)$ that is continuous on $\mathcal{Q}_k \times E$, satisfies $g(0,0,0,x) = 0$ for all $x \in E$ and $h(t,T,u,x) = \e^{g(t,T,u,x)}$. 

Then $\scalarprod{\psi_{t,T}(u),x-x^0} = g(t,T,u,x) - 2 \pi \i \alpha(t,T,u,x)$ with $\alpha(t,T,u,x) \in \mathbb{Z}$. 
Since the left hand side is continuous in $x$ also $\alpha$ is continuous in $x$. Thus $\alpha$ is constant in $x$, on every connected component. 
Denote by $\bar \alpha (t,T,u)$ the value on the connected component containing $x^0, \dots, x^d$. Then
$$\scalarprod{\psi_{t,T}(u),x^i - x^0} = g(t,T,u,x^i) - 2 \pi \i  \bar\alpha(t,T,u), \quad 0 \leq i \leq d.$$
The case $i=0$ proves that $2 \pi \i \bar \alpha(t,T,u) = g(t,T,u,x^0)$ is continuous on $\mathcal{Q}_k$. Then also $\scalarprod{\psi_{t,T}(u),x^i - x^0}$ is continuous on $\mathcal{Q}_k$. Since $\{x^0, \dots, x^d\}$ is an affine basis, $\psi_{t,T}(u)$ is continuous on $\mathcal{Q}_k$. Finally, $\scalarprod{\psi_{0,0}(0),x^i-x^0} =g(0,0,0,x^i) - 2 \pi \i \bar \alpha(0,0,0) = 0$ for $1 \leq i \leq d$ implies $\psi_{0,0}(0) = 0$.
\end{proof}

This shows that for almost every state space stochastic continuity implies continuity of $\Phi$ and $\psi$. Next we show that essentially also the reverse is true. We first state the following well-known fact. 
\begin{lemma} \label{lem:borelexponentialgenerated}
$\mathcal{B}(\R^d) = \sigma(C_{exp}) := \sigma(\{ f^{-1}(A): A \in \mathcal{B}(\C), f \in C_{exp} \})$, where
$$C_{exp} = \left \{\sum_{k=1}^n a_k  f_{u_k}(x), a_k \in \C, u_k \in \i \R^d, n \in \mathbb{N} \right \}.$$
\end{lemma}
\begin{proof}[Sketch of proof]
To simplify notation assume $d = 1$. Take some interval $[a,b]$ and a modified indicator function of this interval, which is $0$ outside of $[a,b]$, $1$ on $(a,b)$ and $1/2$ for $\{a,b\}$. This function can be pointwise approximated by $C^\infty$-functions. By Fourier methods such an approximating $C^\infty$-function $g$ can be written as $g(x) = \frac{1}{2 \pi} \int_{\i \R} \hat g(z) \e^{z x} \dd{z}$. Cutting off the integral and approximating it with Riemann sums gives that it can be pointwise approximated by functions in $C_{exp}$, which then proves the result. 
\end{proof}
\begin{lemma} \label{lem:provestochcont}
Let $(t,T,u) \mapsto P_{t,T} f_u(x)$ be continuous on $\{(t,T,u): 0 \leq t \leq T, u \in \i V \subset \mathcal{U}_1 \}$ for each $x \in E$. Then $\{ P_{t,T} \}$ is stochastically continuous. 
\end{lemma}
\begin{proof}
$(t,T) \mapsto P_{t,T} f_u(x)$ is continuous for all $x \in E$ and $u \in \mathcal{U}_1$.  
Hence this also holds for $f \in C_{exp}$. Let $C_b(E_\Delta)$ denote the bounded measurable functions $f: E_\Delta \rightarrow \C$ and define the vector space
$$\mathcal{H} := \{f \in C_b(E_\Delta): (t,T) \mapsto P_{t,T}f(x) \text{ is continuous for all } x \in E \},$$ 
which is a monotone class containing the constant function. $C_{exp} \subset \mathcal{H}$ is closed under multiplication.
An application of the functional monotone class theorem (\citet{RY10}, Theorem 0.2.2) together with Lemma \ref{lem:borelexponentialgenerated} gives that $\mathcal{H}$ contains all bounded measurable functions, which gives stochastic continuity of $\{P_{t,T}\}$. 
\end{proof}
\begin{corollary}
Let $\{ P_{s,t} \}$ be an affine transition function on a state space $E_\Delta$ which contains $0$ and has a connected component containing an affine basis. 
Then the following are equivalent. 
\begin{itemize}
\item $\{ P_{s,t} \}$  is stochastically continuous. 
\item $\psi$ is continuous on $\mathcal{Q}_k$ with $\psi_{0,0}(0) = 0$ and $\Phi$ and $(t,T,u) \mapsto P_{t,T} f_u(x)$ are continuous on $\{(t,T,u): 0 \leq t \leq T, u \in \mathcal{U}_k \}$ for all $x \in E$ and $k \in \mathbb{N}$. 
\end{itemize}
\end{corollary}
\begin{proof}
Lemma \ref{lem:provestochcont} gives one direction. Lemma \ref{lem:psicontexistence} yields the continuity of $\psi$ on $\mathcal{Q}_k$ and $\psi_{0,0}(0) = 0$.  Lemma \ref{lem:affineproperties1} iii) with the assumption $0 \in E$ the rest. 
\end{proof}
This shows that stochastic continuity and continuity of $\Phi$ and $\psi$ are mostly equivalent. Instead of stochastic continuity of the transition function we consider the following definition. 
\begin{definition}
An affine transition function $\{ P_{t,T} \}$ 
is called a continuously affine transition function if $\psi$ and $\phi$ are continuous on $\mathcal{Q}_k$  with $\psi_{0,0}(0) = 0$. $\Phi$ and $(t,T,u) \mapsto P_{t,T} f_u(x)$ are continuous on $\{(t,T,u): 0 \leq t \leq T, u \in \mathcal{U}_k \}$ for all $x \in E$ and $k \in \mathbb{N}$. 
A Markov process is continuously affine, if its transition function is continuously affine.
\end{definition}
\begin{lemma}  \label{lem:affineproperties3}
Let $\{ P_{t,T} \}$ be a continuously affine transition function. 
\begin{enumerate}[(i) ]
\item $\Phi$, $\phi$ and $\psi$ satisfy the following semi-flow equations. For $s \leq t \leq T$
\begin{equation} \label{eq:semiflow}
\begin{aligned}
\Phi_{s,T}(u)   & =   \Phi_{s,t}(\psi_{t,T}(u))  \Phi_{t,T}(u) , \qquad && \text{for all } u \in \mathcal{U}, \\
\psi_{s,T}(u) & = \psi_{s,t}(\psi_{t,T}(u)), &&  \text{if } \Phi_{s,T}(u) \neq 0, \\
\phi_{s,T}(u) & = \phi_{t,T}(u) + \phi_{s,t}(\psi_{t,T}(u)) \qquad &&  \text{if } \Phi_{s,T}(u) \neq 0.
\end{aligned}
\end{equation}
\item Let $u \in \mathcal{U} \cap V$. Then $\Im \Phi_{t,T}(u) = 0$ for $0 \leq t \leq T$ and $\Im \psi_{t,T}(u) = 0$, $\Im \phi_{t,T}(u) = 0$  for $o(T,u) < t \leq T$.
\item For all $t \geq 0$ it holds that $\Phi_{t,t}(u) = 1$, $\phi_{t,t}(u) = 0$  and $\psi_{t,t}(u) = u$. 

\end{enumerate}
\end{lemma}

\begin{proof}
Regarding (i), if $\Phi_{s,T}(u) \neq 0$, then by Lemma \ref{lem:affineproperties1} i) also $\Phi_{t,T}(u) \neq 0$, so by Lemma \ref{lem:affineproperties1} ii) $\psi_{t,T}(u) \in \mathcal{U}$ and 
\begin{equation} \label{eq:semiflowderiv}
\begin{aligned} 
\Phi_{s,T}(u)  \e^{\scalarprod{\psi_{s,T}(u),x}} & = P_{s,T} f_u(x) = P_{s,t} P_{t,T} f_u(x) =\Phi_{t,T}(u) P_{s,t} f_{\psi_{t,T}(u)}(x)
\\
& = \Phi_{t,T}(u)   \Phi_{s,t}(\psi_{t,T}(u))   \e^{\scalarprod{\psi_{s,t}(\psi_{t,T}(u)),x}}.
\end{aligned}
\end{equation}
Hence also $\Phi_{s,t}(\psi_{t,T}(u)) \neq 0$ and $(s,t,\psi_{t,T}(u)) \in \mathcal{Q}$. Thus we can write \eqref{eq:semiflowderiv} as 
$$\e^{\phi_{s,T}(u)  + \scalarprod{\psi_{s,T}(u),x}} = \e^{\scalarprod{  \phi_{t,T}(u) +   \phi_{s,t}(\psi_{t,T}(u))   + \psi_{s,t}(\psi_{t,T}(u)),x}}.$$
The previously discussed uniqueness of the logarithm together with Assumption \ref{ass:affinespan} gives the semiflow equations. 
If $\Phi_{s,t}(u) = 0$, we consider two cases. If $\Phi_{t,T}(u)=0$ the semiflow equation for $\Phi$ holds. Otherwise by Lemma \ref{lem:affineproperties1} ii) $\psi_{t,T}(u) \in \mathcal{U}$ and \eqref{eq:semiflowderiv} holds. Since $\Phi_{t,T}(u)  \e^{\scalarprod{\psi_{s,t}(\psi_{t,T}(u)),x}} \neq 0$, it follows that $\Phi_{s,t}(\psi_{t,T}(u)) = 0$. This proves the semiflow equation for $\Phi$ in the case $\Phi_{s,t}(u) = 0$.

For (ii) note that for $u \in \mathcal{U}_k \cap V$, $o(T,u) < t \leq T$
$$\e^{\scalarprod{\psi_{t,T}(u),x - x^0}} = \frac{P_{t,T} f_u(x)}{P_{t,T} f_u(x^0)} \in \R_{\geq 0}.$$
Hence $\scalarprod{\Im \psi_{t,T}(u),x - x^0} = 2 \pi \i \alpha^x(t,T,u)$ with $\alpha^x(t,T,u) \in \mathbb{Z}$. As in the proof of Lemma \ref{lem:affineproperties} it follows that $\scalarprod{\Im \psi_{t,T}(u),x - x^0} = 0$. Hence also $\Im \psi_{t,T}(u) = 0$. 
This implies $\Im \phi_{t,T}(u) = 0$ and $\Im \Phi_{t,T}(u) = 0$. Note that $\Phi = 0$ outside of $\mathcal{Q}$ by \ref{lem:affineproperties1} i).

(iii) follows by the same arguments as in ii) since $P_{t,t} f_u(x) = f_u(x)$ and therefore
$\e^{\scalarprod{\psi_{t,t}(u)-u,x - x^0}} = 1.$
\end{proof}

The continuity of $\Phi$ and $\psi$ can also be used to infer the joint measurability of the transition function $\{ P_{s,t} \}$. 
\begin{lemma} \label{lem:affinemeasurability}
For a continuously affine transition function $(s,t,x) \mapsto P_{s,s+t}(x,B)$ is $\mathcal{B}(\R^2_{\geq 0}) \otimes \mathcal{E}_\Delta$-measurable for every $B \in  \mathcal{E}_\Delta$. 
\end{lemma}
\begin{proof}
Let $C \subset \C$ and $u \in \i V$. By continuity of $\Phi$ the set $\mathcal{Q}_u = \{ (s,t) \in \R^2_{\geq 0}: \Phi_{s,s+t}(u) \neq 0 \}$ is open and $\mathcal{B}(\R^2_{\geq 0})$-measurable. $(s,t) \mapsto \psi_{s,t}(u)$ is a continuous function on $\mathcal{Q}_u$. Hence also $(s,t,x) \mapsto P_{s,s+t}f_u(x) = \Phi_{s,s+t}(u) \e^{\scalarprod{\psi_{s,s+t}(u),x}}$ is continuous on $\mathcal{Q}_u \times E$ and $D = \{(s,t,x) \in \R^2_{\geq 0} \times E: P_{s,s+t}f_u(x) \in C \setminus \{ 0 \} \} \in \mathcal{B}(\R^2_{\geq 0}) \otimes \mathcal{E}_\Delta$. Therefore
\begin{align*} \{(s,t,x) \in & \R^2_{\geq 0} \times \mathcal{E}_\Delta:  P_{s,s+t} f_u(x) \in C\} \\ & = \begin{cases}  0 \notin C: & D  \\
0 \in C: & (\R^2_{\geq 0} \setminus \mathcal{Q}_u \times E) \cup (\R^2_{\geq 0} \times \{\Delta \}) \cup D, \end{cases}
\end{align*}
is in $\mathcal{B}(\R^2_{\geq 0}) \otimes \mathcal{E}_\Delta$. Hence
the function $(s,t,x) \mapsto P_{s,s+t} f_u(x)$ is $\mathcal{B}(\R^2_{\geq 0}) \otimes  \mathcal{E}_\Delta $-measurable for every $f_u$. This extends to functions $f \in C_{exp}$. A monoton class argument as in the proof of Lemma \ref{lem:provestochcont} gives the result. 
\end{proof}
This implies in particular that $(s,x) \mapsto \Tilde P_t((s,x),B)$ is $\mathcal{B}(\R_{\geq 0}) \otimes \mathcal{E}_\Delta$-measurable for every fixed $B \in \mathcal{B}(\R_{\geq 0}) \otimes \mathcal{E}_\Delta$, where $\tilde P_t$ is the transition function of the space-time process introduced in Lemma \ref{lem:timehomtransformation}. To see this let $B=[a,b] \times B^\prime$, $B^\prime \in \mathcal{E}_\Delta$ and $C \in \mathcal{B}([0,1])$. Then by Lemma \ref{lem:affinemeasurability} $$D = \{(s,x): P_{s,s+t}(x,B^\prime) \in C\} \in \mathcal{B}(\R_{\geq 0}) \otimes \mathcal{E}_\Delta. $$
By \eqref{eq:spacetimeproducttransitionfunction}
\begin{align*}
\tilde P_t^{-1}(\cdot, & [a,b]\times B^\prime)(C) = \{(s,x): P_{s,s+t}(x,B^\prime) \mathbb{I}_{[a,b]}(s+t) \in C\} \\
& = \begin{cases}
([a-t,b-t] \times E_\Delta) \cap D & 0 \notin C, \\
([a-t,b-t]^C \times E_\Delta) \cup (([a-t,b-t] \times E_\Delta) \cap D) & 0 \in C,
\end{cases}
\end{align*}
and $\tilde P_t^{-1}(\cdot, [a,b]\times B^\prime)(C) \in  \mathcal{B}(\R_{\geq 0}) \otimes \mathcal{E}_\Delta$. 
Hence the space-time process of a Markov process with a time-inhomogeneous continuously affine transition function can be realized with respect to the smaller product $\sigma$-algebras $\mathcal{\tilde E} = \mathcal{B}(\R_{\geq 0}) \otimes \mathcal{E}_\Delta$ and $\mathcal{\tilde A} = \mathcal{B}(\R_{\geq 0}) \otimes \mathcal{A}$ and the filtration $\mathcal{\tilde F}$ given by $ \mathcal{\tilde F}_t := \mathcal{B}(\R_{\geq 0}) \otimes \F_t$ (see the remark after Lemma \ref{lem:timehomtransformation}). Together with Lemma \ref{lem:affinespacetime} we have the following corollary. 
\begin{corollary} \label{cor:contaffinespacetime1}
Let $(X,\F,\{ \PM^{(s,x)} \})$ be a continuously affine Markov process with transition function $\{P_{s,t}\}$. Then $(\tilde X, \mathcal{\tilde F},  \{ \tilde \PM^{(s,x)} \})$ is a time-homogeneous Markov process on $(\tilde \Omega, \mathcal{\tilde A})$ with state space $(\tilde E,\mathcal{\tilde E})$ and the transition function $\{ \tilde P_t \}$ satisfying 
$$\tilde P_t f_{(u^0,u)}(s,x) = \e^{u^0 t} \Phi_{s,s+t}(u) \e^{u^0 s + \scalarprod{\psi_{s,s+t}(u),x}}.$$
Furthermore, $(X, \mathcal{\tilde F},  \{ \tilde \PM^{(s,x)} \})$ is a time-inhomogeneous affine process on $(\tilde \Omega, \mathcal{\tilde A})$ with transition function $\{P_{s,t}\}$. 
\end{corollary}

\subsubsection*{Conservative affine processes}
A transition function is conservative if $P_{s,t}(x,E)$ = 1 for all $t \geq s$ and $x \in E$. By the convention $f(\Delta) = 0$, $f_0(x) = \1{x \neq \Delta}$ this is equivalent to $P_{s,t} f_0(x) = 1$ for all $t \geq s$ and $x \in E$. For an affine process in the sense of Definition \ref{def:markov2inhom} this corresponds to $\PM^{(s,x)}(X_t \neq \Delta) = 1$ for all $s,t\geq 0$ and $x \in E$. Note that in this case
\begin{equation} \label{eq:conserv}
\PM^{(s,x)}(X_t \neq \Delta) = \EV[{(s,x)}]{\e^{\scalarprod{0,X_t}}} = P_{s,s+t} f_0(x) = \Phi_{s,s+t}(0)  \e^{\scalarprod{\psi_{s,s+t}(0),x}}.
\end{equation}
We have the following lemma. 
\begin{lemma} \label{lem:conservative}
A continuously affine transition function is conservative if and only if $\Phi_{s,t}(0) = 1$ and $\psi_{s,t}(0) = 0$ for all $t \geq s$.
\end{lemma}
\begin{proof}
If the transition function is conservative, then $P_{s,t} f_0(x) = 1$ for all $x \in E$ and $t \geq s$. Hence $\Phi_{s,t}(0) \neq 0$ and so $\e^{\scalarprod{\psi_{s,t}(0),x - x^0}} = 1$. As in Lemma \ref{lem:affineproperties3} by the continuity of $\psi$ and Assumption \ref{ass:affinespan} it follows that $\psi_{s,t}(0) = 0$. Then clearly $\Phi_{s,t}(0) = 1$.
\end{proof}

\section{Càdlàg modifications}  \label{sec:cadlagversion}
In this section we prove that Markov processes with a continuously affine transition function have a càdlàg modification. For general affine transition functions this is not true as the deterministic process $X_t = x + \1{t \leq 1}$ shows.
First we consider affine processes $X$ which are Markov in the sense of Definition \ref{def:markov1}, where we have a fixed probability measure $\PM$. We then extend this to Markov processes $(X,\F,\{ \PM^{(s,x)} \})$ in the sense of Definition $\ref{def:markov2inhom}$.
For time-homogeneous processes this was proved in \citet{CT13}. We follow their treatment.

We equip $E_\Delta$ with the Alexandrov topology, where every open set with a compact complement in $E$ is declared an open neighborhood of $\Delta$. Hence a sequence $(y_k)$ in $E_\Delta$ converges to $\Delta$ if $\Vert y_k \Vert \rightarrow \infty$ (here we use the convention that $\Vert \Delta \Vert = \infty$). 
Furthermore, define the limit through rational points (see e.g. II.61 in \citet{RW00}). For $y: \mathbb{Q}_{\geq 0} \rightarrow E_\Delta$ we say that $\lim_{q \downarrow \downarrow t} y_q$ exists if $\lim_{q_k \downarrow t} y_{q_k}$ exists in $E_\Delta$ for every sequence $q_k \in \mathbb{Q}_{\geq 0}$ with $q_k \downarrow t$ and is independent of the choice of the sequence $(q_k)$. In an analogous way $\lim_{q \uparrow \uparrow t} y_q$ is defined. 

Let $\tilde \Omega \subset \Omega$ be the set where for each $t$ the left and right limits of $X$ through $\mathbb{Q}$ exist in $E_\Delta$. 
We define the $E_\Delta$-valued càdlàg process
\begin{equation} \label{eq:Xtilde}
\tilde{X}_t(\omega) := \left\{
	\begin{array}{ll} 
		\lim_{q \downarrow \downarrow t} X_{q}(\omega)  & \omega \in \tilde \Omega, \\
		\Delta & \omega \notin \tilde \Omega.
	\end{array}
\right.
\end{equation}
We show that the process $\tilde{X}$ is a modification of the original process $X$ by considering the martingales
$$M_t^{T,u} = \Phi_{t,T}(u) \e^{\scalarprod{\psi_{t,T}(u),X_t}} = \EV{\e^{\scalarprod{u,X_T}} \vert \F_t }, \qquad 0 \leq t \leq T, u \in \mathcal{U},$$
which always have a càdlàg modification. 
\begin{lemma} \label{lem:exponentialcadlag}
There is a set $ \Omega_0$ with $\PM(\Omega_0) = 1$, such that on $\Omega_0$
\begin{equation} \label{eq:rationalexponentialmartingale}
\tilde M_t^{T,u}(\omega) = \lim_{q \downarrow \downarrow t} M_q^{T,u}(\omega) = \lim_{q \downarrow \downarrow t} \Phi_{q,T}(u) \e^{\scalarprod{\psi_{q,T}(u),X_q}}(\omega)
\end{equation}
exists for $0 \leq t < T$ and defines a $\C$-valued càdlàg function\footnote{For $t = T$ we set $M_T^{T,u} := \e^{\scalarprod{u,X_T}}$.} on $[0,T]$ for almost all (in the sense of Lebesgue) $(T,u) \in (0,\infty) \times \mathcal{U}$. 
\end{lemma}
\begin{proof}
By Doob's regularity theorem (Theorem II.65.1, \citet{RW00}) \eqref{eq:rationalexponentialmartingale} defines a càdlàg function $\PM$-a.s.. However, the null set where this is not true depends on $(T,u)$. Define the $\mathcal{A} \otimes \mathcal{B}((0,\infty)  \times \mathcal{U})$-measurable set
\begin{equation}
\Gamma = \{ (\omega,T,u) \in \Omega \times (0, \infty) \times \mathcal{U}: \eqref{eq:rationalexponentialmartingale}\text{ is not a càdlàg function} \}.
\end{equation}
By the considerations above together with Fubini's theorem we obtain
$$ \int_{\Omega} \int_{(0,\infty)\times\mathcal{U}} \mathbb{I}_\Gamma(\omega,T,u) \dd{\lambda} \dd{\PM} =  \int_{(0,\infty)\times\mathcal{U}} \int_{\Omega} \mathbb{I}_\Gamma (\omega,T,u)  \dd{\PM} \dd{\lambda} = 0, $$
where $\lambda$ denotes the Lebesgue measure on $(0,\infty) \times \mathcal{U}$. This proves the result.
\end{proof}

To infer the càdlàg property of $\tilde X$ from the càdlàg property of the martingales $\tilde M$ we use the purely analytical Lemma \ref{lem:xconvergence}. To formulate this lemma we have to introduce some notation. 
Let $\Re$ and $\Im$ be the real and imaginary part of a complex number (this is then interpreted componentwise on $V + \i V$).  
Denote by $\Pi$ the projection on $\mathrm{span}(\Re \mathcal{U}) \subset V$ and by $\Pi^\bot = id_V - \Pi$ the projection to the orthogonal complement $\mathrm{span}(\Re \mathcal{U})^\bot \subset V$.
Also denote by $\Pi$ and $\Pi^\bot$ the extension of these projections to $V + \i V$, e.g. $\Pi(u) = \Pi (\Re u) + \i \Pi (\Im u)$. 
Fix $m \in \mathbb{N}$ and $r > 0$ and define the set $K = \{ u \in \mathcal{U}^m: \Vert u \Vert \leq r \}$. Let $p$ be the dimension of $\mathrm{span}(\mathcal{\Re U})$. Then there exist linearly independent $(u_1, \dots, u_{p}) \in K \cap \Re \mathcal{U}$ and linearly independent $(u_{p+1},\dots,u_d) \in \Pi^\bot K$. 
Fix $s>0$. $\Phi$ is continuous on the compact set $\{(t,T,u): 0 \leq o(T,u) \leq t \leq T \leq s +1, u \in K \}$. Hence it is uniformly continuous. By Lemma \ref{lem:affineproperties3} iii) $\Phi_{s,s}(u) = 1$, so for each $\tilde c > 0$ there is $\delta > 0$ such that $\vert \Phi_{t,T}(u) \vert \geq \vert \Phi_{s,s}(u) \vert - \tilde c = 1 - \tilde c$ for all $0 \leq s-\delta \leq t$, $s \leq T \leq s+\delta$. Choosing e.g. $\tilde c = \frac{1}{2}$ this gives a compact set on which also $\psi$ is continuous and hence uniformly continuous and bounded. Hence there is $\epsilon > 0$, $\eta > 0$ such that for every $t \in \mathcal{I}_s = (s,s+\epsilon)$ and $q$ in $(s-\epsilon,s+\epsilon) \cap [0,t)$
\begin{equation} \label{eq:linindu} \Pi \psi_{q,t}(\bar u_1), \dots, \Pi \psi_{q,t}(\bar u_{p}) \quad \text{ as well as } \quad \Pi^\bot \psi_{q,t}(\bar u_{p+1}), \dots, \Pi^\bot \psi_{q,t}(\bar u_{d}) \end{equation}
are linearly independent for all $\Vert u_i - \bar u_i \Vert < \eta$, $1 \leq i \leq d$ 
and
$$ \inf_{u \in K} \vert \Phi_{q,t}(u) \vert > c \text{ and } \sup_{u \in K} \Vert \psi_{q,t}(u) \Vert^2 < C .$$
\begin{lemma}
\label{lem:xconvergence}
Let $q_k \rightarrow s$ and $x_{q_k}$ be a sequence with values in $E_\Delta$. Then
\begin{enumerate}[(i) ]
\item If for Lebesgue-almost all $(t,u) \in \mathcal{I}_s \times K$
$$ \lim_{k \rightarrow \infty} N_{q_k}^{t,u} := \lim_{k \rightarrow \infty} \e^{\scalarprod{\psi_{q_k,t}(t,u),x_{q_k}}} \in \C \setminus \{ 0 \}, $$
then also $\lim_{k \rightarrow \infty} x_{q_k}$ exists in $E$. 
\item If there exists some $(t,u) \in \mathcal{I}_s \times K$ such that 
$$ \lim_{k \rightarrow \infty} \e^{\scalarprod{\psi_{q_k,t}(,u),x_{q_k}}} = 0,$$
then $\lim_{k \rightarrow \infty} \Vert x_{q_k} \Vert = \infty$, i.e. $x_{q_k} \rightarrow \Delta$. If  $(x_{q_k})$ is even an $E$-valued sequence, then $\mathrm{relint}(\Re \mathcal{U}) \neq \emptyset$ and for all $u \in \mathrm{relint}(\Re \mathcal{U})$
$$\lim_{k \rightarrow \infty} \e^{\scalarprod{u,x_{q_k}}} = 0.$$
\end{enumerate}
\end{lemma}
We prove this lemma at the end of this section. 
Using this lemma we obtain the following result. 
\begin{lemma} \label{thm:tildeXmod}
$\tilde X$ is a modification\footnote{Here $X$ is already interpreted as a process on the extended probability space $(\Omega,\mathcal{A}^{\PM})$.} of $X$ on $(\Omega,\mathcal{A}^{\PM})$, i.e. $\PM(\tilde{X}_t = X_t) = 1$ for all $t \geq 0$. 
\end{lemma}
\begin{proof}
By Lemma \ref{lem:exponentialcadlag} we can fix ${\Omega_0}$ with $\PM(\Omega_0)=1$ such that for each $\omega \in {\Omega_0}$ \eqref{eq:rationalexponentialmartingale}
is a càdlàg function for almost all $(T,u)$. Fix such $\omega \in \Omega_0$. Let $q_k \uparrow \uparrow s$ or $q_k \downarrow \downarrow s$. For all $(t,u) \in \mathcal{I}_s \times K$ and $k$ large it holds that  $\Phi_{q_k,t}(u) \neq 0$. For almost all $(t,u) \in \mathcal{I}_s \times K$ $\lim_{k \rightarrow \infty} M_{q_k}^{t,u}$ exists finitely valued and since $\Phi_{q_k,T}(u) \neq 0$ also $\lim_{k \rightarrow \infty} \e^{\scalarprod{\psi_{q_k,t}(u),X_{q_k}(\omega)}}$ exists finitely valued. 
 By Lemma \ref{lem:xconvergence}
$\lim_{k \rightarrow \infty} X_{q_k}(\omega) \in E_\Delta$. 
This holds for all $s \geq 0$ and hence for each $\omega \in \Omega_0$ the left and right limits of $X$ through $\mathbb{Q}$ exist in $E_\Delta$ and $\tilde{X}_s(\omega) = \lim_{q_k \downarrow \downarrow s} X_{q_k}(\omega)$ defines a càdlàg function in $E_\Delta$. Then $\Omega_0 \subset \tilde \Omega$ and $\PM(\tilde \Omega) = 1$ (for $\PM(\tilde \Omega)$ to be defined we consider the $\PM$-completion of $\mathcal{A}$). 

It remains to be proved that the $E_\Delta$-valued càdlàg process $\tilde X$  is a modification of $X$. Since $X$ is stochastically continuous (see the remark after Definition \ref{def:stochcont}) and convergence in probability implies almost sure convergence along a subsequence we have for each $t$ that 
$$\lim_{q \downarrow \downarrow t} X_q \overset{p}{\rightarrow} X_t  \Rightarrow \lim_{k \rightarrow \infty} X_{q_k}(\omega) = X_t(\omega) \text{ for all $\omega \in \bar \Omega$, where } \PM(\bar \Omega)=1.$$
However, on $\tilde \Omega$ we have that 
$$ \lim_{k \rightarrow \infty} X_{q_k}(\omega) = \lim_{q \downarrow \downarrow t} X_q(\omega) = \tilde{X}_t.$$ This yields $\tilde{X}_t = X_t$ on $\tilde \Omega \cap \bar \Omega$. Since $\PM(\tilde \Omega \cap \bar \Omega) = 1$, this proves the lemma. 
\end{proof}
\begin{remark}
Under a single measure $\PM$ we could also have defined $\tilde X = \Delta$ on $\Omega_0$ to get a modification. But $\Omega_0$ depends on $\PM$ while $\tilde \Omega$ does not. This is important when considering a Markov process $(X,\F,\{\PM^{(s,x)} \})$, where there is a whole family of probability measures. 
\end{remark}

Next we show that $\tilde X$ a.s. stays in $\Delta$, as soon as it is $\Delta$ or approaches $\Delta$ from the left. 
Note that this is true for $X$ by definition, but not necessarily for $\tilde X$. 
Define 
\begin{equation} \label{eq:explosiontime}
\begin{aligned}
T_\Delta & := \inf \{ t>0: \tilde X_{t-}  = \Delta \text{ or } \tilde X_t = \Delta \}. \\
\end{aligned}
\end{equation}

\begin{lemma} \label{thm:stayinfinity}
$\tilde X = \Delta$ on $[T_\Delta,\infty)$ $\PM$-a.s.. 
\end{lemma}
\begin{proof}
By the definition of $\tilde X$, $\tilde X_{s-} = \Delta$ implies $\lim_{q \uparrow \uparrow s} \Vert X_q \Vert = \infty$ and $\tilde X_{s}  = \Delta$ implies $\lim_{q \downarrow \downarrow s} \Vert X_q \Vert = \infty$. 
In either case if there is a subsequence $q_k$ with $X_{q_k}(\omega) = \Delta$ (the subsequence can depend on $\omega$), then by the definition of $X$ and $\Delta$, $X = \Delta$ on $(s,\infty)$ and then $\tilde X = \Delta$ on $[s,\infty)$. So without loss of generality we can restrict to $E$-valued sequences. 

Assume there is $\omega \in \Omega_0$ and $q \rightarrow s$ with $\lim_{q \rightarrow s} \Vert X_q(\omega) \Vert = \infty$, such that $X_q(\omega) \in E$ for all $q$ (if not, the lemma is already true). By Lemma \ref{lem:xconvergence} i) 
there is a subset $(T,u)$ of $\mathcal{I}_{s} \times K$ of positive Lebesgue-measure such that 
$\lim_{q \rightarrow s} \e^{\scalarprod{\psi_{q,T}(u),X_{q}(\omega)}}$
is either zero or infinite or does not exist at all. 
Since $\omega \in \Omega_0$ by Lemma \ref{lem:exponentialcadlag} the limit exists finitely valued for almost all $(T,u) \in \mathcal{I}_s \times K$ (remember that $\Phi$ is continuous and $\Phi \neq 0$ on this set). Hence there are $T(\omega),u(\omega)$ such that
$$\lim_{q \rightarrow s} \e^{\scalarprod{\psi_{t,T(\omega)}(u(\omega)),X_{q}(\omega)}} = 0.$$
Lemma \ref{lem:xconvergence} ii) then yields that $\mathrm{relint}(\Re \mathcal{U}) \neq \emptyset$ and for all $u \in \mathrm{relint}(\Re \mathcal{U})$
\begin{equation} \label{eq:M0}
\lim_{q \rightarrow s} \e^{\scalarprod{u,X_{q}(\omega)}} = 0.
\end{equation}

Hence we can fix $u \in \mathrm{relint}(\Re \mathcal{U})$. By Lemma \ref{lem:affineproperties3} $\Phi_{t,T}(u)$ and $\psi_{t,T}(u)$ are real-valued functions. For $T>0$ we can choose $s<T$ such that  $\psi_{t,T}(u) \in \mathrm{relint}(\Re \mathcal{U})$ and $\Phi_{t,T}(u) > 0$ for $s \leq t \leq T$. Consider the set $\Omega_{s,T} = \{ \omega \in \Omega_0:  s <  T_{\Delta}(\omega) < T \}$ and the nonnegative càdlàg martingale from equation \eqref{eq:rationalexponentialmartingale}
which on $\Omega_0$ satisfies $$\tilde M_t^{T,u}(\omega) =  \Phi_{t,T}(u) \e^{\scalarprod{\psi_{t,T}(u),\tilde X_t(\omega)}}.$$ For $\omega \in \Omega_{s,T}$ the considerations leading to \eqref{eq:M0} give
$$\e^{\scalarprod{\psi_{T_\Delta(\omega),T}(u),\tilde X_{T_\Delta(\omega)}(\omega)}} = 0 \quad \text{ or } \quad \e^{\scalarprod{\psi_{T_\Delta(\omega),T}(u),\tilde X_{T_\Delta(\omega)-}(\omega)}} = 0.$$
By Theorem II.78.1 in \citet{RW00} $\tilde M_t^{T,u}(\omega) = 0$ for all $T_{\Delta}(\omega) \leq t \leq T$ a.s. on $\Omega_{s,T}$.  
 Since $\Phi_{t,T}(u) > 0$, then
$\tilde X_t(\omega) = \Delta$
for $T_{\Delta}(\omega) \leq t \leq T$ a.s. on $\Omega_{s,T}$. Since $T > T_{\Delta}(\omega)$, by the definition of $\tilde X$ also $X_t(\omega) = \Delta$ for some $t \in (T_{\Delta}(\omega),T]$. By the definition of $\Delta$ and $X$ then $X_t(\omega) = \Delta$ on $[T,\infty)$ and hence also $\tilde X_t(\omega)$. 
Since $T$ was arbitrary, the lemma follows. 
\end{proof}
Define the process $$\bar{X} := \tilde{X} \mathbb{I}_{[0,T_{\Delta})} + \Delta \mathbb{I}_{[T_{\Delta}, \infty)}. $$ 
\begin{corollary} \label{cor:cadlag1}
The process $\bar X$ is càdlàg in $E_\Delta$, a modification of $X$ and an affine process on the completed filtered probability space $(\Omega,\mathcal{A}^{\PM},\F^{\PM},\PM)$.  
\end{corollary}
\begin{proof}
By Lemma \ref{thm:tildeXmod} and Lemma \ref{thm:stayinfinity} $\bar X$ is a modification of $X$. 
By Lemma \ref{lem:MarkovCompletion} it is an affine process with the same transition function. Note that $\bar X$ is adapted to $\F^{\PM}$, but not necessarily to $\F$. 
\end{proof}

So far we have worked with a single probability measure $\PM$. This demonstrates the existence of a càdlàg modification for affine Markov processes in the sense of Definition \ref{def:markov1}. Now let $(X,\mathcal{F},\{ \PM^{(s,x)} \})$ be an affine process in the sense of Definition \ref{def:markov2inhom}. For each $(s,x)$ we can apply Corollary \ref{cor:cadlag1} to the affine process $X$ on $(\Omega,\mathcal{A},\F,\PM^{(s,x)})$ with continuously affine transition function $(P_{s+t,s+T})_{0 \leq t \leq T}$ and immediately get that $\bar X$ is a version on $(\Omega,\mathcal{A}^{(s,x)},\F^{(s,x)},\PM^{(s,x)})$, where $\mathcal{A}^{(s,x)}$ is the completed $\sigma$-algebra and  $\F_t^ {(s,x)}$ the completed filtration with respect to the measure $\PM^{(s,x)}$. The $\sigma$-algebras in this case still depend on the measure $\PM^{(s,x)}$. To get a single filtration define
\begin{equation} \label{eq:affinecompletion}
\mathcal{\bar A} := \bigcap_{(s,x)} \mathcal{A}^{(s,x)}, \quad  \mathcal{\bar F}_t := \ \bigcap_{(s,x)} \F_t^{(s,x)}, t \in [0,\infty].
\end{equation}
We can now formulate the final theorem of this section. 
\begin{theorem}
Let $(X,\mathcal{F},\{ \PM^{(s,x)} \})$ be a Markov process with continuously affine transition function. Then $(\bar X, \mathcal{\bar F},\{ \PM^{(s,x)} \})$ on $(\Omega,\mathcal{\bar A})$ is a modification of the affine process $(X,\mathcal{F},\{ \PM^{(s,x)} \})$ which has càdlàg paths. 
\end{theorem}
\begin{proof}
The set $\{X_t \neq \bar X_t\}$ has $\PM^{(r,x)}$-measure zero for all $(r,x)$. Hence it is in $\mathcal{\bar F}_t$.
$X_t$ is $\mathcal{\bar F}_t$-measurable, hence $\bar X_t$ is $\mathcal{\bar F}_t$-measurable and therefore adapted. 
By corollary \ref{cor:cadlag1} $\bar X$ is a Markov process on $(\Omega,\mathcal{A}^{(s,x)},\F^{(s,x)},\PM^{(s,x)})$ and since $\mathcal{\bar F}_t \subset \F_t^{(s,x)}$, it is also a Markov process on the smaller filtration. As this holds for all $\PM^{(s,x)}$, this proves the claim. 
\end{proof}

\begin{remark}
For a conservative Markov process with a continuously affine transition function (see Lemma \ref{lem:conservative}) one can construct an $E$-valued càdlàg process by setting $\hat X = \bar X \1{T_\Delta = \infty} + X_0 \1{T_\Delta < \infty}$. Then $\{ \hat X_t \neq X_t\} \subset \{T_\Delta < \infty \}$. For a conservative process $\PM^{(s,x)}(T_\Delta < \infty) = 0$ for all $(s,x)$ and hence also $\hat X$ is a modification of $X$. 
\end{remark}

The càdlàg version of the affine process has the strong Markov property. Set $\bar X_\infty := \Delta$ and remember $f(\Delta) = 0$. 
\begin{theorem}
The càdlàg process $(\bar X,\mathcal{\bar F},\{ \PM^{(s,x)} \})$ is a strong Markov process, i.e. for each stopping time $\tau$, bounded measurable $f$ and $s \geq 0$
$$\EV[(r,x)]{f(\bar X_{\tau+s}) \vert \mathcal{\bar F}_\tau} = \EV[(r+\tau,\bar X_\tau)]{f(\bar X_s)} = P_{r+\tau,r+\tau+s} f(\bar X_\tau).$$
\end{theorem}
\begin{proof}
By Lemma \ref{lem:affinemeasurability} the function $(s,x) \mapsto P_{t,t+s} f(x)$ is $\mathcal{B}(\R_{\geq 0}) \otimes \mathcal{E}$-measurable. Additionally  $(\tau,\bar X_\tau)$ is $\mathcal{\bar F}_\tau / (\mathcal{B}(\R_{\geq 0}) \otimes \mathcal{E})$-measurable. So 
$\EV[(r+\tau,\bar X_\tau)]{f_u(X_s)}$ is $\mathcal{\bar F}_\tau$-measurable. 

For a stopping time $\tau$ define an approximating sequence of stopping times $$\tau_n(\omega) := \begin{cases} k 2^{-n} & \text{ if } (k-1)2^{-n} \leq \tau(\omega) < k 2^{-n}, k \in \mathbb{N}, \\
\infty & \text{ if } \tau(\omega) = \infty. \end{cases}$$
Let $\Lambda \in \mathcal{\bar F}_\tau$. Then $\Lambda_{n,k} := \{ \omega: \tau_n(\omega) = k 2^{-n} \} \cap \Lambda  \in \mathcal{\bar F}_{k 2^{-n}}$. The simple Markov property yields
\begin{align*}
\EV[(r,x)]{f_u(\bar X_{\tau_n+s}) \mathbb{I}_{\Lambda}} & = \sum_{k \in \mathbb{N}} \EV[(r,x)]{ \mathbb{I}_{\Lambda_{n,k}} \EV[(r,x)]{f_u(\bar X_{k 2^{-n}+s})  \vert \mathcal{\bar F}_{k 2^{-n}}} }  \\ 
&  =  \sum_{k \in \mathbb{N}} \EV[(r,x)]{\mathbb{I}_{\Lambda_{n,k}} \Phi_{r+k 2^{-n},r+k 2^{-n}+s}(u)
\e^{\scalarprod{ \psi_{r+k 2^{-n},r+k 2^{-n}+s}(u),\bar X_{k 2^{-n}}}}} \\
& = \EV[(r,x)]{\mathbb{I}_{\Lambda} \Phi_{r+\tau_n,r+\tau_n+s}(u)
\e^{\scalarprod{ \psi_{r+\tau_n,r+\tau_n+s}(u),\bar X_{\tau_n}}}}.
\end{align*}
For $u \in \i V$ by dominated convergence, continuity of $\Phi$ and $\psi$ and right-continuity of $\bar X$
\begin{align*} \EV[(r,x)]{f_u(\bar X_{\tau+s}) \mathbb{I}_{\Lambda}} & =  \EV[(r,x)]{\mathbb{I}_{\Lambda} \Phi_{r+\tau,r+\tau+s}(u)
\e^{\scalarprod{ \psi_{r+\tau,r+\tau+s}(u),\bar X_{\tau}}}} \\
& = \EV[(r,x)]{\mathbb{I}_{\Lambda} P_{r+\tau,r+\tau+s} f_u(\bar X_{\tau})} \\
& =  \EV[(r,x)]{\mathbb{I}_{\Lambda} \EV[(r+\tau,\bar X_\tau)]{f_u(\bar X_s)}}.
\end{align*}
An application of the functional monotone class theorem as in Lemma \ref{lem:affinemeasurability} then yields the result for bounded and measurable functions $f$. 
\end{proof}
If we assume that the filtration $\F$ is the natural filtration of the process $X$, we can show that the completed
filtrations $\F^{(s,x)}$ are already right-continuous. By the definition of $\mathcal{\bar F}$ right-continuity transfers to $\mathcal{\bar F}$. 
\begin{theorem} \label{thm:filtrationrightcont}
Let $\F = \F^X$ be the natural filtration of the affine process $X$ and $\mathcal{A} = \F^X_\infty$. Then $\F^{\PM}$ is right-continuous. 
\end{theorem}
\begin{proof}
To show this we use the continuity of $\Phi$ and $\psi$ together with the fact that $\bar X$ is càdlàg. Repeating the arguments in the proof of Theorem 4 in \citet{CT13} with the inhomogeneous versions of $\Phi$ and $\psi$ we get that for $k \in \mathbb{N}$, $u_1, \dots, u_k \in \i V$, $t_1,\dots,t_k \in \R_{\geq 0}$
\begin{equation} \label{eq:filtrightcont1}
\EV{\e^{\scalarprod{u_1,\bar X_{t_1}} + \dots + \scalarprod{u_k,\bar X_{t_k}}} \vert \F^{\PM}_t} = \EV{\e^{\scalarprod{u_1,\bar X_{t_1}} + \dots + \scalarprod{u_k,\bar X_{t_k}}} \vert \F^{\PM}_{t+}}.
\end{equation}
Consider the vector space of functions $$\mathcal{H} := \{ \text{bounded $\mathcal{A^{\PM}}$-measurable } Z: \EV{Z \vert \F^{\PM}_t} = \EV{Z \vert \F^{\PM}_{t+}} \}.$$ This space contains the constant functions and 
is a monotone class. 
The set $\mathcal{C} := \{\prod_{l=1}^n g_l(X_{t_l}), g_l \in C_{exp} ,t_l \geq 0, n \in \mathbb{N}\}$ of functions from $\Omega$ to $\R$ is closed under multiplication, generates $\mathcal{F}^X_\infty$ (see Lemma \ref{lem:borelexponentialgenerated}) and by \eqref{eq:filtrightcont1} is a subset of $\mathcal{H}$. 
Hence we can apply the functional monotone class theorem (see Theorem 0.2.2 in \citet{RY10}). 
This yields $\EV{Z \vert \F^{\PM}_t} = \EV{Z \vert \F^{\PM}_{t+}}$ for bounded $\F^X_\infty$-measurable functions $Z$.

Now consider $\tilde A \in \F_{t+}^{\PM}$. Then there is a set $A \in \F_{t+}$ and $B_1, B_2 \in \F_\infty^{\PM}$, such that $\PM(B_1) = \PM(B_2) = 0$ and $A \setminus B_1 \subset \tilde A \subset A \cup B_2$. For $Z = \mathbb{I}_A$ there is a $\F^{\PM}_t$-measurable $f$, 
such that 
$\mathbb{I}_A = f \mathbb{I}_{N^C} + \mathbb{I}_A \mathbb{I}_N$ with $\PM(N) = 0$. 
Since $\F^{\PM}_t$ contains all null sets, it follows that $f \mathbb{I}_{N^C}$ is $\F^{\PM}_t$-measurable. Furthermore $\PM(A \cap N) = 0$, so also $\mathbb{I}_A \mathbb{I}_N$ is $\F^{\PM}_t$-measurable. Then $\mathbb{I}_A$ is $\F^{\PM}_t$-measurable and $A \in \F^{\PM}_t$. Hence also $\tilde A \in \F^{\PM}_t$ and this shows $\F^{\PM}_t = \F^{\PM}_{t+}$. 
\end{proof}

If $\F = \F^X$ and there exist shift operators $\theta_t$ (e.g. for the canonical realization), the strong Markov property extends to $\mathcal{\bar F}_\infty$-measurable $Z$, i.e. for a stopping time $\tau$, for all $s,t \geq 0$, $x \in E$ and bounded $\mathcal{\bar F}_\infty$-measurable $Z$
$$\EV[(s,x)]{Z \circ \theta_\tau \vert \mathcal{\bar F}_\tau} = \EV[(s+\tau,\bar X_\tau)]{Z} \quad \text{ a.s. on the set $\{\bar X_\tau \neq \Delta \}$}. $$
This follows by the arguments presented in Theorem III.8 and Theorem III.9 of \citet{RW00}.

Furthermore the càdlàg property and possibly right continuity of the filtration carry over to the space-time process $\tilde X = (\Theta,\bar X)$ introduced in Corollary \ref{cor:contaffinespacetime1}.
Additionally we still have continuity of $(s,x) \mapsto \tilde P_r f_{(u^0,u)}(s,x)$ and the proof of the strong Markov property also transfers. Then the extended filtration $\mathcal{\tilde F}$ is a right-continuous strong Markov filtration in the sense of \citet{CJ80} and affine processes can be embedded in the Markov semimartingale setting of \citet{CJ80} (see also section \ref{sec:regularity}). 

I.e. if $(X,\F^X,\{\PM^{(s,x)}\})$ is the coordinate process of the transition function we call the càdlàg time-homogeneous strong Markov process $(\tilde X, \mathcal{\tilde F},\{ \tilde \PM^{(s,x)} \})$ the canonical space-time realization of the continuously affine transition function $\{ P_{s,t} \}$. Here the filtration $\mathcal{\tilde F}_t = \mathcal{B}(\R_{\geq 0}) \otimes \mathcal{\bar F}_t$ is a right-continuous strong Markov filtration.

\subsubsection{Proof of Lemma \ref{lem:xconvergence}}
The proofs in this section follow closely the proofs in \citet{CT13}. 
\begin{lemma} \label{lem:divergencedirections} Let $(x_k)$ be a $E$-valued sequence such that
$\lim_{k \rightarrow \infty} \Pi x_k${ exists finitely valued and}
$\limsup_{k \rightarrow \infty} \Vert \Pi^\bot x_k \Vert = \infty. $
There exists a subsequence $(x_k)$ denoted again by $(x_k)$ and a finite number of orthonormal directions $g_i \in \mathrm{span}(\Re \mathcal{U})^\bot$ such that $$\lim_{k\rightarrow \infty} x_k - \sum_i \scalarprod{x_k,g_i} g_i$$ exists while $\lim_{k \rightarrow \infty} \scalarprod{x_k,g_i} = \infty$, where the rates of divergence are non-increasing in $i$ in the sense that 
$\limsup_{k \rightarrow \infty} \frac{\scalarprod{x_k,g_{i+1}}}{\scalarprod{x_k,g_i}} < \infty$. 

Fix $v < T < \sigma(v,0)$ and consider the set $\mathcal{T} := \{(s,t): v \leq s \leq t \leq T\}$. Then there exist continuous functions $R: \mathcal{T} \rightarrow \R_{>0}$ and $\lambda_i: \mathcal{T} \rightarrow V$, such that for all $u \in B^{\i V}_{R(s,t)} := \{ u \in \i V: \Vert u \Vert < R(s,t) \}$
\begin{equation}
\scalarprod{\psi_{s,t}(u),g_i} = \scalarprod{\lambda_i(s,t),u}.
\end{equation} 
\end{lemma}
\begin{remark}
The above lemma states that if $x_k$ is only diverging in $\mathrm{span}(\Re \mathcal{U})^\bot$,we can find directions of divergence $g_i \in \mathrm{span}(\Re \mathcal{U})^\bot$ and open sets in $\i V$, so that $\scalarprod{\psi_{s,t}(u),g_i}$ is linear in $u$ for all $i$. Note that with $g_i \in <\Re \mathcal{U}>^\bot$ and $ \Re \psi_{s,t}(u) \in \mathcal{U}$
\begin{align} \label{eq:realimagtrans}
\scalarprod{\psi_{s,t}(u),g_i} = \scalarprod{\Pi^\bot \Re \psi_{s,t}(u) +\i  \Pi^\bot  \Im \psi_{s,t}(u),g_i} = \i \scalarprod{\Im \Pi^\bot \psi_{s,t}(u),g_i},
\end{align}
and it suffices to show linearity of $\scalarprod{\Im \Pi^\bot \psi_{s,t}(u),g_i}$.
\end{remark}
\begin{proof}[Proof of \ref{lem:divergencedirections}]
For the first part we can choose directions of convergence by inductively choosing subsequences of $(x_k)$ for which
$$g_l = \lim_{k \rightarrow \infty} \frac{x_k - \sum_{j=1}^{l-1} \scalarprod{x_k,g_j} g_j}{\Vert x_k - \sum_{j=1}^{l-1} \scalarprod{x_k,g_j} g_j \Vert}$$
converges. For instance $g_1 = \lim_{k \rightarrow \infty} \frac{x_k}{\Vert x_k \Vert}$. 
By construction the $g_l$ are orthogonal. Choosing further subsequences of $(x_k)$ one can arrange that the rates of convergence are non-increasing in $l$. 

Next note that $\Phi_{s,t}(0) \neq 0$ for $(s,t) \in \mathcal{T}$. By continuity of $\Phi$ there is $r > 0$, such that $\Phi_{s,t}(u) \neq 0$ for $u \in B^{\i V}_r$ and $(s,t) \in \mathcal{T}$. For $u \in B^{\i V}_r$ define
$$\Theta(u,x) :=  \frac{P_{s,t} f_u(x)}{\Phi_{s,t}(0) \e^{\scalarprod{ \psi_{s,t}(0), \Pi x}}} =\frac{\Phi_{s,t}(u) \e^{\scalarprod{ \psi_{s,t}(u), x}}}{\Phi_{s,t}(0) \e^{\scalarprod{\Pi \psi_{s,t}(0), \Pi x}}}  .$$
By Lemma \ref{lem:affineproperties3} ii) 
$\Im \Phi_{s,t}(0) = \Im \psi_{s,t}(0) = 0$ and $\Phi_{s,t}(0) \e^{\scalarprod{\Pi \psi_{s,t}(0), \Pi x}} > 0$. Since a characteristic function is always positive definite we conclude that $\Theta(\cdot,x)$ is a positive definite function in $B^{\i V}_{r}$ for all $x \in E$.
Since $\psi_{s,t}(0) \in \Re \mathcal{U}$, it follows that $\Pi \psi_{s,t}(0) = \psi_{s,t}(0)$. Hence $\Theta(0,x) = 1$ and by Lemma 3.2. in \citet{KST11}
\begin{equation} \label{eq:theta}
\vert \Theta(u+v,x) - \Theta(u,x)\Theta(v,x) \vert^2 \leq 1, \qquad u,v \in B^{\i V}_{\frac{r}{2}}. 
\end{equation}
Using $\Re \psi_{s,t}(u) \in \Re \mathcal{U} \subset \mathrm{span}(\Re \mathcal{U})$ and $\Pi \Im = \Im \Pi$ one obtains
$$P_{s,t} f_u(x) = \left \vert \Phi_{s,t}(u) \right \vert \e^{ \scalarprod{\Pi \Re \psi_{s,t}(u),\Pi x}} \e^{\i ( \mathrm{arg} (\Phi_{s,t}(u)) + \scalarprod{\Im \Pi  \psi_{s,t}(u), \Pi x}+\scalarprod{\Im \Pi^\bot \psi_{s,t}(u),\Pi^\bot x})}.$$
Define
the following continuous functions 
\begin{align*}
\beta_1(u,v,s,t) & := \Im \Pi^\bot \psi_{s,t}(u+v)  \\
\beta_2(u,v,s,t) & := \Im \Pi^\bot \psi_{s,t}(u) + \Im \Pi^\bot \psi_{s,t}(v)  \\
r_1(u,v,x,s,t) & := 
\left \vert \frac{\Phi_{s,t}(u+v)}{\Phi_{s,t}(0)} \right \vert \ex{\scalarprod{\Re \Pi (\psi_{s,t}(u+v) - \psi_{s,t}(0)), \Pi x}} \\
r_2(u,v,x,s,t) & := \left \vert \frac{\Phi_{s,t}(u)\Phi_{s,t}(v)}{\Phi_{s,t}(0)^2} \right \vert \ex{\scalarprod{\Re \Pi (\psi_{s,t}(u) + \psi_{s,t}(v) - 2 \psi_{s,t}(0)),\Pi x}} \\
\alpha_1(u,v,x,s,t) & := \mathrm{arg}  \left ( \frac{\Phi_{s,t}(u+v)}{\Phi_{s,t}(0)} \right) + \scalarprod{\Im \Pi \psi_{s,t}(u+v),\Pi x} \\
\alpha_2(u,v,x,s,t) & := \mathrm{arg} \left ( \frac{\Phi_{s,t}(u)\Phi_{s,t}(v)}{\Phi_{s,t}(0)^2} \right ) + \scalarprod{\Im \Pi (\psi_{s,t}(u) + \psi_{s,t}(v)),\Pi x}.
\end{align*}
Inserting in equation \eqref{eq:theta} (and suppressing the arguments of the functions)
\begin{align*}
1 & \geq \left \vert r_1 \e^{\i (\alpha_1 + \scalarprod{\beta_1,\Pi^\bot x})} - r_2 \e^{\i (\alpha_2 + \scalarprod{\beta_2,\Pi^\bot x})} \right \vert^2 \\ & = r_1^2 + r_2^2 - 2r_1 r_2 \cos(\alpha_1 - \alpha_2 + \scalarprod{\beta_1 - \beta_2, \Pi^\bot x}).
\end{align*}
Since $r_1^2 + r_2^2 \geq 2 r_1 r_2$ this yields for $u,v \in B^{\i V}_{\frac{r}{2}}$ and $x \in E$
\begin{equation} \label{eq:r1r2greater12}
r_1 r_2 \Big( 1- \cos(\alpha_1 - \alpha_2+ \scalarprod{\beta_1 - \beta_2, \Pi^\bot x}) \Big) \leq \frac{1}{2}.
\end{equation}
Furthermore, define
$$ R(x,s,t) := \sup \left \{ \rho \leq \frac{r}{2}: r_1(u,v,x,s,t) r_2(u,v,x,s,t) \geq \frac{4}{5} \text{ for } u,v \in B^{\i V}_\rho \right \}. $$
$R(x,s,t) > 0$ for $x \in E$, since $r_1(0,0,x,s,t) = r_2(0,0,x,s,t) = 1$ and $r_1$ and $r_2$ are continuous. This also gives continuity of $R(x,s,t)$. Since $\lim_k \Pi x_k$ exists finitely valued, also $$R(s,t) := \frac{1}{2} \inf_k R(x_k,s,t) > 0.$$

Now suppose that $\scalarprod{\Pi^\bot \Im \psi_{s,t}(u),g_1}$ is not linear for $u \in B^{\i V}_{R(s,t)}$. Then by the continuity of $\beta_1$ and $\beta_2$ there exist $u^*,v^* \in B^{\i V}_{R(s,t)}$ and an open neighborhood $O \subset B^{\i V}_{R(s,t)}$, such that 
$$\scalarprod{\beta_1(u,v,s,t) - \beta_2(u,v,s,t),g_1} \neq 0$$
and the left-hand side is not constant on $O$. $\alpha_1, \alpha_2$ are continuous in $x$ and only depend on $x$ via $\Pi x$. Since $\Pi x_k$ is converging, while $\Pi^\bot x_k$ is diverging with highest divergence rate in direction $g_1$ there exist some $(u,v) \in O$ and some $k \in \mathbb{N}$, such that
\begin{align*}
\cos(\alpha_1(u,v, x_k,s,t) - \alpha_2(u,v,x_k,s,t) + \scalarprod{\beta_1(u,v,s,t) - \beta_2(u,v,s,t), \Pi^\bot x_k})  \leq \frac{1}{3}.
\end{align*}
Together with $r_1(u,v,x_k,s,t) r_2(u,v,x_k,s,t) \geq \frac{4}{5}$ this yields
\begin{align*}
 r_1 r_2 & (1- \cos(\alpha_1 - \alpha_2 + \scalarprod{\beta_1 - \beta_2, \Pi^\bot x_k}))  \geq \frac{8}{15} > \frac{1}{2},
\end{align*}
which contradicts equation \eqref{eq:r1r2greater12}. Hence
$$\scalarprod{\beta_1(u,v,s,t) - \beta_2(u,v,s,t),g_1} = 0$$
for all $(u,v) \in B^{\i V}_{R(s,t)}$. This gives linearity of $\scalarprod{\Pi^\bot \Im \psi_{s,t}(u),g_1}$ for $u \in B^{\i V}_{R(s,t)}$.  By equation \eqref{eq:realimagtrans} this is equivalent to linearity of $\scalarprod{\psi_{s,t}(u),g_1}$. Since $\psi_{s,t}(u)$ is continuous there is a continuous function $\lambda_1(s,t)$ such that 
\begin{equation}
\scalarprod{\psi_{s,t}(u),g_1} = \scalarprod{\lambda_1(s,t),u}
\end{equation}
for all $u \in B^{\i V}_{R(s,t)}$. Proceeding inductively gives the assertion.
\end{proof}
\begin{proof}[Proof of \ref{lem:xconvergence} i)]
Choose $(t,u) \in \mathcal{I}_s \times K$, such that $\lim_{k \rightarrow \infty} N_{q_k}^{t,u} 
\in \C \setminus \{0\}$. This implies that there is $N$ such that $x_{q_k} \neq \Delta$ for all $k \geq N$. 
We pass to a subsequence of $q_k$, such that $q_k \leq t$ and $x_{q_k} \in E$ (this excludes only finitely many members of the sequence). Define
$$ A := \limsup_{k \rightarrow \infty} \scalarprod{\Re \psi_{q_k,t}(u), x_{q_k}}, \qquad a := \liminf_{k \rightarrow \infty} \scalarprod{\Re \psi_{q_k,t}(u), x_{q_k}}.$$
Since $\lim_{k \rightarrow \infty} N_{q_k}^{t,u}$ exists finitely valued, $A$ is finite and since the limit does not vanish, $a$ is finite. Hence there exist subsequences, such that 
$A = \lim_{m \rightarrow \infty} \Re A_m$ and $a = \lim_{l \rightarrow \infty} \Re a_l,$
where $A_m :=  \scalarprod{\psi_{q_{k_m},t}(u), x_{q_{k_m}}}$ and $a_l :=  \scalarprod{\psi_{q_{k_l},t}(u), x_{q_{k_l}}}$.
Since $\lim_{k \rightarrow \infty} N_{q_k}^{t,u}$ exists finitely valued, we get that
\begin{align*}
\lim_{m \rightarrow \infty} \e^{\Re A_m}  \cos(\Im A_m) &= \lim_{l \rightarrow \infty} \e^{\Re a_l}  \cos(\Im a_l),  \\
\lim_{m \rightarrow \infty} \e^{\Re A_m}  \sin(\Im A_m) & = \lim_{l \rightarrow \infty} \e^{\Re a_l}  \sin(\Im a_l).
\end{align*}
Taking square and using $\cos(x)^2 + \sin(x)^2 =1$ gives $A=a$ and the existence of
\begin{equation}  \label{eq:Pixconvergence}
\lim_{k \rightarrow \infty} \scalarprod{\Re \psi_{q_k,t}(u), x_{q_k}} = \lim_{k \rightarrow \infty} \scalarprod{\Re \psi_{q_k,t}(u), \Pi x_{q_k}}. 
\end{equation}
Properly choosing $\bar u_1, \dots \bar u_p$ from \eqref{eq:linindu} gives the existence of
\begin{equation} \label{eq:realpart}
\lim_{k \rightarrow \infty} \Pi x_{q_k}.
\end{equation}

We now concentrate on $\Pi^\bot x_{q_k}$. Note that $\Pi^\bot \psi_{s,t}(u) = \i \Pi^\bot \Im \psi_{s,t}(u)$. By \eqref{eq:realpart} and \eqref{eq:Pixconvergence} also
\begin{equation} \label{eq:xnormaldivergence}
 \lim_{k \rightarrow \infty} \e^{\scalarprod{\Pi^\bot \psi_{q_k,t}(u), \Pi^\bot x_{q_k}}} \in \C \setminus \{ 0 \}
\end{equation}
for almost all $(t,u) \in \mathcal{I}_s \times K$. 
Assume that
$$\limsup_{k \rightarrow \infty} \Vert \Pi^\bot x_{q_k} \Vert = \infty.$$
Then we are in the situation of the Lemma \ref{lem:divergencedirections}. There exist directions $g_i$ and a subsequence of $q_k$, again denoted by $q_k$, such that $x_{q_k} - \sum_i \scalarprod{g_i,x_{q_k}} g_i $ converges. For each $t \in \mathcal{I}_s$ after passing to a subsequence such that $q_k < t$ set $R(t) = \inf_{k} R(q_k,t)$, which by the continuity and strict positivity of $R$ is strictly greater than $0$. Thus for $u \in B^{\i V}_{R(t)} \cap K$ and $k$ large
$$\scalarprod{\Pi^\bot \psi_{q_k,t}(u),g_i} = \scalarprod{\lambda_i(q_k,t),u}.$$
Using this we can write the exponent in equation \eqref{eq:xnormaldivergence} (along a subsequence) as
$${\scalarprod{\Pi^\bot \psi_{q_k,t}(u), x_{q_k} - \sum_i \scalarprod{g_i,x_{q_k}}g_i } + \sum_i \scalarprod{g_i,x_{q_k}} \scalarprod{\lambda_i(q_k,t) ,u}}.$$

By Lemma \ref{lem:affineproperties3} iii), $\psi_{s,s}(u) = u$. Therefore we can find $u$ and $g_i$, such that $\scalarprod{\Pi^\bot \psi_{s,s}(u),g_i} \neq 0$.
By the continuity of $\psi$ we can find $t \in \mathcal{I}_s$, $N \in \mathbb{N}$ and an (open) set of positive measure $O \subset B^{\i V}_{R(t)} \cap K$ such that for all 
$k \geq N$, $u \in O$
\begin{equation} \label{eq:lambdadivergence}
 \scalarprod{\lambda_i(q_k,t),u} = \scalarprod{\Pi^\bot \psi_{q_k,t}(u),g_i} \neq 0.
\end{equation}
Furthermore $O$ can be chosen in a way such that
\begin{equation} \label{eq:riemanlebesguedivergent}
\lim_{k \rightarrow \infty} \int_{\i V} \mathbb{I}_O(u) \e^{\scalarprod{\Pi^\bot \psi_{q_k,t}(u), x_{q_k} - \sum_i \scalarprod{g_i,x_{q_k}} g_i}} \e^{ \scalarprod{\sum_i \scalarprod{g_i,x_{q_k}} \lambda_i(q_k,t) ,u}}   \dd{u}   \neq 0.
\end{equation}
$f_k(u) := \mathbb{I}_O(u) \e^{\scalarprod{\Pi^\bot \psi_{q_k,t}(u), x_{q_k} - \sum_i \scalarprod{g_i,x_{q_k}} g_i }}$ converges uniformly to a bounded function $f(u)$, while $\sum_i   \scalarprod{g_i,x_{q_k}} \lambda_i(q_k,t)$ diverges by \eqref{eq:lambdadivergence} and Lemma \ref{lem:divergencedirections}. By the Riemann-Lebesgue Lemma the limit in \eqref{eq:riemanlebesguedivergent} is zero, which gives a contradiction. We conclude that
$$\limsup_{k \rightarrow \infty} \Vert \Pi^\bot x_{q_k} \Vert < \infty.$$

An application of Cauchy-Schwarz gives
$$\vert \scalarprod{\Pi^\bot \psi_{q_k,t}(u), \Pi^\bot x_{q_k}} \vert = \vert \i \scalarprod{\Im \Pi^\bot \psi_{q_k,t}(u), \Pi^\bot x_{q_k}}  \vert \leq \Vert \Im \Pi^\bot \psi_{q_k,t}(u) \Vert  \Vert \Pi^\bot x_{q_k} \Vert.$$
For large enough $k$
$\Vert \Pi^\bot x_{q_k} \Vert$ is bounded by $\limsup_{k \rightarrow \infty} \Vert \Pi^\bot x_{q_k} \Vert + 1$. Furthermore,  
$$\Vert \Im \Pi^\bot \psi_{q_k,t}(u) \Vert \leq \Vert \Im \Pi^\bot (\psi_{q_k,t}(u) - u)\Vert +  \Vert \Im \Pi^\bot (u) \Vert, $$
where using the continuity of $\psi$ the two terms can be made arbitrarily small by choosing $u$ small, $k$ large and $t$ close to $s$.
I.e. there is $k, t, u$ so that $-\pi < \Im \scalarprod{\Pi^\bot \psi_{q_k,t}(u), \Pi^\bot x_{q_k}} < \pi$. 
Since \eqref{eq:xnormaldivergence} converges finitely and non-zero valued, 
$$\lim_{k \rightarrow \infty} \scalarprod{\Pi^\bot \psi_{q_k,t}(u), \Pi^\bot x_{q_k}} \in \i \R.$$
Using this for $d-p$ properly chosen linearly independent vectors $u_{p+1}, \dots, u_d$ from \eqref{eq:linindu} we conclude
$$ \lim_{k \rightarrow \infty} \Pi^\bot x_{q_k}$$
exists finitely valued. 
Together with \eqref{eq:realpart} and since $E$ is closed this proves i). 

To prove ii) note that 
$\psi_{q_k,t}(u)$ is bounded on $\mathcal{I}_s \times K$ for k large, hence the assumption of ii) can only hold if $\lim_{k \rightarrow \infty} \Vert x_{q_k} \Vert = \infty$. 
If $x_{q_k}$ is $E$-valued, furthermore
\begin{equation} \label{eq:uexpldivergence}
\lim_{k \rightarrow \infty} \scalarprod{\Re \psi_{q_k,t}(u),x_{q_k}} = - \infty
\end{equation}
 and hence $ 
\Re \psi_{q_k,t}(u) \neq 0$ for a subsequence $q_k$. Since $\Phi_{q_k,t}(u) > 0$ for $k$ large, we have $u \in \mathcal{Q}$ and $\psi_{q_k,t}(u) \in \mathcal{U}$. 
Hence $\Re \mathcal{U} \neq \emptyset$ and since $\Re \mathcal{U} \subset \overline{\mathrm{relint}(\Re \mathcal{U})}$, $\mathrm{relint}(\Re \mathcal{U})$ cannot be empty. 

By \eqref{eq:uexpldivergence} it follows that $\lim_{k \rightarrow \infty} \Vert \Pi x_{q_k} \Vert = \infty.$ There is a subsequence and some direction $g \in \mathrm{span}(\Re \mathcal{U})$ such  $\lim_{k \rightarrow \infty} \scalarprod{g,x_{q_k}} = \infty$. For $u \in \mathrm{relint}(\Re \mathcal{U})$ there is $\epsilon > 0$ such that $u+\epsilon g \in \mathrm{relint}(\Re \mathcal{U})$ which, by the definition of $\mathcal{U}$, implies $\sup_{x \in E} \scalarprod{u + \epsilon g, x} < \infty$. Hence $\lim_{k \rightarrow \infty}  \scalarprod{u,x_{q_k}} = - \infty$, which proves the remaining parts of ii).
\end{proof}

\section{Regularity and semimartingale characteristics} \label{sec:regularity}
Remember that the state space $E$ is a subset of the real vector space $V$. A $E$-valued\footnote{Semimartingales are strictly speaking only defined for $\R^d$-valued processes. However, for $V$ with dimension $d$ we fix an isomorphism between $\R^d$ and $V$ and in this case the semimartingale property is to be understood for the corresponding process 
on $\R^d$ defined via this isomorphism.} càdlàg process $X$ on a filtered probability space $(\Omega,\mathcal{A},\F,\PM)$ is a semimartingale if it can be decomposed as a $(\F,\PM)$-local martingale $M$ and a finite variation process $A$, i.e. $X = X_0 + M + A$. For an introduction to semimartingales (and semimartingale characteristics) we refer to \citet{JS03}.

Let $X$ be a semimartingale. The jump measure of $X$ is defined as
$$\mu^X(\dd{t},\dd{x};\omega) = \sum_{s} \1{\Delta X_s(\omega) \neq 0} \delta_{(s,\Delta X_s(\omega))}(\dd{t},\dd{x}).$$
Note that here $\Delta X_t = X_t - X_{t-}$ refers to the jump of the process and not the cemetery state. Since both notations are well established in their respective field, we will use $\Delta$ for both. There should be no risk of confusing this. Denote the compensator of this random measure by $\nu(\dd{t},\dd{x};\omega)$. 
 Let $h$ denote a truncation function for the semimartingale characteristics, i.e. a bounded function $h: V \rightarrow V$ with $h(x) = x$ in a neighborhood of $0$. The process 
$$\hat X_t = X_t - \sum_{s \leq t} (\Delta X_s - h(\Delta X_s))$$
then has bounded jumps and is a special semimartingale with unique canonical decomposition $\hat X = X_0 + \hat M + \hat B$, where $\hat M$ is a local martingale and $\hat B$ is a predictable finite variation process. Furthermore consider $C = [X^c] = [\hat M^c]$, 
which is the quadratic variation of the continuous local martingale part of the semimartingale. The triplet $(\hat B,C,\nu)$ is called a version of the semimartingale characteristics of $X$. Note that the characteristic $\hat B$ depends on the choice of the function $h$. They are defined up to a set of probability zero. We also call $(\tilde B, \tilde C, \tilde \nu)$ a version of the characteristics if they are equal to $(\hat B,C,\nu)$ up to a set of probability zero. 

The semimartingale property and the above decompositions depend on the filtration and the probability measure. For a Markov process $(X,\F,\{ \PM^{(s,x)} \})$ there are multiple measures and we make the following definition. 

\begin{definition} \label{def:markovsemimartingale}
An $E$-valued càdlàg Markov process $(X,\F,\{\PM^{(s,x)} \})$ is a Markov semi\-martingale if $X$ is a $(\F,\PM^{(s,x)})$-semimartingale for all $(s,x)$. 
\end{definition}
Note that we restrict ourselves to $E$-valued Markov processes. This means that we only work with conservative affine processes (see Lemma \ref{lem:conservative}). One approach to include a cemetery state $\Delta$ is to consider $\Delta$ as a point in $V \setminus E$ (see \citet{CFY05}). However, this causes several difficulties (compare footnotes \ref{footnote:explosion} and \ref{footnote:cemetery}) and we therefore exclude this case. 
Definition \ref{def:markovsemimartingale} only requires that the semimartingale property holds for all probability measures $\PM^{(s,x)}$. For time-homogeneous Markov processes it is shown in \citet{CJ80} that then there also exist processes which are versions of the the semimartingale characteristics for all measures $\PM^{x}$.

For time-homogeneous affine processes there is the following result\footnote{Here this result is formulated for conservative transition functions.} by \citet{CT13} and \citet{KST11b}. Denote by $S(V)$ the positive semidefinite matrices on $V$ and by $\mathcal{M}(V)$ the set of (signed) measures $\mu$ on $V$.  Let $h$ denote a truncation function for the semimartingale characteristics, i.e. a bounded function $h: V \rightarrow V$ with $h(x) = x$ in a neighborhood of $0$. 
\begin{theorem}
 \label{thm:homaffine}
Let $(X,\F,\{\PM^{x} \})$ be the $E$-valued càdlàg canonical realization of a conservative time-homo\-geneous continuously affine transition function. Then there exist $b \in V$, $a \in S(V)$, $m \in \mathcal{M}(V)$ and restrictions of linear maps $\beta: E \rightarrow V$, $\alpha: E \rightarrow S(V)$ and $M: E \rightarrow \mathcal{M}(V)$, such that for all $u \in \mathcal{U}$
\begin{equation*}
F(u) := \left. \frac{\partial}{\partial t} \phi_t(u) \right \vert_{t=0^+}, \qquad R(u) := \left. \frac{\partial}{\partial t} \psi_t(u) \right \vert_{t=0^+}
\end{equation*}
exist, i.e. $X$ is regular. $F$ and $R$ are given by
\begin{equation} \label{eq:Riccatiaffine}
\begin{aligned}
F(u) & = \frac{1}{2} \scalarprod{u,a u} + \scalarprod{b,u}  +\int_{V} \left( \e^{\scalarprod{\xi,u}}-1-\scalarprod{h(\xi),u} \right) m(\dd{\xi}), \\
 \scalarprod{R(u),x} & = \frac{1}{2} \scalarprod{u,\alpha(x) u} + \scalarprod{\beta(x),u} +\int_{V} \left( \e^{\scalarprod{\xi,u}}-1-\scalarprod{h(\xi),u}\right) M(\dd{\xi}; x).
\end{aligned}
\end{equation}
$\Phi$ and $\psi$ satisfy generalized Riccati differential equations for $t  < \sigma(u) = \inf \{s \geq 0: \Phi_{s}(u) = 0\}$, i.e.
\begin{equation} \label{eq:Riccati}
\begin{aligned}
 \frac{\partial}{\partial t} \Phi_t(u) & = \Phi_t(u)  F(\psi_t(u)), \quad & \Phi_0(u) = 1, \\
 \frac{\partial}{\partial t} \psi_t(u) & = R(\psi_t(u)),  & \psi_0(u) = u. 
\end{aligned}
\end{equation}
Furthermore $X$ is a Markov semimartingale, where $(B,C,\nu)$ given by
\begin{align*}
& B_t = \int_0^{t} b+ B(X_{v-}) \dd{v}, \\
& C_t = \int_0^{t} a+ A(X_{v-}) \dd{v}, \\
& \nu(\dd{t}, \dd{\xi}) =  (m(\dd{\xi}) + M(\dd{\xi};X_{t-})) \dd{t},
\end{align*}
are a version of the characteristics of $X$ under each measure $\PM^x$. 
\end{theorem}
\begin{remark} 
Note that on general state spaces $E$ we cannot say more about the parameters $b$, $\beta$, $a$, $\alpha$, $m$, $M$. For a given state space these parameters have to satisfy additional conditions, which guarantee that the process stays within the state space. For the canonical state space $\R^m_{\geq 0} \times \R^n$ such admissibility conditions have been formulated in \citet{DFS03}. Using these admissibility conditions they give a characterization of time-homogeneous continuously affine processes on the canonical state space. Admissibility conditions have also also been derived for matrix-valued affine processes in \citet{CU11}. 
\end{remark}
We want to extend the results of Theorem \ref{thm:homaffine} to time-inhomogeneous affine processes. The following examples show that this cannot be achieved in full generality. There are processes with a continuously affine transition function, which are not semimartingales and there are time-inhomogeneous Markov semimartingales, which are not regular. 
Some of the examples are presented with Markov processes in the sense of Definition \ref{def:markov1}, since the essential problems are easier to spot. However, given the transition function one can always construct a Markov process in the sense of Definition \ref{def:markov2inhom}, which then yields a counterexample in this class of Markov processes. 
\begin{example} \label{ex:detnotsemimartingale}
Every deterministic real-valued function $f$ generates a time-inhomo\-geneous affine process via the transition function 
$$P_{s,t}(x,\dd{\xi}) = \delta_{x+f(t)-f(s)}(\dd{\xi}).$$
By Lemma \ref{lem:markovconstr} there exist measures $\PM^{(s,x)}$ such that the canonical realization of the transition function, $(X,\F^X,\{\PM^{(s,x)} \})$,
is affine with
\begin{equation}
\begin{aligned}
\Phi_{s,t}(u) & = \e^{u (f(t)-f(s))}, \\
\psi_{s,t}(u) & = u.
\end{aligned}
\end{equation} 
In this case the process $X$ is $\PM^{(s,x)}$-a.s. described by $$X_t^{s,x} = x + f(s+t) - f(s), \quad t \geq 0.$$
If the function $f$ is continuous, the affine transition function is continuous. If additionally $f$ is not of finite variation, the Markov process is not a semimartingale (Proposition I.4.28, \citet{JS03}). 
\end{example}
A non-deterministic example is the following. 
\begin{example}
Let $X$ be an affine process on $(\Omega,\mathcal{A},\F,\PM)$ with $E=[a,\infty)$, $a >0$ (e.g. a shifted CIR process) and $f: \R_{\geq 0} \rightarrow [1,\infty)$ a continuous function, which is not of finite variation. Consider the process $Y_t = f(t) X_t$, which then has the same state. Then for $A \in \mathcal{B}(E)$,
$Y_t^{-1}(A)  = X_t^{-1}(f(t)^{-1} A) \in \F_t$
and
\begin{align*} \EV{\e^{\scalarprod{u,Y_t}} \vert \F_s} =  \EV{\e^{\scalarprod{f(t) u,X_t}} \vert \F_s} = 
\Phi_{s,t}(f(t) u) \e^{\scalarprod{\psi_{s,t}(f(t) u), X_s}}.
\end{align*}
Hence $Y$ is an affine process with
$$\Phi_{s,t}^Y(u) = \Phi_{s,t}^X(f(t) u), \qquad \psi_{s,t}^Y(u) = \frac{\psi_{s,t}(f(t) u)}{f(s)}.$$
The transition operator of $Y$ is defined by
$$P_{s,t}^Y g(y) = \int_{E} g(f(t) \xi) P_{s,t}^X(f(s)^{-1}y,\dd{\xi}).$$
If $X$ is a semimartingale, by It\^o's Lemma $X^{-1}$ is also a semimartingale. Assume that $Y$ is a semimartingale. Then also $f(t) = Y_t X_t^{-1}$ is a semimartingale. Since $f$ is not of finite variation this cannot be true. Hence $Y$ is not a semimartingale. 
\end{example}
Moreover, time-inhomogeneous continuously affine transition functions may not be regular. Here regularity is defined as follows. 
\begin{definition}
A transition function $\{ P_{s,t} \}$ on $E$ is regular\footnote{This is called weakly regular in Definition 2.3 of \citet{FI05}.} if $$\partial_s^- P_{s,t} \e^{\scalarprod{u,x}} \big \vert_{s=t} = - \lim_{s \uparrow t} \frac{1}{t-s} \left( P_{s,t} f_u(x) - f_u(x)\right)$$
exists for all $(t,x,u) \in \R_{> 0} \times E \times \mathcal{U}$ and is continuous at $u=0$ for all $(t,x) \in \R_{> 0} \times E$. 
\end{definition}
\begin{remark}
For affine transition function an application of the chain rule together with Assumption \ref{ass:affinespan} gives that this is equivalent to the existence of the left derivatives of $\Phi$ and $\psi$, i.e.
$\partial_s^- \Phi_{s,t}(u) \vert_{s=t} $,  $\partial_s^- \psi_{s,t}(u) \vert_{s=t}$. 
\end{remark}
\begin{example} \label{ex:timechange}
Let $f: \R_{\geq 0} \rightarrow \R_{\geq 0}$ be a continuous, nondecreasing function and $X$ a time-homogeneous continuously affine process on $(\Omega,\mathcal{A},\mathcal{F},\PM)$. Then the time-transformed process $Y_t = X_{f(t)}$ satisfies
$$\EV{e^{\scalarprod{u, Y_T}} \vert \F_{f(t)}} =\Phi_{f(T)-f(t)}(u)  \e^{\scalarprod{\psi_{f(T)-f(t)}(u),Y_t}}.$$
This corresponds to the transition function
$$P_{s,t}^Y(x,\cdot) = P_{f(t)-f(s)}^X(x,\cdot).$$
$\{ P_{s,t}^Y \}$ is again affine with $\Phi^Y_{s,t}(u) = \Phi^X_{f(t)-f(s)}(u)$  and $\psi^Y_{s,t}(u) = \psi^X_{f(t)-f(s)}(u)$. 
Since $\Phi^X$ and $\psi^X$ are continuous, $\Phi^Y$  and $\psi^Y$ are also continuous. 

The Cantor function is a continuous nondecreasing function $\tilde f: [0,1] \rightarrow [0,1]$ such that the derivative of $\tilde f$ is zero almost everywhere and does not exist otherwise (this function is therefore not absolutely continuous). Set $f = \tilde f$ on $[0,1]$ and $f(t) = 1$ for $t \geq 1$.  With this $f$ the transition function $\{ P_{s,t}^Y \}$ is not regular. A concrete example is $X_t = B_{f(t)}$, where $B$ is a Brownian motion. 
\end{example}

The final example shows that regularity itself does not imply that $X$ is a semimartingale. Hence the semimartingale property and regularity have to be treated separately. 
\begin{example}
Consider the function 
$$g(t) = \begin{cases} \frac{t \sin(1/t)}{\logn{t/2}} & 0  < t \leq 1, \\0 & t \leq 0. \end{cases}$$
$g$ is differentiable on $[0,1]$, but has unbounded variation (see \citet{GU12}, p.294). Note that the derivative of this function diverges and oscillates for $t \rightarrow 0$. 
Using $f(t) = g(1-t)$ in Example \ref{ex:detnotsemimartingale} yields an affine process, which is regular, but not a semimartingale. 
\end{example}
\subsection{Affine Markov semimartingales}
Let $(X,\F,\{\PM^{(s,x)} \})$ be a Markov semimartingale in the sense of Definition \ref{def:markovsemimartingale} with a continuously affine transition function. 
So far we only required that the semimartingale property holds for all probability measures $\PM^{(s,x)}$. 
To find processes which are versions of the characteristics under each measure $\PM^{(s,x)}$ we use the setup of \citet{CJ80}. Hence we work with the space-time process $\tilde X = (\Theta,X)$ introduced in Lemma \ref{lem:affinespacetime}. Note that with a continuously affine transition function 
we can use the product $\sigma$-algebras $\mathcal{\tilde E} := \mathcal{B}(\R_{\geq 0}) \otimes \mathcal{E}$ on $\tilde E := \R_{\geq 0} \times E$ and $\mathcal{\tilde A} := \mathcal{B}(\R_{\geq 0}) \otimes \mathcal{A}$ on $\tilde \Omega := \R_{\geq 0} \times \Omega$. The filtration $\mathcal{\tilde F}$ is given by $\mathcal{\tilde F}_t = \mathcal{B}(\R_{\geq 0}) \otimes \mathcal{F}_t$. Then $(\tilde X,\mathcal{\tilde F},\{ \tilde \PM^{(s,x)} \})$ is a time-homogeneous Markov process on $(\tilde{\Omega},\mathcal{\tilde A})$ with state space $(\tilde E,\mathcal{\tilde E})$. 
Since on the extended space $(\tilde{\Omega},\mathcal{\tilde A})$ $X(s,\omega) = X(\omega)$, we use the same letter $X$ to denote the process on $(\Omega,\mathcal{A})$ and on $(\tilde{\Omega},\mathcal{\tilde A})$. If unclear we specifically refer to the underlying probability space. 

We assume that the filtration $\mathcal{\tilde F}$ is a right-continuous strong Markov filtration (see \citet{CJ80}).  
Note that the definition of a strong Markov filtration requires the existence of shift operators $\theta_t$ for $X$ and $\Omega$, which are extended to $\tilde X$ and $\tilde \Omega$ by $\tilde \theta_t(s,\omega) = (s+t,\theta_t(\omega))$. For the canonical space-time realization of the continuously affine transition function this is always true (see section \ref{sec:cadlagversion}). 
With this setup we can use the results of \citet{CJ80}.

\begin{definition}
A stochastic process $F$ is called additive, if $\tilde \PM^{(s,x)}$-a.s. for all $(s,x)$ $F_0 = 0$ and
$$F_{v+t}(r,\omega) = F_v(r,\omega) + F_t \circ \tilde \theta_v(r,\omega) =  F_v(r,\omega) + F_t (r+v,\theta_v(\omega)).$$
\end{definition}
\begin{lemma} \label{lem:CJ80} 
There exist
\begin{enumerate}
\item an adapted additive process $F(r,\omega)$, which is nondecreasing from $F_0=0$ a.s. for each measure $\tilde \PM^{(s,x)}$, such that $F$ is $\tilde \PM^{(s,x)}$-indistinguishable from a predictable process for all $(s,x)$. 
\item Optional processes $b_t(r,\omega)$ and $c_t(r,\omega)$ with values in $V$, respectively $S(V)$,
\item a positive kernel $K_t(\dd{y};(r,\omega))$ from $(\tilde \Omega, \mathcal{O}(\mathcal{\tilde F}_t))$\footnote{$\mathcal{O}(\mathcal{\tilde F}_t)$ denotes the $\mathcal{\tilde F}_t$-optional $\sigma$-algebra (see \citet{CJ80}).} giving measures in $\mathcal{M}(V)$,
\end{enumerate}
such that
\begin{equation} \label{eq:levykhintchine1}
B = b \cdot F, \quad C = c \cdot F, \quad \nu(\dd{t},\dd{\xi};(r,\omega)) = K_t(\dd{\xi};(r,\omega)) \dd{F}_{t}(r,\omega)
\end{equation}
are a version of the semimartingale characteristics of $X$ on $(\tilde{\Omega},\mathcal{\tilde A})$ under each measure $\tilde \PM^{(s,x)}$. 
We say that $(b,c,K,F)$ give a version of the semimartingale characteristics of $(X,\mathcal{\tilde F},\{ \tilde \PM^{(s,x)} \})$.
\end{lemma}
\begin{proof}
If $(X,\F,\{\PM^{(s,x)} \})$ is a Markov semimartingale, also $(X,\mathcal{\tilde F}, \{ \tilde \PM^{(s,x)}\})$ is a Markov semimartingale. Furthermore $X_t -  X_0$ is an additive process.
Applying Theorem 6.25 in \citet{CJ80} to $X -  X_0$ gives the result.\footnote{\label{footnote:explosion}
Note that $X^\tau-X_0$ is no longer an additive process for general stopping times $\tau$, which is why we cannot directly use Theorem 6.25 for exploding processes by stopping the process before it explodes. }
\end{proof}

The function $F$ from Lemma \ref{lem:CJ80} is far from uniquely defined. Motivated by Example \ref{ex:timechange} we expect that the function $F$ from Lemma \ref{lem:CJ80} can even be chosen deterministic (by deterministic we here always mean independent of $\omega$). The idea is to construct such a candidate and show that $F$ is `absolutely continuous' with respect to this candidate. We then go on and show that when $F$ is deterministic, $b$, $c$ and $K$ can be expressed as functions of $(\Theta,X)$, which are affine in $X$. Note that we cannot expect that this is true for a non-deterministic $F$. 
\begin{example}
Assume $F$ to be independent of $\omega$ and choose a bounded function $f: E \rightarrow \R_{>0}$. Define $\tilde{F}_t = \int_0^t f(\tilde X_s) \dd{F}_s$. Then $\tilde{F}$ is additive, nondecreasing and $\tilde{F}_0 = 0$, but $\tilde{F}$ depends on $\omega$. The characteristics of $\tilde X$ can also be written as integrals with respect to any such $\tilde F$. 
\end{example}

\begin{definition} \label{def:phipsifinitevar}
We call an affine transition function of finite variation (FV) if the maps $t \mapsto \Phi_{t,T}(u)$ and $t \mapsto \scalarprod{\psi_{t,T}(u),x}$ for $x \in E$ are of finite variation on $[v,T]$ for every $v > o(T,u)$.
\end{definition}
\begin{remark}
Let $\{ P_{s,t} \}$ be a continuously affine transition function of a Markov semimartingale. We conjecture that the affine transition function is automatically of FV (see also Lemma \ref{lem:phipsifinvar}). However, so far we were not able to prove this. 
\end{remark}

For a Markov semimartingale with an affine transition function of finite variation we construct a deterministic candidate for $F$ as follows. Fix an affine basis of $E$, which exists by Assumption \ref{ass:affinespan}, and set $\psi_{t,T}^i(u) := \scalarprod{\psi_{t,T}(u), x^i - x^0}$. Set $\tilde o(T,u) = \frac{1}{2} (o(T,u)+T)$ and define the nondecreasing continuous functions
$$G^{T,u}(t) := \int_{t \land \tilde o(T,u)}^{t \land T}  \dd{\vert \Phi_{s,T}(u) \vert }+ \sum_{i=1}^d   \dd{\vert \psi_{s,T}^i (u) \vert}, $$
where we integrate with respect to $s$ and $\dd{ \vert f(s) \vert}$ denotes the total variation of $f$.
Consider a countable dense subset of $\R_{\geq 0} \times \i V$ and denote the indexed elements of this set by
$(T_i,u_i), i \in \mathbb{N}$. From now on when writing indexed $T_i$ or $u_i$, this always refers to points in this set. Define the weights $$w_i= 2^{-i} \frac{1}{G^{T_i,u_i}(T^i)} \1{G^{T_i,u_i}(T^i) > 0}$$ and set
\begin{equation} G(t):= t + \sum_{i \in \mathbb{N}} w_i G^{T_i,u_i}(t). \label{eq:Gfunc}
\end{equation}
The weights guarantee that this sum converges. The function $G$ is continuous, strictly increasing and $\vert G(t) \vert \leq 1+t$. 
We use the function $G$ to define the process $G_t(r,\omega) := G_t(r) := G(r+t)-G(r)$. $G$ is independent of $\omega$ and additive, i.e.
$G_{s+t}(r) = G_s(r) + G_t(s+r) .$
This is the mentioned candidate. With this construction $\dd{\Phi_{\cdot,T_i}(u_i)}$ and $\dd{\psi_{\cdot,T_i}^j(u_i)}$, $j=1, \dots, d,$ are absolutely continuous w.r.t. $\dd{G}$ on $[\tilde o(T_i,u_i),T_i]$) for all $i$. Denote their (deterministic) densities with respect to $\dd{G}$ by $s \mapsto f^\Phi(s;T_i,u_i)$ and $s \mapsto f^\psi(s;T_i,u_i)$. 

Fix $(s,x)$ and for $T > s$ and $u \in \i V$ consider the uniformly bounded $(\mathcal{\tilde F},\tilde \PM^{(s,x)})$-martingales\footnote{ \label{footnote:cemetery} If one considers the cemetery state $\Delta$ as a point in $V \setminus E$, the second identity in equation \eqref{eq:semimartingaleMmartingale} no longer holds for the classical exponential function. One has to use a modified function $f$ which is equal to $\e^{\scalarprod{u,x}}$ on $\mathcal{U} \times E$ and satisfies $f(u,\Delta) = 0$. However in this case equation \eqref{eq:Mmartingale} has to be modified to correct the jumps to $\Delta$. 
}
\begin{equation} \label{eq:semimartingaleMmartingale}
M_{t}^{T-s,u} = \EV[(s,x)]{\mathrm{e}^{\scalarprod{u,X_{T-s}}} \vert \mathcal{\tilde F}_{t}} = \Phi_{s+t,T}(u) \mathrm{e}^{\scalarprod{\psi_{s+t,T}(u),X_{t}}}, \quad 0 \leq t \leq T-s.
\end{equation}
By Theorem II.34 \citet{JS03} and Lemma \ref{lem:CJ80} $X$ can be decomposed as
\begin{equation} \label{eq:Xdecomp}
X_v = x + \int_0^{v} b_t \dd{F_t} + \int_0^{v} \int_V (\xi - h(\xi)) \mu^{X}(\dd{t},\dd{\xi}) + N_v, 
\end{equation}
where $\mu^X$ is the random measure associated with the jumps of $X$ and 
$N$ is a $(\mathcal{\tilde F},\tilde \PM^{(s,x)})$-local martingale (dependences on $(r,\omega)$ are suppressed).
Integration by parts and It\^o's Lemma applied to $M^{T-s,u}_{t}$ together with the properties of the semimartingale characteristics give for $r \geq 0$ and $o(T,u) - s < r \leq v \leq T-s$
\begin{align*}
M_{v}^{T-s,u}  - &   M_{r }^{T-s,u}  = \int_r^v  M^{T-s,u}_{t-} \frac{ \dd{\Phi_{s+t,T}(u)}}{\Phi_{s+t,T}(u)} + \Phi_{s+t,T}(u) \dd{\e^{\scalarprod{\psi_{s+t,T}(u),X_t}}} \\
= &  \int_r^v M^{T-s,u}_{t-} \left( \frac{\dd{\Phi_{s+t,T}(u)}}{\Phi_{s+t,T}(u)} + \dd{\scalarprod{\psi_{s+t,T}(u),X_{t}}} + \dd{[\scalarprod{\psi_{s+t,T}(u),X_t}]^c} \right) \\
&  + \int_r^v \int_V M^{T-s,u}_{t-} \left( \e^{\scalarprod{\psi_{s+t,T}(u),\xi}}  - 1 -  \scalarprod{\psi_{s+t,T}(u),\xi} \right) \mu^{X}(\dd{t},\dd{\xi}) \\
= &  \int_r^v M^{T-s,u}_{t-} \bigg( \frac{\dd{\Phi_{s+t,T}(u)}}{\Phi_{s+t,T}(u)} +\scalarprod{ \dd{\psi_{s+t,T}(u)},X_{t-}} + \scalarprod{\psi_{s+t,T}(u),\dd{B_t}}  \\
& \qquad \qquad \qquad + \scalarprod{\psi_{s+t,T}(u),\dd{C_t}\psi_{s+t,T}(u)} +   \scalarprod{\psi_{s+t,T}(u),\dd{N_t}}    \bigg) \\
&  + \int_r^v \int_V M^{T-s,u}_{t-} \left( \e^{\scalarprod{\psi_{s+t.,T}(u),\xi}}  - 1 -  \scalarprod{\psi_{s+t,T}(u),h(\xi)} \right) \mu^{X}(\dd{t},\dd{\xi}).
\end{align*}
Compensating with $\nu$ and using \eqref{eq:levykhintchine1} gives
\begin{equation} \label{eq:Mmartingale}
\begin{aligned}
M_{v}^{T-s,u} &= M_{r }^{T-s,u}  + (\tilde N_v - \tilde N_r) \\
& + \int_{r }^{v } M_{t-}^{T-s,u} \bigg( \frac{\dd{\Phi_{s+t,T}(u)}}{\Phi_{s+t,T}(u)} + \scalarprod{\dd{\psi_{s+t,T}(u)},X_{t-}} + \kappa(\psi_{s+t,T}(u)) \dd{F}_{t} \bigg), 
\end{aligned}
\end{equation}
where integration is with respect to $t$, $\kappa$ is defined for $(r, \omega) \in \tilde \Omega$ as
$$\kappa_t(u;(r,\omega)) := \scalarprod{u,b_t(r,\omega)} +  \frac{1}{2} \scalarprod{u,c_{t}(r,\omega) u} + \int_V \left(\e^{\scalarprod{u,\xi}} - 1 - \scalarprod{u,h(\xi)} \right) K_{t}(\dd{\xi};(r,\omega))$$
and $\tilde N_r$ is a local martingale given by
\begin{align*}
\tilde N_r & = \int_0^{r} M_{t-}^{T-s,u} \scalarprod{\psi_{s+t,T}(u),\dd{N_t} } \\
& + \int_0^{r} M_{t-}^{T-s,u}  \int_V  \left(\e^{\scalarprod{\psi_{s+t,T}(u),\xi}} - 1 - \scalarprod{\psi_{s+t,T}(u),h(\xi)} \right) ( \mu^{X} - \nu)(\dd{t},\dd{\xi}) .
\end{align*}
Since $M_{t}^{T-s,u}$ is a martingale, the finite variation part of the right hand side of \eqref{eq:Mmartingale} has to vanish. 
For $(T,u) = (T_i,u_i)$ this gives $\tilde \PM^{(s,x)}$ a.s. for $t$ in $[\tilde o(T_i,u_i) -s,T_i-s]$
\begin{equation} \begin{aligned} \label{eq:abscontphipsi}
\frac{f^\Phi(s+t;T_i,u_i)}{\Phi_{s+t,T_i}(u_i)} + & \scalarprod{f^\psi(s+t;T_i,u_i),X_{t-}} \dd{G(s+t)} = - \kappa_{t}(\psi_{s+t,T_i}(u_i)) \dd{F_t}.
\end{aligned} \end{equation}
For each $i$ the set where equation \eqref{eq:abscontphipsi} does not hold, is a set of $\tilde \PM^{(s,x)}$-measure zero. Taking the countable union of these sets we find a set of $\tilde \PM^{(s,x)}$-measure zero, where on the complement \eqref{eq:abscontphipsi} holds for all $i$ simultaneously.

Equation \eqref{eq:abscontphipsi} almost gives absolute continuity of $\dd{F_{\cdot}(s,\omega)}$ with respect to $\dd{G_{\cdot}(s)}$. If $\kappa_t \neq 0$, we could `divide' by $\kappa_t$ in equation \eqref{eq:abscontphipsi} and would be mostly finished. The next lemmas show that we can essentially always find $(T_i,u_i)$ for which we can do this. 
\begin{lemma} \label{lem:Tuconv}
Consider the sets
\begin{align*}
\Gamma(s,\omega) := \{t > 0:  \ & \kappa_{t}(\psi_{s+t,T_i}(u_i);(s,\omega))  = 0 
 \text{ for all $i$ with }
 s+t \in  [\tilde o(T_i,u_i),T_i]
\}.
\end{align*}
On $\Gamma(s,\omega)$ it holds that $$b_{t}(s,\omega) = 0, \quad c_{t}(s,\omega) = 0, \quad K_{t}(\dd{\xi},(s,\omega)) = 0.$$ 
\end{lemma}
\begin{proof}
Fix $(s, \omega)$ and $t \in \Gamma(s,\omega)$, $u \in \i V \subset \mathcal{U}_1$. There is a sequence $(T_k,u_k)$ with $T_k \downarrow s+ t$, $u_k \rightarrow u$. By passing to a subsequence, again denoted by $(T_k,u_k)$, one may assume that $0 \leq o(T_k,u_k) < T_k \leq s+t+1$ and $o(T_k,u_k) \rightarrow v \geq 0$. The continuity of $\Phi$ yields $\Phi_{o(T_k,u_k),T_k}(u_k) \rightarrow \Phi_{v,s+t}(u)$. Hence $v=0$ or $\Phi_{v,s+t}(u) = 0$ by the definition of $o(T_k,u_k)$. In both cases $v \leq o(s+t,u) < s+t$. Set $\epsilon :=s + t - v > 0$. There is an $N$, such that for all $k \geq N$ $$o(T_k,u_k) \leq v + \frac{ \epsilon}{2} = s+ t - \frac{ \epsilon}{2}.$$
Choosing $N$ large enough that additionally $s+t \leq T_k \leq s + t + \frac{ \epsilon}{4},$ it follows that for all $k \geq N$
$$\tilde o(T_k,u_k) = \frac{T_k + o(T_k,u_k)}{2} \leq \frac{s+t + \frac{ \epsilon}{4} + s + t - \frac{ \epsilon}{2}}{2} = s+ t - \frac{ \epsilon}{8} < s+t.$$
Hence one obtains that $\kappa_t(\psi_{s+t,T_k}(u_k);(s,\omega)) = 0$ or all $k \geq N$. 
By continuity of $\psi$, $\psi_{s+t,T_k}(u_k) \rightarrow \psi_{s+t,s+t}(u) = u$ and since $u \mapsto \kappa_t(u;(s,\omega))$ is continuous for each $t$, $(s, \omega)$, it follows that $ \kappa_{t}(u;(s,\omega))  = 0$.
This holds for all  $u \in \i \R^d$ and by the uniqueness of the Lévy-Khintchine representation (Lemma II.2.44, \citet{JS03}) this transfers to $b,c,K$. 
\end{proof}
For each $(s,\omega)$ the measure $\dd{F}_{\cdot}(s,\omega)$ restricted to the complement of $\Gamma(s,\omega)$ is measure on $\R_{\geq 0}$. Denote by $\tilde F_t(s,\omega)$ its distribution function with $\tilde F_0(s,\omega) = 0$. By Lemma \ref{lem:Tuconv} and \eqref{eq:levykhintchine1} $(b,c,K,\tilde F)$ is again a version of the semimartingale characteristics of $X$. Define
$$\Gamma^* := \{(r,\omega) \in \tilde \Omega: \dd{\tilde F}_{\cdot}(r,\omega) \ll \dd{G_{{\cdot}}(r)} \}.$$
\begin{lemma} \label{lem:gammastar}
$\PM^{(s,x)}(\Gamma^*)=1$ for all $(s,x)$. 
\end{lemma}
\begin{proof}
Fix $(s,x)$. Since  $\tilde \PM^{(s,x)}(\{s\} \times \Omega) = 1$ it suffices to show $\tilde \PM^{(s,x)}(\Gamma^*_s) = 1$ with $\Gamma^*_s = \{(s,\omega): \omega \in \Omega,  \dd{\tilde F}_{\cdot}(s,\omega) \ll \dd{G_{{\cdot}}(s)} \}$. 
Using the associativity of integrals\footnote{See Proposition 0.4.10 \citet{RY10} or \citet{FT12} for a discussion on the rules for exchanging Stieltjes integrals.}
we can \lq divide\rq \   by $\kappa_{t}(\psi_{s+t,T_i}(u_i);(s,\omega))$ in \eqref{eq:abscontphipsi}, when it is not equal to zero. For each $i$ define $$\Lambda^i(s,\omega) := \{t \geq 0: \kappa_{t}(\psi_{s+t,T_i}(u_i);(s,\omega)) \neq 0 \text{ and } s+t \in [\tilde o(T_i,u_i),T_i] \}.$$
Fix $(s,x)$. Then \eqref{eq:abscontphipsi} together with the construction of $\tilde F$ yields $\tilde \PM^{(s,x)}$-a.s. for all $i \in \mathbb{N}$
\begin{equation} \mathbb{I}_{\Lambda^i(s,\omega)}(t) \dd{\tilde F_t(s,\omega)} = H_t^i(s, \omega) \dd{G_t(s)} \label{eq:Ftildeabscon}
\end{equation}
with some function $t \mapsto H_t^i(s,\omega)$. The weighted sum $\sum_i 2^{-i}  \mathbb{I}_{\Lambda^i(s,\omega)} \leq 1$ is strictly positive on the complement of $\Gamma(s,\omega)$. Together with equation \eqref{eq:Ftildeabscon} this then yields that $\tilde \PM^{(s,x)}$-a.s. on $\Gamma^*_s$ $ \dd{\tilde F}_{\cdot}(s,\omega) \ll \dd{G_{\cdot}(s)}$. 
\end{proof}
By Lemma \ref{lem:gammastar} also $(b,c,K,\bar F)$ with $\bar F = \tilde F \mathbb{I}_{\Gamma^*}$ give a version of the semimartingale characteristics of $X$ under each $\tilde \PM^{(s,x)}$. 
$\dd{\bar F_{\cdot}(r,\omega)} \ll \dd{G_{\cdot}(r)}$ for all $(r,\omega)$ and we denote the density by $t \mapsto f(t;r,\omega)$. 
Define 
\begin{equation} \label{eq:deterministiccharacteristics}
\text{ $\tilde b = f b$, $\tilde c = f c$ and $\tilde K_t(\dd{\xi},(r,\omega))  = K_t(\dd{\xi},(r,\omega)) f(t;r,\omega)$.}
\end{equation}
Then $(\tilde b, \tilde c, \tilde K, G)$ give the same characteristics as $(b,c,K,\bar F)$. 

$G$ is continuous and hence by Proposition II.2.9 i) in \citet{JS03} it follows that $X^\tau$ is quasi-left continuous under each measure $\PM^{(s,x)}$. Hence we can also apply Theorem 6.27 in \citet{CJ80} to get that there exist
\begin{enumerate}
\item an adapted nondecreasing continuous additive process $F$
\item $\mathcal{\tilde E}$-measurable functions $b$ and $c$ with values in $V$ and $S(V)$
\item a positive kernel $K(\dd{y};t,x)$ from $(\tilde E,\mathcal{\tilde E})$ giving measures in $\mathcal{M}(V)$,
\end{enumerate}
such that $(B,C,\nu)$ defined by
\begin{equation}  \label{eq:newchar}
\begin{aligned}
&  B_t(r,\omega) = \int_0^{t }  b(\Theta_v(r),X_{v-}(\omega)) \dd{F_v(r,\omega)}, \\
&  C_t(r,\omega) = \int_0^{t }  c(\Theta_v(r),X_{v-}(\omega)) \dd{F_v(r,\omega)}, \\
&  \nu(\dd{t},\dd{\xi};(r,\omega)) =  K(\dd{\xi};\Theta_t(r),X_{t-}(\omega)) \dd{F}_{t}(r,\omega)
\end{aligned}
\end{equation}
are a version of the characteristics of $X$ under each $\PM^{(s,x)}$. Construct $\bar F$ as before, so that $(b,c,K,\bar F)$ also give a version of the semimartingale characteristics and $\dd{\bar F_{\cdot}(r,\omega)} \ll \dd{G_{\cdot}(r)}$. $G$ is continuous and strictly increasing. We define a time-change
\begin{equation} \label{eq:dettimechange}
\eta_t(r) = \eta_t(r,\omega) :=  \inf \{ v \geq 0: G(v+r)-G(r) \geq t \}.
\end{equation}
Then for $\hat F_t(r,\omega) = \bar F_{\eta_t(r)}(r,\omega)$ we have $\dd{\hat F} \ll \dd{t}$ (note that $G_{\eta_s(r)} - G_0(r) = s$). The first part of the proof of Theorem 3.55 in \citet{CJ80} together with Proposition 3.56 applied to $\hat F$ shows that there exists a nonnegative $\mathcal{\tilde E}$-measurable function $h$, such that $$\bar F_t(r,\omega) = \int_0^t h(\Theta_v(r),X_{v-}(\omega)) \dd{G_v(r)}$$ $\tilde \PM^{(s,x)}$-a.s. under each $(s,x)$. 
Set $\bar b = b h$, $\bar c = c h$, $\bar K = K h$. Then $(\bar b, \bar c, \bar K, G)$ give a version of the characteristics of $X$ via $\eqref{eq:newchar}$ with $(b,c,K,F)$ replaced by $(\bar b, \bar c, \bar K, G)$. We can summarize this in the following lemma. 
\begin{lemma} \label{lem:deterministicA} A version of the characteristics of $(X,\mathcal{\tilde F},\{ \tilde \PM^{(s,x)} \})$ is given by $(\bar b,\bar c,\bar K,G)$ via equation \eqref{eq:newchar}, where
\begin{enumerate}
\item $G$ is an additive process given by 
$G_t(r,\omega) := G(r+t)-G(r)$ for a deterministic strictly increasing continuous function $G(t)$ with $G(0) = 0$. 
\item $\bar b(t,x)$ and $\bar c(t,x)$ are $\mathcal{\tilde E}$-measurable functions  with values in $V$ and $S(V)$,
\item $\bar K$ is a positive kernel $\bar K(\dd{y};t,x)$ 
from $(\tilde E,\mathcal{\tilde E})$
giving measures in $\mathcal{M}(V)$. 
\end{enumerate}
\end{lemma}

To show a theorem in the spirit of Theorem \ref{thm:homaffine} we use the affine basis from Assumption \ref{ass:affinespan}. To do this we consider independent copies of $X$ starting at different values. Define
\begin{itemize}
\item the probability space $\hat \Omega := \R_{\geq 0} \times \Omega^{d+1}$, 
\item the process $\hat X(\hat \omega) := (\Theta(s), X(\omega^0), \dots, X(\omega^d)), $ where $\hat \omega = (s,\omega^0, \dots, \omega^d)$,
\item the filtration $\mathcal{\hat F}$ defined by $\mathcal{\hat F}_t :=\mathcal{B}(\R_{\geq 0}) \times (\otimes_{i=0}^d \F_t)$,
\item for $s \geq 0$ and $x = (x^0, \dots, x^d) \in E^{d+1}$ the probability measures $$\hat \PM^{(s, x)} := \delta_s \otimes (\otimes_{i=0}^d \PM^{(s,x^i)}).$$
\end{itemize}
Then for each $s \geq 0$ the processes $\hat X^0, \dots, \hat X^d$ with $\hat X^i(\hat \omega) = X(\omega^i)$ are independent, $\hat \PM^{(s,x)}(\hat X^i_0 = x^i) = 1$ for $i = 0, \dots, d,$ and each component $\hat X^i$ behaves like the time-inhomo\-geneous affine process starting in $(s,x^i)$. I.e. each component $\hat X^i$ has semimartingale characteristics as described in Lemma \ref{lem:deterministicA} (with the same $G$) 
depending only on $(s,\omega^i)$, respectively $(\Theta_t(s), X_{t-}(\omega^i))$ (and not $\hat \omega$)

For the affine basis $(x^0,\dots,x^d)$ 
we can express $\hat X_t^i$ in terms of the basis $(x^1-x^0,\dots,x^d-x^0)$, i.e. we define processes $X_{t,j}^i$ by $\hat X_t^i = \sum_{j=1}^d X_{t,j}^i (x^j - x^0)$ and consider the matrix-valued stochastic process
\begin{align}
& H_t := H(\hat X^0_t,\dots, \hat X_t^d) := \mat{1 & X_{t,1}^0 & \hdots & X_{t,d}^0 \\ \vdots &\vdots &\ddots & \vdots \\ 1 & X_{t,1}^{d} & \hdots &  X_{t,d}^{d} },  \label{eq:Ht} \\
&  Y_t = \inf_{0 \leq s \leq t} \vert \det H_s \vert.
\end{align} 
Set $\bar x = (x^0, \dots, x^d)$. We have the following lemma (compare \citet{KM11}). 
\begin{lemma} \label{lem:linindX}
For each $s > 0$ there exists $\delta(s) > 0$ such that 
 $$\hat \PM^{(s,\bar x)}(\det H_t \neq 0 \text{ for all } 0 \leq t \leq \delta(s)) > 0.$$
\end{lemma}
\begin{proof}
Fix $s > 0$ and set $g(t) = \hat \PM^{(s,\bar x)}(Y_t > 0)$. We first show that $g(t)$ is left-continuous. Let $t_k \uparrow t$. $g(0) = 1$ and $g(t)$ is decreasing. Hence $g(t_k)$ converges. 
Since $X$ is stochastically continuous and Markov, it is also continuous in probability under $\PM^{(s, x)}$ and hence $Y_t$ as well. Then there is a subsequence $t_{k_m}$, such that $Y_{t_{k_m}} \downarrow Y_t$ a.s.. By dominated convergence also $g(t_{k_m}) \downarrow g(t)$. Hence $g(t_k)$ also converges to $g(t)$ and we get that $g(t)$ is left-continuous. Since $Y_t$ is a.s. right-continuous by dominated convergence this also holds for $g(t)$. Hence $g(t)$ is continuous in $t$ and we can define
$$\delta(s) := \inf \{t > 0:  \hat \PM^{(s,\bar x)}(Y_t > 0)  = \frac{1}{2} \} > 0.$$
\end{proof}

We conjecture that a uniform version of Lemma \ref{lem:linindX} is true for a continuously affine transition function, i.e. that for each $T > 0$ there exists $\delta > 0$ such that for all $0 \leq s \leq T$ $\hat \PM^{(s,\bar x)}(\det H_t \neq 0 \text{ for all } 0 \leq t \leq \delta) > 0.$ In this case the affine transition function is of finite variation.
\begin{lemma} \label{lem:phipsifinvar}
Assume that for each $T > 0$ there exists $\delta > 0$ such that for all $0 \leq s \leq T$ $$\hat \PM^{(s,\bar x)}(\det H_t \neq 0 \text{ for all } 0 \leq t \leq \delta) > 0.$$
Then the functions $t \mapsto \Phi_{t,T}(u)$ and $t \mapsto \scalarprod{\psi_{t,T}(u),x}$ are of finite variation on $[v,T]$ for every $v > o(T,u)$ and $x \in E$. 
\end{lemma}
\begin{proof}
Fix $T > 0$, $u \in \mathcal{U}$.
For $o(T,u) \leq s \leq T$, $0 \leq i \leq d$
consider the $(\hat \F, \hat \PM^{(s,\bar x)})	$-martingales
$$\hat M_t^{T-s,u,i} = \Phi_{s+t,T}(u) \e^{\scalarprod{\psi_{s+t,T}(u),\hat X_t^i}}, \quad 0 \leq t \leq T-s.$$
By assumption 
there is $\delta > 0$ and a set $\Lambda(s)$ with $\hat \PM^{(s,\bar x)}(\Lambda(s)) > 0$, such that $H_t$ is invertible on $[0, \delta]$. Let $\hat{\tau}$ be a stopping time such that $H_t$ is invertible on $[0,\hat{\tau}]$ and $\hat{\tau} \geq \delta$ on $\Lambda(s)$. We then have $\hat \PM^{(s,\bar x)}$-a.s. that 
\begin{equation} \label{eq:phipsifinvar}
\mat{\logn{\Phi_{s+t \land \hat{\tau},T}(u)} \\ \psi_{s+t \land \hat{\tau},T}^1(u) \\ \vdots \\ \psi_{s+t \land \hat{\tau},T}^d(u)} = \mat{1 & X_{t  \land \hat{\tau},1}^0 & \hdots & X_{t \land \hat{\tau},d}^0 \\ \vdots &\vdots &\ddots & \vdots \\ 1 & X_{t \land \hat{\tau},1}^{d} & \hdots &  X_{t \land \hat{\tau},d}^{d} }^{-1} \mat{ \logn{\hat M_{t \land \hat{\tau}}^{T-s,u,0}} \\ \vdots \\ \logn{\hat M_{t \land \hat{\tau}}^{T-s,u,d}}}.
\end{equation}
On $\Lambda(s)$ and $[0,\delta \land (T-s)]$ the left side is deterministic and coincides with a semimartingale. Hence $\Phi$ and $\psi$ are of finite variation on $[s,(s+\delta) \land T]$. Let $v > o(T,u)$. Then we can cover $[v,T]$ with the open intervals $(s,s+\delta)$, $s \in [v,T]$ and by compactness there is a finite subcover. Since $\Phi$ and $\psi$ are of finite variation on $[s,(s+\delta) \land T]$, they are of finite variation on $[v,T]$. 
\end{proof}

While it seems plausible that the assumption of Lemma \ref{lem:phipsifinvar}
is fulfilled for a continuously affine transition function, the following counterexample shows that this is not true for all affine transition functions.
\begin{example}
Construct a time-inhomogeneous Markov processes, which is $\PM^{(s,x)}$-a.s. equal to 
\begin{align*}
X_t^{s,x} := & \1{s < 1} ((1-s-t) B^x_{t}\1{s + t \leq 1} + (B^x_t - B^x_{1-s}) \1{s+t > 1}) \\ 
+ & \1{s \geq 1} B^x_{t},
\end{align*}
where $B^x_t$ is a Brownian motion starting in $x$. This is an affine process with
\begin{align*}
\psi_{s,t}(u) = \begin{cases} s \leq t < 1: & u \frac{1-t}{1-s} \\ s < 1 \leq t: & 0 \\ 1 \leq s \leq t: & u \end{cases} \qquad \qquad \phi_{s,t}(u) = \begin{cases} s \leq t < 1: & \frac{1}{2} u^2 (t-s) (1-t)^2 \\ s < 1 \leq t: & \frac{1}{2} u^2 (t-1) \\ 1 \leq s \leq t: & \frac{1}{2} u^2 (t-s) \end{cases}
\end{align*}
This process is not stochastically continuous in the sense of Definition \ref{def:stochcont}. In particular, for the sequence $(s,t)_n := (1-\frac{1}{n},1) \rightarrow (1,1)$ stochastic continuity (in $s$) fails. Moreover we cannot find a uniform $\delta$ on the interval $[0,1]$, since $\delta(s)$ from Lemma \ref{lem:linindX} in this case satisfies $\delta(s) \leq 1-s$ for $0 \leq s \leq 1$. 
\end{example}
We now formulate and prove the central theorem of this section. 
\begin{theorem} \label{thm:affinesemiinhom}
Let $(X,\mathcal{\F},\{ \PM^{(s,x)}\})$ be a Markov semimartingale whose transition function $\{ P_{s,t} \}$ is continuously affine of finite variation. 
Then there exist a deterministic $\R_{\geq 0}$-valued strictly increasing continuous function $G$, maps $ b, a, m$ from $\R_{\geq 0}$ to $V$, $S(V)$ and $\mathcal{M}(V)$ and maps $\beta: \R_{\geq 0} \times E \rightarrow V$, $\alpha: \R_{\geq 0} \times E \rightarrow S(V)$ and $M: \R_{\geq 0} \times E \rightarrow \mathcal{M}(V)$, which for $\dd{G}$-a.e. $t\geq0$ are restrictions of linear maps on $E$, such that for all $u \in \mathcal{U}$, $\Phi$ and $\psi$ satisfy generalized Riccati integral equations
\begin{equation} \label{eq:RiccatiinhomA}
\begin{aligned}
\Phi_{s,T}(u) & = 1 + \int_s^T \Phi_{t,T}(u) F(t,\psi_{t,T}(u)) \dd{G(t)}, \\
\psi_{s,T}(u) & = u + \int_s^T R(t,\psi_{t,T}(u)) \dd{G(t)},
\end{aligned}
\end{equation}
where
\begin{equation} \label{eq:FGinhom}
\begin{aligned}
F(t,u) & = \frac{1}{2} \scalarprod{u,a(t) u} + \scalarprod{b(t),u} +\int_{V} \left( \e^{\scalarprod{\xi,u}}-1-\scalarprod{h(\xi),u} \right) m(\dd{\xi}; t), \\
\scalarprod{R(t,u),x} & = \frac{1}{2} \scalarprod{u,\alpha(t,x) u} + \scalarprod{\beta(t,x),u} \\ & \quad +\int_{V} \left( \e^{\scalarprod{\xi,u}}-1-\scalarprod{h(\xi),u} \right) M(\dd{\xi}; t,x).
\end{aligned}
\end{equation}
If furthermore $\mathcal{\tilde \F}$ is a right-continuous strong Markov filtration on the extended probability space $(\tilde \Omega, \mathcal{\tilde A})$\footnote{Note that this holds for the canonical space-time realization of the Markov process $X$.}, then 
\begin{align}
& B_t = \int_0^{t} b(\Theta_v)+ \beta(\Theta_v,X_{v-}) \dd{G(\Theta_v)}, \notag \\
& C_t = \int_0^{t} a(\Theta_v) + \alpha(\Theta_v,X_{v-}) \dd{G(\Theta_v)}, \label{eq:inhomcharacteristics}\\
& \nu(\dd{t}, \dd{\xi}) = \left( m(\dd{\xi};\Theta_t) + M(\dd{\xi}; \Theta_t,X_{t-}) \right) \dd{t} \notag
\end{align}
are a version of the semimartingale characteristics of $X$ under each measure $\tilde \PM^{(s,x)}$. Here $\Theta_t(r,\omega) = t+r$. 
\end{theorem}
\begin{remark}
Equation \eqref{eq:RiccatiinhomA} reads in terms of $\phi$
$$ \phi_{s,T}(u) = \int_s^T F(t,\psi_{t,T}(u)) \dd{G(t)}.$$
\end{remark}
\begin{proof}
If $\mathcal{\tilde \F}$ is not a strong Markov filtration we consider instead of $X$ the canonical space-time realization of section \ref{sec:cadlagversion}. Note that if $(X,\F,\{ \PM^{(s,x)} \})$ is a Markov semimartingale, it is also a Markov semimartingale with respect to the smaller filtration generated by $X$ and then also the canonical space-time realization is a Markov semimartingale (we only add null sets to the filtration) with a right-continuous strong Markov filtration $\tilde \F$. Hence without loss of generality we may assume that $\mathcal{\tilde \F}$ is a strong right-continuous Markov filtration. 

Fix $u \in \mathcal{U}$ and $s \geq 0$ and consider the $(\hat \F,\hat \PM^{(s,\bar x)})$-martingales introduced in the proof of Lemma \ref{lem:phipsifinvar},
$$\hat M_t^{T-s,u,j} = \Phi_{s+t,T}(u) \e^{\scalarprod{\psi_{s+t,T}(u), \hat X_t^j}}, \quad 0 \leq t \leq T-s.$$

For $(T,u) = (T_i,u_i)$ and $j = 0, \dots, d$ we can repeat the steps leading to \eqref{eq:abscontphipsi} using the characteristics from Lemma \ref{lem:deterministicA}. This then gives that $\hat \PM^{(s,x)}$-a.s. on $\{s\} \times \hat \Omega$ (note that $\hat \PM^{(s,x)}(\{s\} \times \hat \Omega) = 1$) simultaneously for all $j=0, \dots d$ and all $i$
\begin{align*}
\frac{f^\Phi(s+t;T_i,u_i)}{\Phi_{s+t,T_i}(u_i)} & +  \scalarprod{f^\psi(s+t;T_i,u_i),X_{t-}(\omega^j)} \dd{G_t(s)} = \\
& - \bar \kappa(\psi_{s+t,T_i}(u_i);s+t,X_{t-}(\omega^j)) \dd{G_t(s)}
\end{align*}
on $[\tilde o(T_i,u_i)-s,T_i-s]$, where for $t \geq 0$, $u \in \mathcal{U}$, $x \in E$
$$\bar \kappa(u;t,x) := \scalarprod{u,\bar b(t,x)} +  \frac{1}{2} \scalarprod{u,\bar c(t,x) u} + \int_V \left(\e^{\scalarprod{u,\xi}} - 1 - \scalarprod{u,h(\xi)} \right) \bar K(\dd{\xi};t,x).$$
Formulated differently $\dd{G}_{\cdot}(s)$-a.e. on $([\tilde o(T_i,u_i),T_i]-s)$
\begin{equation}  \label{eq:abscontphipsi2}
\frac{f^\Phi(s+t;T_i,u_i)}{\Phi_{s+t,T_i}(u_i)} + \scalarprod{f^\psi(s+t;T_i,u_i),X_{t-}(\omega^j)} = - \bar \kappa(\psi_{s+t,T_i}(u_i);s+t,X_{t-}(\omega^j)).
\end{equation}
By the proof of Lemma \ref{lem:Tuconv} we can find sequences $(T_k,u_k)$ with $T_k \downarrow s+t$, $u_k \rightarrow u$, such that \eqref{eq:abscontphipsi2} holds with $(T_k,u_k)$ for all $k \geq N$. The right side then converges to $-\bar \kappa_{t}(u;s+t, X_{t-}(\omega^j))$ for each $j \in 0, \dots, d$.
Hence also the left side converges. By Lemma \ref{lem:linindX} there is $\delta(s) > 0$ and a set $\Lambda(s)$ with $\hat \PM^{(s,\bar x)}(\Lambda(s)) > 0$, such that $H_t$ as defined in equation \eqref{eq:Ht} is invertible on $\Lambda(s)$ for $t \in [0,\delta(s)]$. 
Hence the limits 
\begin{align*}
\tilde f^\Phi(s+t,u) & = \lim_{k \rightarrow \infty} f^\Phi(s+t,T_k,u_k), \\
\tilde f^\psi(s+t,u) & = \lim_{k \rightarrow \infty} f^\psi(s+t,T_k,u_k),
\end{align*}
exist. So $\tilde f^\Phi(t,u)$ and $\tilde f^\psi(t,u)$ are defined for $\dd{G}_{\cdot}(s)$-a.e. $t \in [s,\delta(s)]$. Since $s$ was arbitrary, this limits exist $\dd{G}$-a.e. on $\R_{\geq 0}$ and \eqref{eq:abscontphipsi2} yields
\begin{equation}
\tilde f^\Phi(s+t,u)  + \scalarprod{\tilde f^\psi(s+t,u) , X_{t-}(\omega^j)} = - \bar \kappa(u;s+t,X_{t-}(\omega^j))
\end{equation}
The left-hand side is affine in $X_{t-}(\omega^j)$, hence this is also true for right hand side. By the uniqueness of the Levy-Khintchine formulas (Lemma II.2.44 in \citet{JS03}) it then follows that $\bar b(t,x)$, $\bar c(t,x)$, $\bar K(\dd{\xi};t,x)$ are also affine in $x$ $\dd{G}$-a.e. and that the maps $a,b,\alpha,\beta,m,M$ as stated in the theorem are $\dd{G}$-a.e. given by 
\begin{align*}
\bar b(t,x) & = b(t) + \beta(t,x) \\
\bar c(t,x) & = a(t) + \alpha(t,x) \\
\bar K(\dd{\xi};t,x) & = m(\dd{\xi};t) + M(\dd{\xi}; t,x) .
\end{align*}
This proves the parts regarding the semimartingale characteristics. Finally we prove the generalized Riccati integral equations. Note that $\bar \kappa(u,t,x) = F(t,u) + \scalarprod{R(t,u),x}.$ Fix $(T,u)$ and $o(T,u) < s \leq T$. By considering the martingales $\hat M_t^{T-s,u,j}$ from equation \eqref{eq:semimartingaleMmartingale} and applying again integration by parts, the It\^o formula with the semimartingale characteristics of Lemma \ref{lem:deterministicA} one obtains
$$\frac{\dd{\Phi_{s+t,T}(u)}}{\Phi_{s+t,T}(u)} + \scalarprod{\dd{\psi_{s+t,T}(u)}, X_{t-}(\omega^j)} = - \bar \kappa(\psi_{s+t,T}(u);s+t,X_{t-}(\omega^j)) \dd{G(s+t)}$$
for $t \in [\tilde o(T,u)-s,T-s])$ and $j=0,\dots, d$. By Lemma \ref{lem:linindX} there is a $\delta(s) > 0$ and sets $\Lambda(s)$ with $\hat \PM^{(s,\bar x)}(\Lambda(s)) > 0$ such that $H_t$ is invertible on $\Lambda(s)$ for $t \in [0,\delta(s)]$. Hence for all $0 \leq r \leq \delta(s) \land (T-s)$,
the associativity of integrals and the fact that the integrands are affine in $X_{t-}(\omega^j)$ yields 
$\dd{\Phi_{\cdot,T}(u)} \ll  \dd{G(\cdot)}$ with density $t \mapsto - \Phi_{t,T}(u) F(t,\psi_{t,T}(u))$ and $\dd{\psi_{\cdot,T}(u)}  \ll  \dd{G(\cdot)}$ with density $t \mapsto -R(t,\psi_{t,T}(u))$ on $[s,(s+\delta(s)) \land T]$. Since $s$ was arbitrary, 
this holds on $[v,T]$ with $v > o(T,u)$ and the generalized Riccati integral equations in \eqref{eq:RiccatiinhomA} follow. 
\end{proof}
\begin{remark}
The semimartingale characteristics depend only on $(\Theta,X)$ with $\Theta$ depending on $s$ and $X$ depending on $\omega$. For fixed $s$ $\tilde \PM^{(s,x)}(\Theta_0 = s) = 1$ for all $x \in E$. Hence \eqref{eq:inhomcharacteristics} with $\Theta_t$ replaced by $s+t$ gives a version of the characteristics of $X$ under $\tilde \PM^{(s,x)}$ for $s$ fixed and each $x \in E$. These characteristics only depend on $\Omega$, so they are also characteristics of $X$ on $(\Omega,\mathcal{A})$ with the original filtration $\F_t$ under each probability measure $ \PM^{(s,x)}$, $x \in E$. 
\end{remark}

\begin{remark}
This theorem relates to the results for time-homogeneous affine processes in \citet{CT13} as follows. By using a full and complete class of functions it is shown in \citet{CT13} that the canonical realization of a time-homogeneous continuously affine transition function is always a semimartingale and that one can always choose $G(s) = s$. From there one could directly use the above proof.  The first part shows that the differentiated semimartingale characteristics only depend on $X_{s-}$ and are affine. The second part then gives the functions $F$ and $R$, which in this case also do not depend on time, so that we get generalized Riccati integral equations with $G(s) = s$. In this case the integrands are continuous, so we get differentiability and therefore regularity. 
\end{remark}
The following example shows that in the time-inhomogeneous case the local characteristics, although affine in $X$ may be very irregular with respect to the time parameter. This similarly holds for the functions $F$ and $R$ in equation \eqref{eq:FGinhom}.
\begin{example}  \label{ex:smithvolterracantor}
Consider the Smith-Volterra-Cantor set (or fat cantor set) on $[0,1]$ denoted by $A$. This set is generated iteratively as follows. In the first step the interval $(\frac{3}{8},\frac{5}{8})$ is removed from $[0,1]$. The iteration step removes from every remaining interval a centered open subinterval of one quarter of the length of this interval (see also Definition 2.1 in \citet{MU14}, where this is referred to as the SVC(4) set). The Smith-Volterra-Cantor set has Lebesgue measure $\frac{1}{2}$. 
We define an affine process as in Example \ref{ex:detnotsemimartingale} using the deterministic function
$f(t) = \int_0^{t} \mathbb{I}_A(v) \dd{v}.$
This gives a continuously affine process with $\psi_{s,t}(u) = u$ and $$\phi_{s,t}(u) = u \int_s^t \mathbb{I}_A(v) \dd{v} = 
u \lambda(A \cap [s,t]),$$
where $\lambda$ denotes the Lebesgue measure. 
By the Lebesgue differentiation theorem $\phi$ is differentiable almost everywhere with derivative $\mathbb{I}_A (t)$. 
Differentiable semimartingale characteristics $(B,C,\nu)$ are given by $B_t(s,\omega) = \int_0^t \mathbb{I}_A(s+v) \dd{v}$, $C = 0$, $\nu = 0$. The function $\mathbb{I}_A(t)$ does not have a left-handed limit for all left border points of a removed interval and no right-handed limit for all right border points of a removed interval. To see this consider a left border point (the right border point case follows analogously) $t$ (e.g. $t=\frac{3}{8}$). Then for every $\epsilon > 0$ the sets $A \cap [t-\epsilon,t)$ and $A^C \cap [t-\epsilon,t)$ have positive measure. Hence we can always find sequences of times $s_n \uparrow t$ and $\tilde{s}_n \uparrow t$ with $\mathbb{I}_A (s_n) = 0$ and $\mathbb{I}_A (\tilde{s}_n) =1$. 
\end{example}

Theorem \ref{thm:affinesemiinhom} shows that although the left derivatives of $\Phi$ and $\psi$ might not exist everywhere (i.e. they are not regular), we can think of derivatives of $\Phi$ and $\psi$ existing with respect to $\dd{G}$ and $\dd{G}$-a.e.. In particular if $G$ can be chosen as $G(s) = s$ this then corresponds to actual derivatives, which exist almost everywhere. This motivates the following definition. 
\begin{definition}
A continuously affine transition function $P_{s,t}$ is called absolutely continuously affine if
the functions $s \mapsto \Phi_{s,T}(u)$ and $s \mapsto \psi_{s,T}(u)$ are absolutely continuous on $[v,T]$ for each $o(T,u) < v \leq T$.
\end{definition}
\begin{corollary}
Assume that the assumptions of Theorem \ref{thm:affinesemiinhom} hold. In this case
$((\Theta,X),\mathcal{\tilde F},\{\tilde \PM^{(s,x)} \})$ is an It\^o process (see \citet{CJ80}) if and only if the affine transition function is absolutely continuous. 
\end{corollary}
\begin{proof}
If $(\Theta,X)$ is a It\^o process, we know from the start that we can find semimartingale characteristics of $X$ generated by $(b,c,K,\dd{s})$. Then the derivation goes through as before.
On the other hand if the functions $\Phi$ and $\psi$ are absolutely continuous, it easy to see that we could have used $\dd{s}$ as candidate in the proof of Lemma \ref{lem:deterministicA}. Note that $\Theta_t(s) = \Theta_0(s) + t$ and that $(B^\Theta,0,0)$ with $B_t^\Theta = \int_0^t \dd{s}$ are differentiated characteristics of $\Theta$. Hence $X$ is an It\^o process. 
\end{proof}
$\eta_t$ of \eqref{eq:dettimechange}, i.e.
$$\eta_t(r) =  \inf \{ v \geq 0: G(v+r)-G(r) \geq t \}.$$
\begin{corollary}
Assume that the assumptions of Theorem \ref{thm:affinesemiinhom} hold.
The time-changed Markov process $((\hat \Theta, \hat X), \mathcal{\hat F},\{\tilde \PM^{(s,x)} \})$ with
$$(\hat \Theta,\hat X)_t(r,\omega) = (\Theta, X)_{\eta_t(r)}(r,\omega), \quad \mathcal{\hat F}_t = \mathcal{\tilde F}_{\eta_t}$$ is a Markov semimartingale with differentiable characteristics, i.e. it is an It\^o process and the differentiable characteristics are affine in $\hat X$. 
\end{corollary}
\begin{proof}
By construction $\dd{s} \ll \dd{G(\cdot)}$. Denote the density by $s \mapsto f^\Theta(s)$. Then the characteristics of $\Theta$ can also be written as $B^\Theta_t(r) = \int_0^t f^\Theta(r+s) \dd{G_s(r)}$. By Proposition 7.13 in \citet{CJ80} the time-changed process $((\hat \Theta, \hat X), \mathcal{\hat F},\{\tilde \PM^{(s,x)} \})$ is a Markov semimartingale and a version of the semimartingale characteristics is given by (note that $G_{\eta_s(r)} - G_0(r) = s$)
\begin{equation*}
\begin{aligned}
& \hat B_t = \int_0^t b(\hat \Theta_v)+ B(\hat \Theta_v,\hat X_{v-}) \dd{v}, \\
& \hat C_t = \int_0^t a(\hat \Theta_v) + A(\hat \Theta_v, \hat X_{v-}) \dd{v}, \\
& \hat \nu(\dd{t}, \dd{\xi}) =  \left(m(\dd{\xi}; \hat \Theta_t) + M(\dd{\xi}; \hat \Theta_t,\hat X_{t-})\right) \dd{t}.
\end{aligned}
\end{equation*}
\end{proof}

\begin{corollary}
Assume that the assumptions of Theorem \ref{thm:affinesemiinhom} hold.
Define the time-inhomogeneous affine transition operator $\{\hat P_{s,T} \}$ by $\hat P_{s,T} = P_{G^{-1}(s),G^{-1}(T)}$, i.e.
\begin{align*}
\hat \Phi_{s,T}(u) & = \Phi_{G^{-1}(s),G^{-1}(T)}(u),  \\
\hat \psi_{s,T}(u) & = \psi_{G^{-1}(s),G^{-1}(T)}(u).
\end{align*}
Then $\hat P_{s,T}$ is an absolutely continuously affine transition function. 
\end{corollary}
\begin{proof}
By Theorem \ref{thm:affinesemiinhom} the Riccati equations \eqref{eq:RiccatiinhomA} hold for $\Phi$ and $\psi$. Together with the definition of $\hat \Phi$ and $\hat \psi$ this yields
\begin{equation}
\begin{aligned}
\hat \Phi_{s,T}(u) & = 1 + \int_{G^{-1}(s)}^{G^{-1}(T)} \Phi_{v,G^{-1}(T)}(u) F(v,\psi_{v,G^{-1}(T)}(u)) \dd{G(v)}, \\
& = 1 + \int_{s}^{T} \Phi_{G^{-1}(w),G^{-1}(T)}(u) F(G^{-1}(w),\psi_{G^{-1}(w),G^{-1}(T)}(u)) \dd{w} \\
& = 1 + \int_{s}^{T} \hat \Phi_{w,T}(u) F(w,\hat \psi_{w,T}(u)) \dd{w} \\
\hat{\psi}_{s,T}(u) & = u + \int_s^t R(w,\hat \psi_{w,T}(u)) \dd{w}, \\[-1 em]
\end{aligned}
\end{equation}

\end{proof}
Note that the time-changed process $(\hat X, \mathcal{\hat F},\{\tilde \PM^{(r,x)}\})$ is no longer a time-inhomo\-geneous Markov process in the sense of Definition \ref{def:markov2inhom}. However, under each measure $\tilde \PM^{(r,x)}$ $\hat X$ is a time-inhomo\-geneous Markov process on $(\tilde \Omega, \mathcal{\tilde A}, \mathcal{\hat F},\tilde \PM^{(r,x)})$ in the sense of Definition \ref{def:markov1} with a regular affine transition function given by $\{ \hat P_{r+\eta_s(r),r+\eta_t(r)} \}_{0 \leq s \leq t}$.

\renewcommand{\T}{\cdot}
\chapter{Affine LIBOR models driven by real-valued affine processes} \label{cha:cosh}
\input{AffineCosh.tex}
\chapter{Affine inflation market models} \label{cha:infl}
\input{AffineFwdCPI.tex}

\input{Dissertation.bbl}

\end{document}

%% file: mathcommands.tex
\usepackage{amsfonts,amsmath,amssymb,amsthm}
\usepackage{enumerate}
\usepackage{csquotes} 
\makeatletter
\g@addto@macro{\thm@space@setup}{\thm@headpunct{:}}
\makeatother

\usepackage{array} 
\usepackage{arydshln}	
\newcolumntype{L}[1]{>{\raggedright\let\newline\\\arraybackslash\hspace{0pt}}m{#1}}
\newcolumntype{C}[1]{>{\centering\let\newline\\\arraybackslash\hspace{0pt}}m{#1}}
\newcolumntype{R}[1]{>{\raggedleft\let\newline\\\arraybackslash\hspace{0pt}}m{#1}}
\newcolumntype{H}{>{\setbox0=\hbox\bgroup}c<{\egroup}@{}} 

\newcommand{\e}{\mathrm{e}} 
\newcommand{\dd}[1]{\, \mathrm{d} #1}
\newcommand{\ex}[1]{\, \mathrm{exp} \Big(#1 \Big)}
\newcommand{\1}[1]{\, \mathbb{I} \left\{ #1 \right\}}
\renewcommand{\i}{i}
\newcommand{\EV}[2][\!]{\mathbb{E}^{#1} \left[ #2 \right]} 
\newcommand{\Var}[2][\!]{\; {\mathbb{V}\mathrm{ar}}^{#1} \big[ #2 \big]}
\newcommand{\logn}[1]{\, \ln \left( #1 \right)}
\newcommand{\sign}[1]{\, \mathrm{sgn} \left( #1 \right)}
\newcommand{\mat}[1]{\begin{pmatrix} #1 \end{pmatrix}}
\newcommand{\T}{^\top} 
\newcommand{\Transp}[1]{\left(#1\right)^{\! \top}} 
\newcommand{\scalarprod}[1]{ \langle #1  \rangle}

\newcommand{\ceil}[1]{\left\lceil#1\right\rceil}

\DeclareMathOperator{\E}{\mathbb{E}}
\DeclareMathOperator{\PM}{\mathbb{P}}
\DeclareMathOperator{\QM}{\mathbb{Q}}

\DeclareMathOperator{\R}{\mathbb{R}}
\DeclareMathOperator{\C}{\mathbb{C}}
\DeclareMathOperator{\F}{\mathcal{F}}
\DeclareMathOperator{\I}{\mathcal{I}}
\DeclareMathOperator{\Flt}{{\mathrm{Flt}}}
\DeclareMathOperator{\Cor}{\mathbb{C}\mathrm{or}}

\setlength{\parindent}{1.5ex}

%% file: deckblatt.tex

\thispagestyle{empty}
\begin{center}
\vspace*{\fill}

\includegraphics[width=3.5cm]{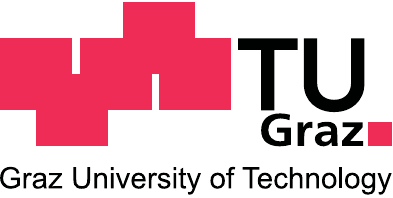}

%

\vspace*{2cm}
\sffamily
\normalsize  Dipl.-Ing. Stefan Waldenberger
\\[1cm]

{\Large\bfseries  Time-inhomogeneous affine processes \\ and affine market models
 \\[2cm]}

{\large\bfseries  DISSERTATION}\\[0.5cm]

\normalsize zur Erlangung des akademischen Grades\\[5mm]
Doktor der technischen Wissenschaften \\[1cm]

eingereicht an der\\[5mm]

{\large\bfseries  Technischen Universität Graz}\\[1cm]

\normalsize  Betreuer/in:\\
Prof. Dr. Wolfgang Müller \\[0.5cm]
Institut für Statistik 
\\[1.5cm]

\small  Graz, September 2015
\end{center}
\vspace*{\fill}

\newpage{}
\thispagestyle{empty}
\vspace*{\fill}
\noindent\textbf{EIDESSTATTLICHE ERKLÄRUNG}\\[2mm]
\emph{\bfseries AFFIDAVIT}\\[5mm]

\noindent Ich erkläre an Eides statt, dass ich die vorliegende Arbeit
selbstständig verfasst, andere als die angegebenen Quellen/Hilfsmittel nicht
benutzt, und die den benutzten Quellen wörtlich und inhaltlich entnommenen
Stellen als solche kenntlich gemacht habe. Das in TUGRAZ-online hochgeladene
Textdokument ist mit der vorliegenden Dissertation identisch. 

\vspace*{2cm}

\noindent\textit{I declare that I have authored this thesis independently, that I have
not used other than the declared sources/resources, and that I have explicitly
indicated all material which has been quoted either literally or by content
from the sources used. The text document uploaded to TUGRAZonline is identical
to the present doctoral dissertation.}

\vspace*{5cm}

\noindent
\begin{minipage}[h]{0.4\linewidth}
  \begin{center}
    \hrulefill\\
    Datum/Date
  \end{center}
\end{minipage}
\hspace*{0.1\linewidth}
\begin{minipage}[h]{0.5\linewidth}
  \begin{center}
    \hrulefill\\
    Unterschrift/Signature
  \end{center}
\end{minipage}
\vspace*{\fill}



%% file: AffineCosh.tex
\newcommand{\p}{\  .}

Market models, the most famous example being the LIBOR market model, are very popular in the area of interest rate modeling. If these models generate nonnegative interest rates they usually do not give semi-analytic formulas for both basic interest rate derivatives, caps and swaptions. One exception is the class of affine LIBOR models proposed by \citet{KPT11}. Using nonnegative affine processes as driving processes affine LIBOR models guarantee nonnegative forward interest rates and lead to semi-analytical formulas for caps and swaptions, so that calibration to interest rate market data is possible.

In this chapter we modify the setup of \citet{KPT11} to allow for not necessarily nonnegative affine processes. This modification still leads to semi-analytical formulas for caps and swaptions and guarantees nonnegative forward interest rates, but allows for a wider class of driving affine processes and hence is more flexible in producing interest rate skews and smiles.
\citet{DGG12} also propose a modification of affine LIBOR models. There driving processes are affine processes with values in the space of positive semidefinite matrices. The approach in this paper has the advantage that a flexible class of implied volatility surfaces can be produced with a much smaller number of parameters.

The structure of this chapter is as follows. In section \ref{sec:affineprocesscosh} affine processes and their properties are reviewed. Section \ref{sec:marketmodels} introduces the necessary notation and market setup and reviews affine LIBOR models. It concludes with some comments on practical implementation. Section \ref{sec:modaffinemodel} is the main section of this chapter. The first part presents the modified affine LIBOR model and semi-analytical pricing formulas for caps and swaptions are derived. The second part 
then gives some examples of usable affine processes with numerical calculations.

\section{Affine processes} \label{sec:affineprocesscosh}
Let  $X = (X_t)_{0 \leq t \leq T}$ be a homogeneous Markov process with values in $D = \R^m_{\geq 0} \times \R^n$ realized on a measurable space $(\Omega,\mathcal{A})$ with filtration  $(\F_t)_{0 \leq t \leq T}$, with regards to which $X$ is adapted. Denote by $\PM^x[\cdot]$ and $\EV[x]{\cdot}$ the corresponding probability and expectation when $X_0 = x$. 
X is said to be an affine process, if its characteristic function has the form
\begin{equation}  \label{eq:affineMomentGencosh}
\EV[x]{ \e^{u\T X_{t}}}  = \ex{\phi_{t}(u) + \psi_{t}(u) \cdot x}, \quad u \in \i \R^d, x \in D,
\end{equation}
where $\phi: [0,T]\times \i \R^d  \rightarrow \C$ and $\psi: [0,T] \times \i \R^d  \rightarrow \C^d$ with $\i \R^d = \{u \in \C^d: \mathrm{Re}(u) = 0 \}$ and $\cdot$ denoting the scalar product in $\R^d$. 
By homogeneity and the Markov property the conditional characteristic function satisfies
$$ \EV[x]{e^{u\T X_{t}} \vert \F_s} = \ex{\phi_{t-s}(u) + \psi_{t-s}(u) \T X_s}. $$
Accordingly affine processes can also be defined for inhomogeneous Markov processes (see \citet{FI05}), in which case the above equality reads
\begin{equation*} \label{eq:affineMomentGeninhomcosh}
\EV[x]{ \e^{u\T X_t} \vert \F_s} = \ex{\phi_{s,t}(u) + \psi_{s,t}(u)\T X_s}, \quad u \in \i \R^d, x \in D, 
\end{equation*}
with $\phi_{s,t}: \i \R^d  \rightarrow \C$ and $\psi_{s,t}: \i \R^d  \rightarrow \C^d$ for $0 \leq s \leq t$.

$X$ is called an analytic affine process (see \citet{KR08}), if $X$ is stochastically continuous and the interior of the set\footnote{$\mathcal{V}$ can be described as the (convex) set, where the extended moment generating function of $X_t$ is defined for all times $t \leq T$ and all starting values $x \in E$. By Lemma 4.2 in \citet{KM11} the set $\mathcal{V}$ is in fact equal to the seemingly smaller set 
$\left \{ u \in \C^d:   \exists x \in \mathrm{int}(D): \EV[x]{ \e^{\mathrm{Re}(u)\T X_T}} < \infty \right \}.$}
\begin{align}
\mathcal{V} & := \left \{ u \in \C^d: \sup_{0 \leq s \leq T}  \EV[x]{ \e^{\mathrm{Re}(u)\T X_s}} < \infty \quad  \forall x \in D \right \}, \label{eq:momsetcosh}
\end{align}
contains 0\footnote{This also implies that $X$ is conservative, i.e. $\PM^x(X_t \in D) = 1 \  \forall x \in D \text{ and } 0 \leq t \leq T$. 
}. In this case the functions $\phi$ and $\psi$ have continuous extensions to $\mathcal{V}$, which are analytic in the interior, such that \eqref{eq:affineMomentGencosh} holds for all $u \in \mathcal{V}$.

The class of affine processes includes Brownian motion and more generally all Lévy processes. Since Lévy processes have stationary independent increments, in this case $\psi_t(u) = u$ and $\phi_t(u) = t \kappa(u)$, where $\kappa$ is the cumulant generating function of the Lévy process. Ornstein-Uhlenbeck processes are further important examples of affine processes. They are discussed in section \ref{sec:OUprocess}. 

The standard reference for affine processes is \citet{DFS03}. There they give a characterization of affine processes, where $\phi$ and $\psi$ are specified as solutions of a system of differential equations\footnote{The fact that this characterization holds for all stochastically continuous affine processes was first shown in \citet{KST11} and later for affine processes with more general state spaces in \citet{KST11b} and \citet{CT13}.}. 
Of all the rich theory of affine processes the methods in this chapter only use the specific form \eqref{eq:affineMomentGencosh} of their moment generating function and the following property.
\begin{lemma} \label{lem:psimontonicity}
Let $X$ be a one-dimensional analytic affine process and $\mathrm{Re}(u) < \mathrm{Re}(w)$, $u,w \in \mathcal{V}$. Then
$\mathrm{Re}(\psi_t(u)) < \psi_t(\mathrm{Re}(w)),$
i.e. $\psi_t \vert_{\mathcal{V} \cap \R}$ is strictly increasing. 
\end{lemma}
\begin{proof}
The case $D = \R_+$ is already contained in \citet{KPT11}. In case $D = \R$ the lemma follows from the fact that by Proposition 3.3 in \citet{KST11} $\psi_t(u) = \e^{\beta t} u$ for some constant $\beta$.
\end{proof}
%
%
\begin{remark}
If $D = \R_+$, it is known that both, $\psi$ and $\phi$, are monotonically increasing (\citet{KPT11}). 
With $D=\R$ this stays true for $\psi$, but not $\phi$, as the deterministic affine process $X_t = x_0 - t$ shows. 
\end{remark}

\section{Interest rate market models} \label{sec:marketmodels}

\subsection*{Classical market models}
Consider a tenor structure $0 < T_1 < \dots < T_N < T_{N+1} =: T$ and a market consisting of zero coupon bonds with maturities $T_1,\dots,T_{N+1}$. Their price processes $(P(t,T_k))_{0 \leq t \leq T_k}$ are assumed to be nonnegative semimartingales on a filtered probability space $(\Omega, \mathcal{A}, \F, \PM)$ (here $\F = (\mathcal{F}_t)_{0 \leq t \leq T}$), which satisfy $P(T_k,T_k) = 1$ almost surely.  If there exists an equivalent probability measure $\QM^{T}$ such that the normalized bond price processes $P(\cdot,T_k) / P(\cdot,T)$ are martingales\footnote{One can extend bond price processes to $[0,T]$ by setting $P(t,T_k) := \frac{P(t,T)}{P(T_k,T)}$ for $t > T_k$, so that $P(\cdot,T_k) / P(\cdot,T)$ is a martingale on $[0,T]$ if and only if it is a martingale on $[0,T_k]$. Economically this can be interpreted as immediately investing the payoff of a zero coupon bond into the longest-running zero coupon bond.
}, the market is arbitrage-free. In this case we can define equivalent martingale measures $\QM^{T_k}$ for the numeraires $P(t,T_k)$ instead of $P(t,T)$ by
\begin{equation} \frac{\dd{\QM^{T_k}}}{\dd{\QM^{T}}} =  \frac{1}{P(T_k,T)} \frac{P(0,T)}{P(0,T_k)}.  \label{eq:measurechangecosh}
\end{equation}
In particular under the measure $\QM^{T_k}$ the forward bond price process $\frac{P(\cdot,T_{k-1})}{P(\cdot,T_{k})}$ and the forward interest rate process $F^k(\cdot)$,
\begin{equation} \label{eq:fwdintrate}
F^k(t) = \frac{1}{\Delta_k} \left(\frac{P(t,T_{k-1})}{P(t,T_{k})} - 1\right), \qquad \Delta_k = T_{k}-T_{k-1},
\end{equation}
are martingales. This is the basic market setup used throughout the rest of the chapter.

In the classical LIBOR market models forward interest rate processes $F^k$ are modeled as continuous exponential martingales under their respective martingale measure $\QM^{T_k}$. Hence forward interest rates are positive. Using driftless geometric Brownian motions as driving processes caplet prices are given by the Black formula (\citet{BF76}) while swaption prices cannot be calculated analytically.
Alternatively one can start with modeling the forward bond price processes ${P(\cdot,T_{k-1})}/{P(\cdot,T_{k})}$ instead of forward interest rate processes. Using again exponential martingales like a driftless Brownian motion it is then 
possible to analytically calculate caplet and swaption prices (see \citet{EO05}). The drawback of this approach is that forward interest rates will be negative with positive probability.

\citet{KPT11} proposed the affine LIBOR models, where forward interest rates are nonnegative while swaption and caplet prices can still be calculated semi-analytically, i.e. up to a numerical integration. The above approaches model the individual forward interest rate processes (resp. forward bond price process) with respect to the individual measure $\QM^{T_k}$ under which they are a martingale. Contrary \citet{KPT11} model the price processes $P(\cdot,T_k)/P(\cdot,T)$, which are all martingales under the same probability measure $\QM^{T}$.
\begin{remark}
Note that all models mentioned in this chapter do not fully specify the whole term structure, but only part of it. In order to price derivatives not contained within the specified tenor structure it is necessary to specify some kind of interpolation scheme. Arbitrary interpolations may lead to arbitrage, however one can always choose an interpolation method, such that the model stays arbitrage-free (\citet{WE10}).
\end{remark}

\subsection*{The affine LIBOR models}

This section presents a summary of the affine LIBOR model introduced in \citet{KPT11}. 
On the filtered probability space $(\Omega, \mathcal{A}, \F, \QM^T)$ consider a nonnegative analytic affine process $X$ with a fixed starting value $x_0 \in \R_{\geq0}^d$. For the tenor structure  $0 < T_1 < \dots < T_N < T_{N+1} =: T$ define for $k=1,\dots,N$ and $0 \leq t \leq T_k$
\begin{equation} \label{eq:affinefwdprices}
 \frac{ P(t,T_k)}{P(t,T)} := \EV[\QM^T]{\e^{u_k\T X_T} \vert \F_t} = \e^{\phi_{T-t}(u_k) + \psi_{T-t}(u_k)\T X_t}, \qquad u_k \geq 0, u_k \in \mathcal{V},
\end{equation}
where $\EV[\QM]{\cdot}$ denotes the expectation with respect to a probability measure\footnote{Since $x_0$ is fixed, contrary to to section \ref{sec:affineprocesscosh} any dependence of probability measures on the starting value of the Markov process $X$ will be suppressed from now on.}  $\QM$. These price processes are martingales and the resulting model is arbitrage-free.

Writing
\begin{equation} \label{eq:fwdnormbonds}
\frac{P(t,T_{k-1})}{P(t,T_k)} = \left. {\frac{P(t,T_{k-1})}{P(t,T)}} \right / { \frac{P(t,T_k)}{P(t,T)}} 
\end{equation}
in \eqref{eq:fwdintrate} shows that forward interest rates being nonnegative is equivalent to normalized bond prices of \eqref{eq:affinefwdprices} satisfying
\begin{equation} \label{eq:bondmonotonicity}
\frac{ P(t,T_1)}{ P(t,T)} \geq .. \geq \frac{ P(t,T_N) }{P(t,T)} \geq 1.
\end{equation}
Since for $x\geq 0$, $\e^{u^T x}$ is monotonically increasing in every component of $u$ the monotonicity for normalized bond prices in \eqref{eq:bondmonotonicity} is satisfied as long as $u_1 \geq \dots u_{N} \geq 0$. 

The parameters $u_k$ in \eqref{eq:affinefwdprices} should be determined, so that the starting values of normalized bond prices ${ P(0,T_k)}/{P(0,T)} =
 \ex{\phi_T(u_k) + \psi_T(u_k)\T x_0}$
fit the initial term structure inferred from actual market data. 
For most affine processes every term structure can be fitted and for currently nonnegative forward interest rates this can be done using an decreasing sequence $u_1 \geq \dots \geq u_N \geq 0$ (see \citet{KPT11}).

\begin{remark}
Since $X$ is nonnegative, the k-th normalized bond price is not only greater equal to one, but is bounded from below by the time-dependent constant given by
$ \ex{\phi_{T-t}(u_{k}) },$
which is strictly greater than one. Accordingly in the affine LIBOR models forward interest rates are bounded from below by a strictly positive time-dependent constant.
\end{remark}

Affine LIBOR models lead 
to nonnegative forward interest rates. Additionally this specification is appealing because the density processes for changes of measures are again exponentially affine in $X_t$, i.e. inserting \eqref{eq:affinefwdprices} into \eqref{eq:measurechangecosh} gives
\begin{align*}
\frac{\dd{\QM^{T_k}}}{\dd{\QM^{T}}} & = \frac{P(0,T)}{P(0,T_k)}  \e^{\phi_{T-T_k}(u_k) + \psi_{T-T_k}(u_k)\T X_{T_k}}.
\end{align*}
Moreover normalized bond prices and because of \eqref{eq:fwdnormbonds} also forward bond prices are of exponential affine form. It follows that the moment generating function of the logartihm of normalized bond prices under $\QM^{T_k}$ is also of exponential affine form and that calculation of caplet prices is possible via a one-dimensional Fourier inversion. 
If the dimension of the driving process is one, swaption prices can also be calculated via one-dimensional Fourier inversion (see \citet{KPT11}). Hence this approach satisfies both, nonnegative interest rates and analytical tractability of standard interest rate market instruments. If the dimension is larger than one, the exact price of swaptions can only be calculated via higher-dimensional integration, the dimension of which is the length of the underlying swap. Alternatively \citet{GPSS14} provide approximate formulas for swaptions. 

\subsubsection*{Practical application of the affine LIBOR model} \label{sec:practicalapplication}
Although this framework is elegant from a theoretical point of view, a practical implementation faces several difficulties which shall be discussed here. 

First, calibration of interest rates and implied volatilities cannot be separated. The initial term structure can be fitted using the $u_k$, but the parameters $u_k$ also have a strong impact on implied volatilities. This can be seen by looking at the forward bond price
\begin{equation} \frac{1}{P(T_{k-1},T_k)} =  \ex{\phi_{T-t}(u_{k-1})-\phi_{T-t}(u_{k})+(\psi_{T-t}(u_{k-1})-\psi_{T-t}(u_k))\T X_{T_{k-1}}}, \label{eq:KRTfltrv} \end{equation}
which is the random variable responsible for the payoff of a caplet.
The driving process $X$ influences the distribution of this random variable through two different channels. First via the parameters of the driving process itself and second via the parameters $u_k$ (depending on $X$ and the initial interest rate term structure). Hence for changes in the yield curve different parameters are required to reproduce the same implied volatility surface. If $X$ is a Lévy process, then as mentioned in section \ref{sec:affineprocesscosh} $\psi_t(u) = u$ and it follows that the distribution of \eqref{eq:KRTfltrv} depends on the difference $u_{k+1} - u_k$, which in turn is related to the steepness of the initial yield curve\footnote{This is similar for most affine processes, but is best visible for Lévy processes.}. Hence caplet implied volatilities are especially sensitive with regards the steepness of the initial yield curve.

Second, interest rates and volatilities of this model depend on the final horizon $T$. Changing the horizon $T$ while using the same affine process $X$ will lead to different results and
there is no general way of rescaling the parameters of $X$ to negate such an effect. This is rather counterintuitive, since extending the horizon of a model should not change the results for quantities already included with the shorter horizon. 
 
Third, the types of possible volatility surfaces is rather constrained in the fully analytically tractable one-dimensional case. For example, we were only able to generate volatility skews\footnote{The smile example of \citet{KPT11}, figure 9.2, using an Ornstein-Uhlenbeck process seems to be numerically incorrect for strikes smaller than 0.4. With the mentioned initial yield curve the underlying interest rate is always larger than the strike, which corresponds to a zero implied volatility, destroying the displayed smile.}.
This might be resolved by using higher-dimensional nonnegative processes. However, in multidimensional affine LIBOR models swaptions can no longer be calculated efficiently by Fourier methods. On the other hand allowing arbitrary affine processes destroys the nonnegativity of forward interest rates, a central property of affine LIBOR models. We propose a modification, that preserves the nonnegativity of forward interest rates without the restriction to nonnegative affine processes.

\section{The modified affine LIBOR model} \label{sec:modaffinemodel} 

On the filtered probability space $(\Omega, \mathcal{A}, \F, \QM^T)$ consider an analytic one-dimen\-sional affine process $X$ with a fixed starting value $x_0$, i.e. the set $\mathcal{V}$ defined in \eqref{eq:momset} contains $0$ in the interior. For $u \in \mathcal{V}$ with $-u \in \mathcal{V}$ consider the martingales $M^u$,
\begin{equation} \label{eq:coshmartingales}
M_t^u :=  \EV[\QM^T]{\cosh(u X_T) \vert \F_t}  = \frac{1}{2} \left(\e^{\phi_{T-t}(u) + \psi_{T-t}(u) X_t} + \e^{\phi_{T-t}(-u) + \psi_{T-t}(-u) X_t} \right).
\end{equation}
By the symmetry of the cosinus hyperbolicus $M^u = M^{-u}$, hence one may restrict $u$ to be nonnegative. 
For the given tenor structure $0 < T_1  < \cdots < T_N \leq T_{N+1} = T$ and the market setup of section \ref{sec:marketmodels} define the normalized bond prices for $k=1, \dots, N$ and $t\leq T_k$ as
\begin{equation*}
\begin{aligned}
\frac{ P(t,T_k)}{P(t,T)} & :=  M_{t}^{u_k}, \qquad u_k \in \{v \in \mathcal{V}: v \geq 0, -v \in \mathcal{V} \}.
\end{aligned}
\end{equation*} 
With $M_t^{u_k}$ being a $\QM^T$-martingale the model is arbitrage-free. For every $x \in \R$ the function $u \mapsto \cosh(u x)$ is increasing in $u \in \R_{\geq 0}$ and satisfies $\cosh(ux) \geq 1$ so that if
$$u_1 \geq u_2 \geq \dots \geq u_N \geq 0,$$
equation \eqref{eq:bondmonotonicity} holds and forward interest rates
\begin{equation*} F^k(t) = \frac{1}{\Delta_k} \left( \frac{M_t^{u_{k-1}}}{M_t^{u_k}} - 1 \right), \quad 0 \leq t \leq T_{k-1}. \end{equation*}
are nonnegative for all $t$. 
To fit initial market data one has to choose the sequence $(u_k)$ so that $M_0^{u_k} = {P(0,T_k)}/{P(0,T)}. $
The following lemma gives the condition for the affine process $X$ under which a given initial term structure can be reproduced and shows that the $u_k$ are uniquely determined.

\begin{lemma} \label{thm:interestfit}
If $$P(0,T_1) / P(0,T) < \sup_{u \in \mathcal{V}: -u \in \mathcal{V}} \EV[\QM^T]{\cosh(u X_T) \vert \F_0},$$ then the model can fit any term structure of nonnegative forward interest rates. Additionally there exists a unique decreasing sequence $u_1 \geq \cdots \geq u_N$, such that  $${ P(0,T_k)}/{P(0,T)} = \EV[\QM^T]{\cosh(u_k X_T) \vert \F_0} = M_0^{u_k}.$$
If forward interest rates are strictly positive, the sequence is strictly decreasing.
\end{lemma}
\begin{proof}
$m(u) = \EV[\QM^T]{\cosh( u X_T) \vert \F_0}$ is a continuous function which is increasing in case $u \geq 0$. By the assumption of the theorem there exists $\overline{u} >0$ with $m(\overline{u})
 > P(0,T_1)/P(0,T)$. Furthermore $
m(0) = 1$, which proves the lemma.
\end{proof}
 \begin{remark}
Generalizing this approach to a $d$-dimensional driving process is possible by setting $$M_t^u = \EV{\prod_{l=1}^d \cosh \left(u^{(l)} X_T^{(l)} \right) \Big \vert \F_t}, \quad u =(u^{(1)},\dots,u^{(d)}) \geq 0.$$
In this case it is guaranteed that $M_t^{u} \geq M_t^w$ for $u \geq w$, which guarantees the nonnegativity of forward interest rates.  However, the option pricing formulas in the following sections do not generalize.
\end{remark}
As in the affine LIBOR model for a monotonically decreasing sequence $(u_k)$ forward interest rates are not only nonnegative, but bounded below by strictly positive time-dependent constants (the bounds can be calculated numerically). This is not a big issue if these bounds are close to zero, but has to be checked during the calibration process. 

In the modified affine LIBOR model the change of measure to the $T_k$-forward measure $\QM^{T_k}$ is given by
\begin{equation} \frac{\dd{\QM^{T_k}}}{\dd{\QM^{T}}}  = \frac{P(0,T)}{P(0,T_k)}  M_{T_k}^{u_k} = \frac{M_{T_k}^{u_k}}{M_0^{u_k}}. \label{eq:coshmeasurechange} \end{equation}
Here $M_t^{u_k}$ is a sum of exponentials of $X_t$, while in the affine LIBOR model the corresponding term is a single exponential. This means that contrary to the affine LIBOR model the process $X$ is not an inhomogeneous affine process under $\QM^{T_k}$ and it is not possible to calculate the moment generating function of the logarithm of foward bond prices under $\QM^{T_k}$. Nevertheless it is possible to get analytical formulas for the prices of caplets and swaptions. 

\subsection{Option pricing} \label{sec:option pricing}
The derivation of the pricing formulas for caplets and swaptions is based on a method first applied in \citet{JA89}. First caplets are dealt with, swaptions follow afterwards\footnote{Actually caplet prices coincide with prices of swaptions with only one underlying period. The difference between those two derivatives is the payoff time.}.
Note that if $u_k = u_{k-1}$ the corresponding forward interest rate $F^{k}$ always stays zero. To exclude such pathological examples assume that the sequence $(u_k)$ is strictly decreasing.
In this section random variables are often viewed as functions of the value of the driving process $X$. Specifically consider the functions $M_t^u: \R \rightarrow \R$, 
\begin{equation} \label{eq:coshmartingalefunction}
x \mapsto M_t^u(x) := \frac{1}{2} \left(\e^{\phi_{T-t}(u) + \psi_{T-t}(u) x} + \e^{\phi_{T-t}(-u) + \psi_{T-t}(-u) x} \right).\end{equation}
The time $t$ value of martingale $M^u$ in \eqref{eq:coshmartingales} is then $M_t^u = M_t^u(X_t)$. In the rest of the chapter $M_t^u$ will denote both, the function and the value of the stochastic processes, where the correct interpretation should be clear from context. 

The payoff of a caplet for the $(k+1)^\text{th}$ forward rate $F^{k+1}(T_{k})$ with strike $K$ is
$$ \Delta_{k+1} \left(F^{k+1}(T_k)- K \right)_+ = \left(\frac{1}{P(T_k,T_{k+1})} - \tilde{K} \right)_+ = \left( \frac{M_{T_k}^{u_{k}}}{M_{T_k}^{u_{k+1}}} - \tilde{K} \right)_+, $$
where $\tilde{K} = 1 + \Delta_{k+1} K$. Since this payoff has to be paid at time $T_{k+1}$ 
the price of the caplet and the corresponding floorlet is
\begin{align*}
\mathrm{Cpl}(t,T_k,T_{k+1},K) &=  P(t,T_{k+1}) \EV[\QM^{T_{k+1}}]{\left( \frac{M_{T_k}^{u_{k}}}{M_{T_k}^{u_{k+1}}} - \tilde{K} \right)_+ \Big \vert \F_t}, \notag \\
\mathrm{Flt}(t,T_k,T_{k+1},K) & =  P(t,T_{k+1}) \EV[\QM^{T_{k+1}}]{\left(\tilde{K} -  \frac{M_{T_k}^{u_{k}}}{M_{T_k}^{u_{k+1}}}  \right)_+ \Big \vert \F_t}. 
\end{align*}
Since price processes are martingales, the put/call parity holds and prices of caplets follow from floorlets and vice versa. Because Fourier analysis is easier for floorlets, where the payoff is bounded, formulas are derived for floorlets.

Since the moment generating function of $\ln({{M_{T_k}^{u_{k}}}/{M_{T_k}^{u_{k+1}}}})$ is unknown, Fourier methods are not directly applicable. However, 
the function $x \mapsto {M_{T_k}^{u_{k}}(x)}/{M_{T_k}^{u_{k+1}}(x)}$ has a unique minimum and is monotonically increasing moving away from this minimum. Using this one can get rid of the positive part and use Fourier inversion to calculate the above expectations. 
The above mentioned monotonicity is very fortunate and follows from a close interplay between the monotonicity of the sequence $(u_k)$ and the function $\psi$ with properties of the cosinus hyperbolicus. Details are laid out in the proof of the following lemma, which can be found in section \ref{sec:proofsCosh}. 
\begin{lemma} \label{lem:monoton} For $i=1, \dots, n$ let $u_0 \geq u_i \geq 0$, where for at least one $i$ $u_0 > u_i$. Let $c_i > 0$ be positive constants. Define a function $g: \R \rightarrow \R$ by 
\begin{equation}
g(x) := \sum_{i=1}^n c_i \frac{M_t^{u_i}(x)}{M_t^{u_0}(x)} \p \label{eq:monoton}
\end{equation}
Then $g$ has a unique maximum at some point $\xi \in \R$ and and is strictly monotonically decreasing to $0$ on the left and right side of $\xi$.
\end{lemma}

For floorlet valuation this lemma is not directly applicable as $u_k > u_{k+1}$, which is the wrong inequality. However, there is only one summand and the lemma can be applied to the inverse ${M_{T_k}^{u_{k+1}}(x)}/{M_{T_k}^{u_{k}}(x)}$. 
It follows that ${M_{T_k}^{u_{k}}(x)}/{M_{T_k}^{u_{k+1}}(x)}$ has a unique minimum at some point $\xi$ and is increasing to infinity to the left and right. 
Hence it is possible to write
\begin{equation} \label{eq:fltpayoff}
\left(\tilde{K} -  \frac{M_{T_k}^{u_{k}}(x)}{M_{T_k}^{u_{k+1}}(x)}  \right)_+ = \left(\tilde{K} -  \frac{M_{T_k}^{u_{k}}(x)}{M_{T_k}^{u_{k+1}}(x)}  \right) \1{\kappa_1 < x <  \kappa_2},
\end{equation}
where $\kappa_1$ and $\kappa_2$ are two uniquely determined constants satisfying $\kappa_1 \leq \xi \leq \kappa_2$. 
If $\kappa_1 = \xi = \kappa_2$ the payoff is 
zero, which corresponds to ${M_{T_k}^{u_{k}}}/{M_{T_k}^{u_{k+1}}}  > \tilde{K}$. This happens if the forward interest rate is bounded from below by K, 
which only happens for very low strikes $K$. Inserting \eqref{eq:fltpayoff} into the price of a floorlet it follows by a change of measure that
\begin{align}
\mathrm{Flt}(t,T_k,T_{k+1},K) 
& =  P(t,T_{k+1}) \EV[\QM^{T_{k+1}}]{\left(\tilde{K} -  \frac{M_{T_k}^{u_{k}}}{M_{T_k}^{u_{k+1}}}  \right)  \1{\kappa_1 < X_{T_k} <  \kappa_2} \Big \vert \F_t} \notag \\
&=  P(t,T) \EV[\QM^{T}]{\left(\tilde{K} M_{T_k}^{u_{k+1}} -  M_{T_k}^{u_{k}}   \right)  \1{\kappa_1 < X_{T_k} <  \kappa_2} \Big \vert \F_t}. \label{eq:floorlet}
\end{align}
$\tilde{K} M_{T_k}^{u_{k+1}} -  M_{T_k}^{u_{k}}$ is the sum of exponentials of the random variable $X_{T_k}$.
The expectation in \eqref{eq:floorlet} is calculated under the measure $\QM^T$, where the conditional moment generating function
$$\mathcal{M}_{X_t \vert X_s}(z) := \EV[\QM^T]{\e^{z X_t}\vert \F_s} =  \EV[\QM^T]{\e^{z X_t}\vert X_s} = \ex{\phi_{t-s}(z) + \psi_{t-s}(z) X_s}$$
is known for $z \in \mathcal{V}$. Hence the expectation in \eqref{eq:floorlet} can be calculated via Fourier inversion. The Fourier inversion formula for terms of the above form is stated in Lemma \ref{lem:fouriertransform}, the proof of which is given in section \ref{sec:proofsCosh}.

\begin{lemma}
\label{lem:fouriertransform}
Assume that the function $f: \R \rightarrow \R$ has the representation
$$f(x) = \sum_k C_k \e^{v_k x}\1{\kappa_1 < x < \kappa_2}, \qquad \lim_{x \downarrow \kappa_1} f(x) = \lim_{x \uparrow \kappa_2} f(x) = 0,$$
where the summation is over a finite index set and the $C_k$ and $v_k$ are real constants.
Then for $R \in \mathcal{V} \cap \R$ 
the Fourier inversion formula  
\begin{equation*}
\EV{f(X_t)\vert \F_s} = \frac{1}{\pi} \int_{0}^{\infty} \mathrm{Re} \left(  { \mathcal{M}_{X_t \vert X_s}(\i u + R)} \hat{f}(u- \i R) \right)  \dd{u}
\end{equation*}
holds, where $\hat{f}$ is the analytic Fouier transform given by
\begin{equation}
\hat{f}(z)  = \frac{1}{\i z} \sum_k  {\frac{C_k v_k }{v_k - \i z}  { \left(\e^{(v_k - \i z) \kappa_2} - \e^{(v_k - \i z) \kappa_1 } \right) } }, \qquad z \neq 0, z \neq -\i v_k.
\end{equation}
\end{lemma}

To calculate the price of a floorlet in \eqref{eq:floorlet} apply Lemma \ref{lem:fouriertransform} to
$f_{k+1}^{K}(X_{T_{k}})$ with
\begin{equation}
f_{k+1}^{K}(x) := \left(\tilde{K} M_{T_k}^{u_{k+1}}(x) -  M_{T_k}^{u_{k}}(x)   \right)  \1{\kappa_1 < x <  \kappa_2}.  \label{eq:ffloorlet}
\end{equation}  
Its Fourier transform is
\begin{equation} \label{eq:fhatfloorlet}
\hat{f}^{K}_{k+1}(z) = \frac{1}{\i z} \left( (1 + \Delta_{k+1} K)  h_{\kappa_1,\kappa_2}^{T_{k}} (-\i z,u_{k+1}) -  h_{\kappa_1,\kappa_2}^{T_{k}}(-\i z,u_{k}) \right)
\end{equation}
with
\begin{equation}  \label{eq:hfunction}
\begin{split}
 h_{\kappa_1,\kappa_2}^t(z,u) & := \e^{\phi_{T-t}(u)} \frac{\psi_{T-t}(u)}{2 ( z + \psi_{T-t}(u))} \left(\e^{(z + \psi_{T-t}(u)) \kappa_2} - \e^{( z + \psi_{T-t}(u)) \kappa_1 } \right) \\
& + \e^{\phi_{T-t}(-u)} \frac{\psi_{T-t}(-u)}{2 (z + \psi_{T-t}(-u))} \left(\e^{( z + \psi_{T-t}(-u)) \kappa_2} - \e^{( z + \psi_{T-t}(-u)) \kappa_1 } \right).
\end{split}
\end{equation}

The case of swaptions is similar. Consider a swap which is part of the tenor structure. That is, consider $1 \leq \alpha < \beta \leq N$ and the according interest rate swap with forward swap rate
$$  S_{\alpha,\beta}(t)  = \frac{P(t,T_\alpha) - P(t,T_\beta)}{\sum_{k=\alpha+1}^\beta \Delta_k P(t,T_k)}, \qquad  \Delta_k = T_{k} - T_{k-1}. $$
The payoff of a put swaption on the above swap with strike $K$ is then 
\begin{align*}
\sum_{k=\alpha+1}^\beta P(T_\alpha,T_k) \Delta_k \left(K - S_{\alpha,\beta}(T_\alpha) \right)_+ & = \left(P(T_\alpha,T_\beta) + K \sum_{k=\alpha+1}^\beta \Delta_k P(T_\alpha,T_k) - 1\right)_+ \\
& = \left(\frac{M_{T_\alpha}^{u_\beta}}{M_{T_\alpha}^{u_\alpha}} +  \sum_{k=\alpha+1}^\beta K \Delta_k \frac{M_{T_\alpha}^{u_k}}{M_{T_\alpha}^{u_\alpha}} - 1 \right)_+ .
\end{align*}
Since the function ${M_{T_\alpha}^{u_\beta}(x)}/{M_{T_\alpha}^{u_\alpha}(x)} +  \sum_{k=\alpha+1}^\beta K \Delta_k {M_{T_\alpha}^{u_k}(x)}/{M_{T_\alpha}^{u_\alpha}(x)} $ is of the form of Lemma \ref{lem:monoton}, it has a unique maximum $\xi$ and one can find constants\footnotemark  $\ \kappa_1 \leq \xi \leq \kappa_2$ such that 
after a change of measure the value of a put swaption is 

\footnotetext{
As in the floorlet case if $\kappa_1 = \kappa_2$ then the forward swap rate is always larger than the strike. Note that $S_{\alpha,\beta}(t)$ can also be written as 
$S_{\alpha,\beta}(t) = \sum_{k=\alpha+1}^\beta w_k(t) F^k(t)$
with $w_k > 0$ (see e.g. \citet{BM06}). It follows that if forward interest rates are bounded below by positive constants the same will be true for forward swap rates. This bound is then at most an average of the corresponding forward interest rates bounds and is therefore of the same order of magnitude, which for a meaningful model will be small enough.
}

\begin{equation*}
\mathrm{Put Swaption}(t,T_\alpha,T_\beta,K) = P(t,T) \EV[\QM^T]{f^K_{\alpha,\beta}(X_{T_\alpha}) \Big \vert \F_t},
\end{equation*}
where 
\begin{equation} f^K_{\alpha,\beta}(x) = \left( {M_{T_\alpha}^{u_\beta}(x)}  -  M_{T_\alpha}^{u_\alpha}(x) +  \sum_{k=\alpha+1}^\beta K \Delta_k {M_{T_\alpha}^{u_k}(x)} \right) \1{\kappa_1 < x < \kappa_2}.
\label{eq:fputswaption}
\end{equation}
Again this of the form in Lemma \ref{lem:fouriertransform} and in this case
\begin{equation} \label{eq:fhatputswaption}
\begin{aligned}
\hat{f}^K_{\alpha,\beta}(z) = 
\frac{1}{\i z} \Big( &   h_{\kappa_1,\kappa_2}^{T_\alpha}(-\i z,u_\beta) -   h_{\kappa_1,\kappa_2}^{T_\alpha}(- \i z ,u_\alpha) + {K} \sum_{k=\alpha+1}^\beta \Delta_k h_{\kappa_1,\kappa_2}^{T_\alpha}(-\i z,u_k) \Big),
\end{aligned}
\end{equation}
where $h^t_{\kappa_1,\kappa_2}(z,u)$ is defined in \eqref{eq:hfunction}.
The pricing formulas are summarized in the following theorem.
\begin{theorem}
Let $R \in \mathcal{V} \cap \R$. In the modified affine LIBOR model prices of a forward interest rate put and a put swaption are
\begin{equation}
\Flt(t,T_k,T_{k+1},K) = \frac{P(t,T)}{\pi} \int_{0}^{\infty} \mathrm{Re} \left(  { \mathcal{M}_{X_{T_k}\vert X_t}(R + \i u)}
\hat{f}^{K}_{k+1}(u - \i R) \right) \dd{u}, 
\end{equation}
\begin{equation}
\mathrm{Put Swaption}(t,T_\alpha,T_\beta,K) = \frac{P(t,T)}{\pi} \int_{0}^{\infty} \mathrm{Re} \left(  { \mathcal{M}_{X_{T_\alpha}\vert X_t}(R + \i u)}
\hat{f}^K_{\alpha,\beta}(u - \i  R) \right) \dd{u},
\end{equation}

The Fourier transforms
$\hat{f}^{K}_{k+1}$ repectively $\hat{f}^K_{\alpha,\beta}$ are given in \eqref{eq:fhatfloorlet} respectively \eqref{eq:fhatputswaption} for $R \notin \{0,u_k, u_{k+1}\}$ respectively $R \notin \{0,u_\alpha, \dots, u_\beta \}$.
\end{theorem}

In order to calculate $\hat{f}^{K}_{i}$ respectively $\hat{f}^K_{\alpha,\beta}$ one has to find the roots $\kappa_1, \kappa_2$ of the functions
\begin{align} 
g^K_{k}(x) & := \tilde{K} -  \frac{M_{T_k}^{u_{k}}(x)}{M_{T_k}^{u_{k+1}}(x)},  \label{eq:fltfunc} \\
{g}^K_{\alpha,\beta}(x) & := \frac{M_{T_\alpha}^{u_\beta}(x)}{M_{T_\alpha}^{u_\alpha}(x)} +  \sum_{k=\alpha+1}^\beta K \Delta_k \frac{M_{T_\alpha}^{u_k}(x)}{M_{T_\alpha}^{u_\alpha}(x)} - 1 .\label{eq:swaptionfunc}
\end{align}
By Lemma \ref{lem:monoton} this amounts to finding the roots of a function which has a single optimum and is monotonic when moving away from this optimum.
Numerical determination of the roots of such well-behaved one-dimensional functions poses no problem. 
Having determined those bounds valuation reduces to a one-dimensional integration of a function that is falling at least like $1/x^2$ (depending on the moment generating function of the affine process), so also numerical integration is feasible.
Note that besides caps, floors and swaptions, options like digital options or Asset-or-Nothing options can be calculated in a similar manner.

\subsection{Examples} \label{sec:examples}
The first part of this section looks at the benchmark case of a Brownian motion, where everything can also be calculated in closed form. Afterwards Ornstein-Uhlenbeck processes are discussed. The section concludes with examples of possible volatility surfaces.
\subsubsection*{Brownian motion}
Choose $X_t = B_t$, a standard Brownian motion starting in 0. 
The conditional moment generating function is
$$\mathcal{M}_{B_T \vert B_t}(u) = \EV{\e^{u B_T} \vert \F_t} = \ex{u B_t + \frac{u^2}{2} (T-t)}.$$
Hence this is an affine process with $\phi_t(u) = \frac{u^2}{2} t$ and $\psi_t(u) = u$. 
Consider the time $0$ price of a floorlet as given in \eqref{eq:floorlet} with $t=0$.
Since $M_t^u(-x) = M_t^u(x)$ one finds that in this case  
$\kappa_2 = \kappa$ and $\kappa_1 = - \kappa$, where $\kappa$ is the unique positive root of \eqref{eq:fltfunc} if $g^{K}_{k+1}(0) < 0$ and $\kappa=0$ otherwise. 
By \eqref{eq:floorlet} the floorlet price $\Flt(0,T_k,T_{k+1},K)$ is
\begin{align*}P(0,T) \EV[\QM^{T}]{ \left( {\tilde{K}} \e^{\frac{u_{k+1}^2}{2} (T-T_k)}\cosh(u_{k+1} B_{T_k}) -  \e^{\frac{u_{k}^2}{2} (T-T_k)} \cosh(u_{k} B_{T_k})\right)  \1{\vert B_{T_k} \vert \leq \kappa} }. 
\end{align*}
By the symmetry of a Brownian motion starting in $0$
\begin{align*}  \EV{\cosh(z B_t) \1{\vert B_t \vert \leq \kappa}} & = \EV{\e^{z B_t} \1{\vert B_t \vert \leq \kappa}} = \EV{\e^{-z B_t} \1{\vert B_t \vert \leq \kappa}} \\
&= \e^{\frac{1}{2} t z^2} \left( \Phi \Big(\frac{\kappa}{\sqrt{t}}- z \sqrt{t} \Big) - \Phi \Big (-\frac{\kappa}{\sqrt{t}} - z \sqrt{t} \Big) \right) ,
\end{align*}
where $\Phi$ denotes the cumulative distribution function of a standard normal distributed random variable. Hence the floorlet price $\mathrm{Flt}(0,T_k, T_{k+1},K)$ is 
\begin{align*}
 {\tilde{K}} & P(0,T) \e^{u_{k+1}^2  \frac{T}{2} } \left( \Phi \Big(\frac{\kappa}{\sqrt{T_k}}- u_{k+1} \sqrt{T_k} \Big) - \Phi \Big(- \frac{\kappa}{\sqrt{T_k}}-u_{k+1} \sqrt{T_k}\Big)  \right) \\
-  &  P(0,T) \e^{u_k^2  \frac{T}{2} } \left( \Phi \Big(\frac{\kappa}{\sqrt{T_k}}-u_k \sqrt{T_k}\Big) - \Phi \Big(- \frac{\kappa}{\sqrt{T_k}}-u_k \sqrt{T_k}\Big)  \right).
\end{align*}
Slightly more complicated formulas exist when $B$ is replaced with a Brownian motion with constant drift and volatility and a starting value different from $0$. 

Swaptions can be treated the same way. Let $\kappa$ be the unique positive root of \eqref{eq:swaptionfunc} if $g^K_{\alpha,\beta}(0) > 0$ and $\kappa=0$ otherwise. Then
\begin{align*}
\mathrm{Put}&\mathrm{Swaption}(0,T_\alpha,T_\beta,K) \\
& =  P(0,T)   \e^{u_\beta^2  \frac{T}{2} } \left( \Phi \Big( \frac{\kappa}{\sqrt{T_\alpha}}- u_\beta \sqrt{T_\alpha}\Big) - \Phi \Big(- \frac{\kappa}{\sqrt{T_\alpha}}-u_\beta \sqrt{T_\alpha} \Big)  \right) \\
& - P(0,T) \e^{u_\alpha^2  \frac{T}{2} } \left( \Phi \Big(\frac{\kappa}{\sqrt{T_\alpha}}-u_\alpha \sqrt{T_\alpha} \Big) - \Phi \Big(-\frac{\kappa}{\sqrt{T_\alpha}} -u_\alpha \sqrt{T_\alpha} \Big)   \right) \\
& + P(0,T)  \sum_{k=\alpha+1}^\beta K \Delta_k \e^{u_k^2  \frac{T}{2} } \left( \Phi \Big( \frac{\kappa}{\sqrt{T_\alpha}}- u_k \sqrt{T_\alpha}\Big) - \Phi \Big(- \frac{\kappa}{\sqrt{T_\alpha}}- u_k \sqrt{T_\alpha} \Big)  \right).
\end{align*}

\subsubsection*{Ornstein-Uhlenbeck (OU) processes} \label{sec:OUprocess}
The OU process $X$ generated by a Lévy process $L$ is defined as the unique strong solution of (see \citet{SA99}, section 17)
\begin{equation} \dd{X_t} = - \lambda X_t \dd{t} + \dd{L_t}, \quad X_0 = x_0.  \label{eq:OUprocess} \end{equation}
Then $ Y_t := \e^{\lambda t} X_t 
= x_0 + \int_0^t \e^{\lambda s} \dd{L_s}. $
Using the key formula of Lemma 3.1 in \citet{ER99} (here we assume the according integrability condition of the jumps of the Levy process) it follows that
$$ \EV{\e^{u X_t}} = \EV{\ex{\e^{-\lambda t} u Y_t}} = \ex{\e^{-\lambda t} u x_0 + \int_0^t \kappa(\e^{-\lambda s} u) \dd{s}},$$
where $\kappa(u) = \logn{\EV{\e^{u L_1}}}$ is the cumulant generating function of the Lévy process $L$.
Hence this process is affine with 
\begin{equation} \psi_t(u) = \e^{-\lambda t} u \quad \text{ and } \quad \phi_t(u) =  \int_0^t \kappa(\e^{-\lambda s} u) \dd{s} = \frac{1}{\lambda} \int_{\e^{-\lambda t}}^1 \frac{\kappa(v u)}{v} \dd{v}. \label{eq:phiOU} \end{equation}
By Corollary 2.10 in \citet{DFS03} every affine process with state space $\R$ is in fact an OU process. Hence in the context of affine processes defined on the real line OU processes are the right class to consider.  
For application it should be possible to calculate the integral in \eqref{eq:phiOU} analytically. Two examples where this is possible are presented below.

\begin{remark}
If $L$ is a martingale, the process in \eqref{eq:OUprocess} is mean-reverting to zero, however shifting the mean to $\theta$ is easily done by using $Z_t = \theta+X_t.$ Then $\dd{Z_t} = \lambda (\theta-Z_t) \dd{t} + \dd{L_t}$ and
$$\EV{\e^{u Z_t} } = \ex{(\phi_t(u) + \theta u  (1-\e^{-\lambda t})) + \psi_t(u) Z_0}.$$ 
Hence $Z$ is again affine with $\psi_t^\theta(u) = \psi_t(u)$ and $\phi_t^\theta(u) = \phi_t(u) + \theta u  (1-\e^{-\lambda t})$.
Note that this process is then generated by the Lévy process $\tilde{L}_t = L_t + \theta \lambda t$, i.e. the original Lévy process plus an additional drift of $\theta \lambda$.
\end{remark}

The first example is the classical OU process generated by a Brownian motion $ \sigma B$, where $\kappa(u) = \frac{1}{2} \sigma^2 u^2$. This process is described by
$$\dd{X_t} =  - \lambda  X_t \dd{t} + \sigma \dd{B_t},  \quad X_0 = x. $$ The integral in \eqref{eq:phiOU} is 
\begin{align} \label{eq:phicont}
\phi_t(u) = \frac{1}{\lambda} \int_{\e^{-\lambda t}}^1 \frac{ \kappa(v u)}{v} \dd{v} 
 =  \frac{\sigma^2 u^2}{4 \lambda} (1-\e^{-2 \lambda t}).
\end{align}

%
%
%
%
%
%

With Brownian motion describing the continuous part of Lévy processes., for the second example we consider a pure jump process, namely a Double $\Gamma$-OU process. $\Gamma$-OU processes are generated by a compound Poisson process with jump intensity $\lambda \beta$ ($\lambda$ being the same as in \eqref{eq:OUprocess}) and exponentially distributed jumps with expectation value $\alpha$. The limit distribution of this process is a $\Gamma$-distribution, which gives the process its name. 
As the generating compound Poisson process is strictly increasing, the generated $\Gamma$-OU process is a subordinator and stays above 0. In order to find an OU process with values in $\R$ consider the difference of two independent compound $\Gamma$-OU processes $L^+, L^-$ with parameters $\alpha^+, \beta^+, \alpha^-, \beta^-$ and set $\lambda^+ = \lambda \beta^+, \lambda^- =  \lambda \beta^-$. Then $L = L^+ - L^-$ is a compound Poisson process, where positive jumps with expected jump size $\frac{1}{\alpha^+}$ are arriving at rate $\lambda^+$, while negative jumps with expected jump size $\frac{1}{\alpha^-}$ are arriving at rate $\lambda^-$. 

The cumulant generating function of a compound Poisson process with exponential jumps is $\frac{\lambda \beta u}{\alpha -u}$, which is defined for $u < \alpha$. Hence for $u \in (-\alpha^-,\alpha^+)$ the moment generating function of the combined process $L$ is
\begin{align*}
\EV{\e^{u L_1}} & 
= \EV{\e^{u L^+_1 }} \EV{\e^{ - u L^-_1}} 
= \ex{\lambda \frac{(\beta^+ + \beta^-) u^2 + (\beta^+ \alpha^- - \beta^- \alpha^+) u}{(\alpha^+ -u)(\alpha^- +u)}}.
\end{align*}
Inserting this into \eqref{eq:phiOU} gives
\begin{align*}
\phi_t(u) = & \frac{1}{\lambda} \int_{\e^{-\lambda t}}^1  \frac{\kappa(v u)}{v} \dd{v} = \int_{\e^{-\lambda t}}^1 \frac{(\beta^+ + \beta^-) u^2 v + (\beta^+ \alpha^- - \beta^- \alpha^+) u}{(\alpha^+ -u v)(\alpha^- +u v)} \dd{v} \\
= & - \frac{\beta^+ + \beta^-}{2} \int_{\e^{-\lambda t}}^1 \frac{-2 u^2 v + (\alpha^+ - \alpha^-) u}{(\alpha^+ -u v)(\alpha^- +u v)} \dd{v} \\
& + \int_{\e^{-\lambda t}}^1 \frac{(\beta^+ \alpha^- - \beta^- \alpha^+) u + (\beta^+ + \beta^-) (\alpha^+ - \alpha^-) \frac{u}{2} }{(\alpha^+ -u v)(\alpha^- +u v)} \dd{v}
 \end{align*}
The first integral is the differentiated logarithm of the denominator, so for this integral we get
$$ - \frac{\beta^+ + \beta^-}{2} \logn{(\alpha^+ -u v)(\alpha^- +u v)} \big \vert_{v=\e^{-\lambda t}}^1 = \frac{\beta^+ + \beta^-}{2} \logn{\frac{(\alpha^+-\e^{-\lambda t}u)(\alpha^-+\e^{-\lambda t}u)}{(\alpha^+-u)(\alpha^-+u)}}. $$
The second integral is
$$ \frac{\beta^+ + \beta^-}{2} \int_{\e^{-\lambda t}}^1 \frac{u}{\alpha^+ - vu} + \frac{u}{\alpha^- + vu} \dd{v},$$
so that the integrands are again differentiated logarithms of their denominators. Hence we get that the integral is
$$(-\logn{\alpha^+ - vu} + \logn{\alpha^- + vu} ) \vert_{v=\e^{-\lambda t}}^1 =  \logn{\frac{(\alpha^+-\e^{-\lambda t}u)(\alpha^-+u)}{(\alpha^+-u)(\alpha^-+\e^{-\lambda t}u)}}.$$
In summary the function $\phi$ is given by
\begin{equation} \label{eq:phijump}
\begin{aligned}
\phi_t(u) & = \frac{\beta^+ + \beta^-}{2} \logn{\frac{(\alpha^+-\e^{-\lambda t}u)(\alpha^-+\e^{-\lambda t}u)}{(\alpha^+-u)(\alpha^-+u)}} \\
& + \frac{\beta^+ - \beta^-}{2 } \logn{\frac{(\alpha^+-\e^{-\lambda t}u)(\alpha^-+u)}{(\alpha^+-u)(\alpha^-+\e^{-\lambda t}u)}}.
\end{aligned}
\end{equation}

It is also possible to combine the two approaches by considering an OU process generated by a Lévy process which is the difference of two compound Poisson processes plus a Brownian motion, all of which are independent. The resulting $\phi$ then follows by adding up the two functions \eqref{eq:phicont} and \eqref{eq:phijump} and for this process $\mathcal{V} = \{u \in \C: - \alpha^- < \mathrm{Re}(u)  < \alpha^+ \}$. By the previous remark it is also possible to shift this process by $\theta$. Such OU processes are used in the following numerical examples.

\subsubsection*{Volatility surfaces}
With the OU process of the previous section it is possible to generate volatility smiles as well as volatility skews. For illustration we consider a term structure with constant interest rates of $3.5\%$. The tenor structure and therefore the forward interest rates are based on half year intervals. Implied volatilities are then calculated for caplets with maturities over a 5-year period and strikes ranging from $0.02$ to $0.07$. Figure \ref{fig:OUskew} shows a skewed volatility surface while figure \ref{fig:OUsmile} shows a very pronounced smile, both of which are generated by an $OU$ process of the just introduced type. 
As mentioned in the previous chapters forward interest rates in this type of model will be bounded from below. The bounds in these examples are at $1\%$ for the forward interest rate expiring after half a year and decrease to basically $0\%$ for the forward interest rate which expires in 5 years. Hence they are well within reasonable boundaries. 
\begin{figure}
\centering
\includegraphics[width=0.7 \textwidth, trim=0cm 2cm 0cm 2.5cm,clip=true]{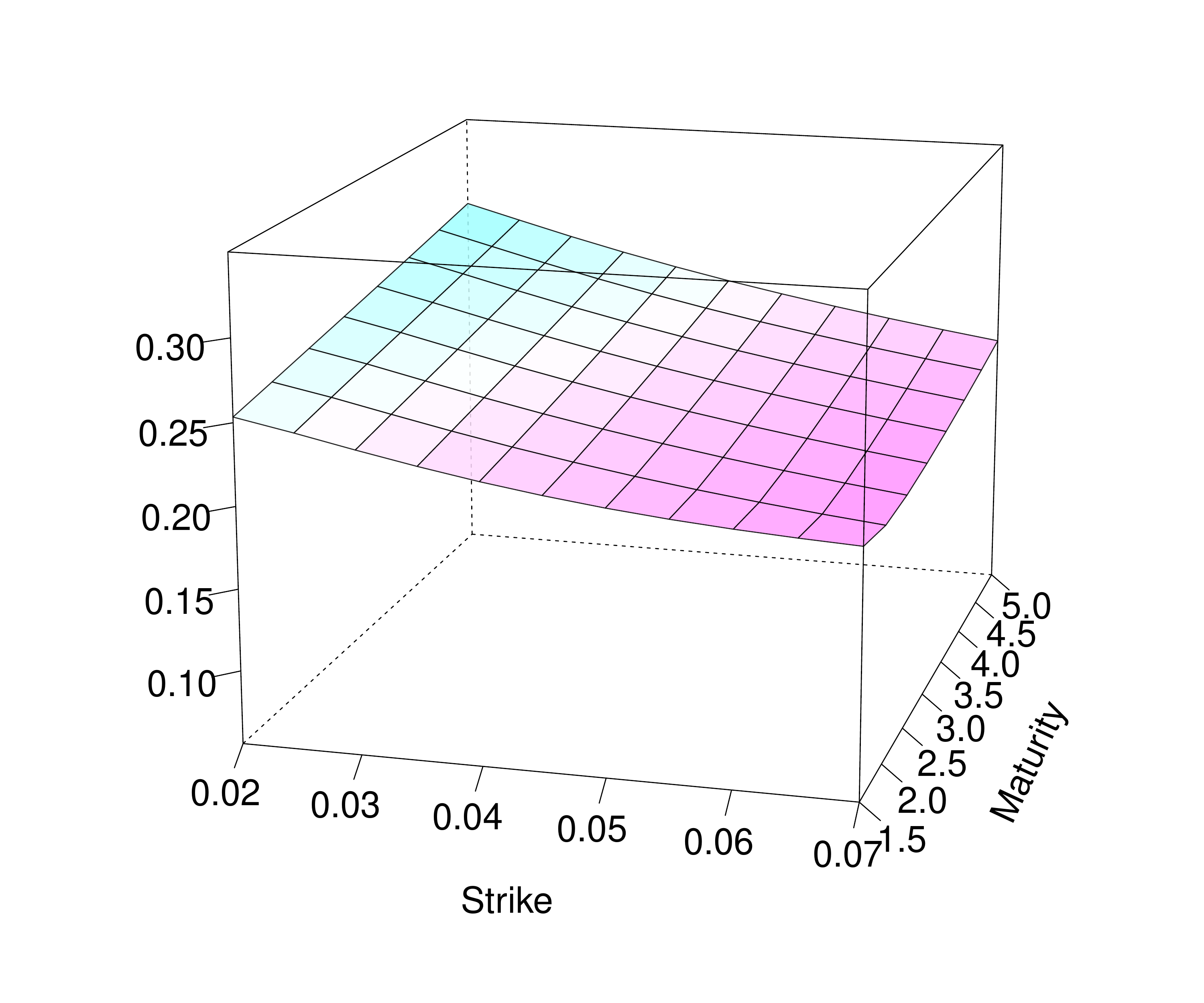}
\captionsetup{singlelinecheck=false, margin = 1.5 cm}
\caption[OU skew]{Implied volatility skew of caplets generated by an OU process with parameters $\lambda = 0.02, \alpha^+ = 12, \alpha^- =10, \beta^+=50, \beta^- = 5, \sigma=0.3, \theta=0.5, x=0.7$ and $T = 10$.}
\label{fig:OUskew}
\end{figure}
\begin{figure}
\centering
\includegraphics[width=0.7 \textwidth,trim=0cm 2cm 0cm 2.5cm,clip=true]{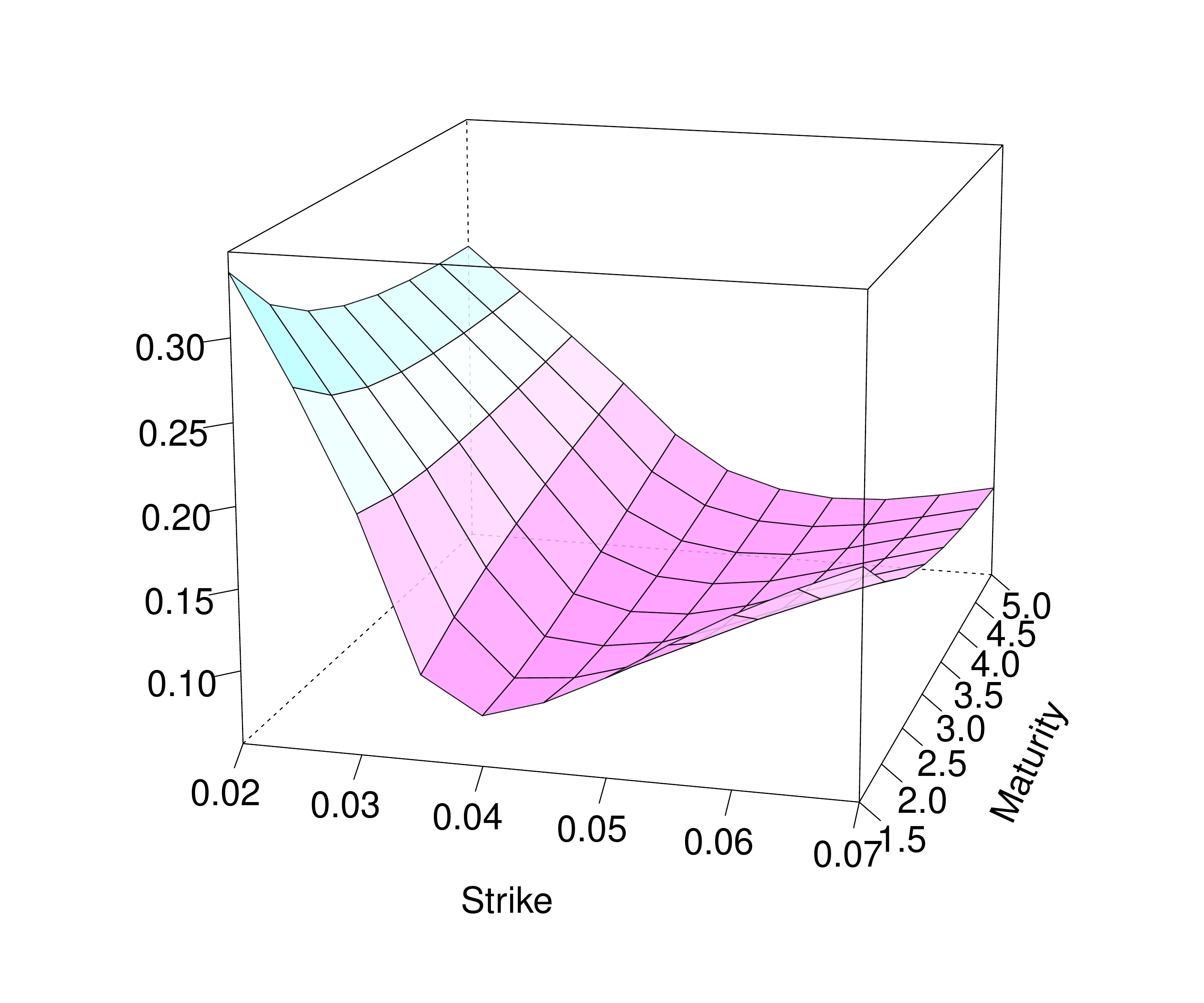}
\captionsetup{singlelinecheck=false, margin = 1.5 cm}
\caption[OU smile]{Implied volatility smile of caplets generated by an OU process with parameters $\lambda = 0.02, \alpha^+ = 50, \alpha^- =5, \beta^+=50, \beta^- = 10, \sigma=0, \theta=0, x=1$ and $T = 10$.}
\label{fig:OUsmile}
\end{figure}
For completeness an example of at-the-money implied volatilities for swaptions with maturities and underlying swap rates ranging from 2 to 7 years is displayed in figure \ref{fig:swaptionsurface}. 

\begin{figure}
\centering
\includegraphics[width=0.7 \textwidth]{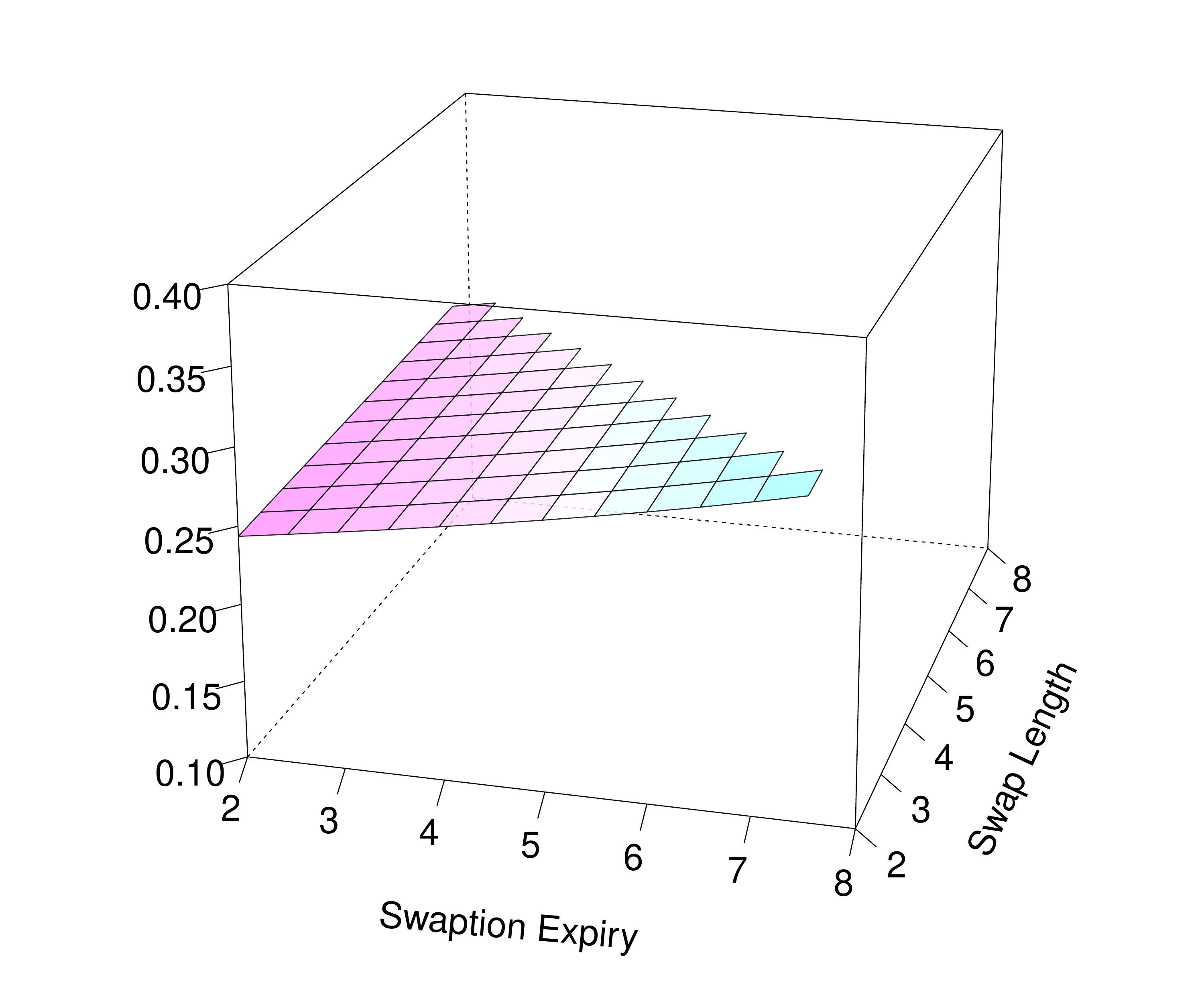}
\captionsetup{singlelinecheck=false, margin = 1.5 cm}
\caption[Swaption Implied volatilites]{Swaption implied volatilites generated by an OU process with parameters $\lambda = 0.02, \alpha^+ = 12, \alpha^- =10, \beta^+=50, \beta^- = 5, \sigma=0.3, \theta=0.5, x=0.7$ and $T = 10$.}
\label{fig:swaptionsurface}
\end{figure}

\section*{Conclusion}
Classical interest rate market models are not capable of simultaneously allowing for semi-analytical pricing formulas for caplets and swaptions and guaranteeing nonnegative forward interest rates. One exception are the affine LIBOR models presented in \citet{KPT11}. We modify their approach to also allow for driving processes which are not necessarily nonnegative. Caplet and swaption valuation is possible via one-dimensional numerical integration. This allows for a fast calculation of implied volatilities for these types of interest rate derivatives. With the additional flexibility of real-valued affine processes this type of model is capable of producing skewed implied volatility surfaces as well as implied volatility surfaces with pronounced smiles.

\section{Proofs} \label{sec:proofsCosh}
\begin{proof}[Proof of Lemma \ref{lem:monoton}]
For a function $f(x)$ denote its even and odd part by
$$f^e(x) = \frac{1}{2} (f(x) + f(-x)), \qquad f^o(x) = \frac{1}{2} (f(x) - f(-x)).$$
Note that if $f$ is monotonically increasing, the same is true for $f^o$.
Then \eqref{eq:coshmartingalefunction} can be written as
\begin{align*}
M_t^u(x) 
& = \frac{1}{2} \left(  \e^{\phi_{T-t}(u) + \psi_{T-t}(u) x}  +  \e^{\phi_{T-t}(-u) + \psi_{T-t}(-u) x} \right) \\
& =  \e^{\phi_{T-t}^e(u) + \psi_{T-t}^e(u) x} \cosh(\phi_{T-t}^o(u)+\psi_{T-t}^o(u) x)
\end{align*}
and
\begin{equation} \label{eq:Mufrac}
\frac{M_t^{u_{i}}(x)}{M_t^{u_{0}}(x)} =  \e^{(\phi_{T-t}^e(u_i) - \phi_{T-t}^e(u_{0}) )+ (\psi_{T-t}^e(u_i) -  \psi_{T-t}^e(u_{0})  ) x} \frac{\cosh(\phi_{T-t}^o(u_i)+\psi_{T-t}^o(u_i) x)}{\cosh(\phi_{T-t}^o(u_{0})+\psi_{T-t}^o(u_{0}) x)}.
\end{equation}
If $u_i = u_0$, then \eqref{eq:Mufrac} is constant and has no influence regarding monotonicity or maxima. Hence from now on assume $u_0 > u_i$ for all $i$. The function $g$ of equation \eqref{eq:monoton} can be written as
$$g(x) = \sum_{i=1}^n {c_i} \e^{A_i} \e^{a_i x} \frac{\cosh(B_i + b_i x)}{\cosh(B_0 + b_0 x)},$$
where for $i=0,\dots,n$
\begin{equation*}
\begin{aligned}
A_i & = (\phi_{T-t}^e(u_i) -  \phi_{T-t}^e(u_{0})), \qquad & B_i = \phi_{T-t}^o(u_i),  \\
a_i & = (\psi_{T-t}^e(u_i) -  \psi_{T-t}^e(u_{0})), & b_i = \psi_{T-t}^o(u_i). 
\end{aligned}
\end{equation*}

Since $\psi$ is monotonically increasing (see Lemma \ref{lem:psimontonicity}), also $\psi^o$ is monotonically increasing. With $\psi^o(0)=0$ it follows that $b_i \geq 0$ for all $i$. Furthermore note that $a_i < b_0 - b_i$ is equivalent to $\psi_{T-t}(u_i) < \psi_{T-t}(u_0)$ and $-a_i < b_0 - b_i$ is equivalent to $\psi_{T-t}(-u_0) < \psi_{T-t}(-u_i)$. Since $u_0 > u_i \geq 0$ the monotonicity of $\psi$ yields
\begin{equation} \vert a_i \vert < b_0 - b_i. \label{eq:aiest} \end{equation}

An elementary calculation gives
$$g^\prime(x) =  \frac{1}{ \cosh(B_0 + b_0 x)^2} \sum_{i=1}^n c_i \e^{A_i} \e^{a_i x} f_i(x), $$
where
\begin{align*} 
f_i(x) & = a_i \cosh(B_i+ b_i x)\cosh(B_0+b_0 x)+ b_i \sinh(B_i+b_i x) \cosh(B_0+b_0 x) \\
& - b_0 \cosh(B_i+b_i x) \sinh(B_0 +b_0 x).
\end{align*}
The derivative of $f_i$ is 
\begin{equation} \label{eq:gprime}
\begin{aligned} f_i^\prime(x)  =&   b_i \cosh(B_0 + b_0 x) \Big(a_i \sinh(B_i + b_i x) + (b_i - b_0) \cosh(B_i + b_i x)\Big) \\
 + &   b_0 \cosh(B_i + b_i x) \Big(a_i\sinh(B_0 + b_0 x) +  (b_i - b_0)  \cosh(B_0 + b_0 x)\Big) .
\end{aligned}
\end{equation}
Using \eqref{eq:aiest} the terms inside the brackets of each row in \eqref{eq:gprime} are strictly less than
$$ \vert a_i \vert (\sign{a_i} \sinh(B_j + b_j x) - \cosh(B_j + b_j x)) \leq 0, \qquad (j=i,0).$$
The last inequality is true since $\cosh(x) \pm \sinh(x) \geq 0$. The terms outside of the brackets in \eqref{eq:gprime} are all positive. Hence $f^\prime \geq 0$ and the $f_i$ are monotonically decreasing. Using \eqref{eq:aiest} a simple calculation shows that $ \lim_{x\rightarrow - \infty} f_i(x) =  \infty  \text{ and } \lim_{x\rightarrow  \infty} f_i(x) = -  \infty $. Since $c_i > 0$ for all $i$ the same is true for $\sum_{i=1}^n c_i \e^{a_i x} f_i(x)$. Hence $g$ has a single maximum and is decreasing to the left and right of it. Furthermore $g(x) \geq 0$ and again using \eqref{eq:aiest} $\lim_{x \rightarrow  \infty} g(x) = \lim_{x \rightarrow - \infty} g(x) = 0$.
\end{proof}

\begin{proof}[Proof of Lemma \ref{lem:fouriertransform}]
$f$ is continuous with compact support. Hence the extended Fourier transform $\hat{f}(z) = \int_{\R} f(x) \e^{- \i z x} \dd{x}$ exists for all $z \in \C$ and is analytic. For $z \neq 0, z \neq - \i v_k$, it is given by
\begin{align*}
\hat{f}(z)& = \int_{\kappa_1}^{\kappa_2} \e^{-\i z x} f(x) \dd{x} =  \frac{1}{\i z} \int_{\kappa_1}^{\kappa_2} \e^{-\i z x} f^\prime(x) \dd{x} \\
& = \frac{1}{\i z} \sum_k  {\frac{C_k v_k }{v_k - \i z}  { \left(\e^{(v_k - \i z) \kappa_2} - \e^{(v_k - \i z) \kappa_1 } \right) } }. 
\end{align*}
Since $\hat{f}(u - \i R) = O(u^{-2})$ for fixed $R$, it is absolutely integrable. By Fourier inversion 
\begin{equation*} \begin{split} f(x) & = \frac{1}{2 \pi} \int_{\mathrm{Im}(z) = - R} \e^{\i z x} \hat{f}(z) \dd{z} = \frac{1}{2 \pi} \int^\infty_0 \mathrm{Re}\left( \e^{(\i u + R) x} \hat{f}(u-\i R) \right) \dd{u},
\end{split} \end{equation*}
where the last equation follows from the fact that $f$ is real valued and the symmetry $\overline{\hat{f}(z)} = \hat{f}(-\overline{z})$. Since 
$\int \EV{\vert \e^{(\i z + R) X_t} \vert \vert \F_s}  \vert \hat{f}(z) \vert \dd{z} = \mathcal{M}_{X_t\vert X_s}(R) \int \vert \hat{f}(z) \vert \dd{z}$ is bounded if $R \in \mathcal{V} \cap \R$, conditional expectation and integration can be interchanged.
\end{proof}

%% file: AffineFwdCPI.tex
\newcommand{\defeq}{\mathrel{\mathop:}=}

\renewcommand{\T}{\cdot} 
\renewcommand{\Transp}[1]{\big( #1 \big) \cdot} 

Arbitrage-free inflation models have first been rigorously introduced in \citet{JY03}. Since then several inflation models have been proposed. Similar to interest rate models one can distinguish between short rate models and market models. While short rate models in the spirit of \citet{JY03} aim at modeling the unobservable continuous nominal and real short rate, market models use discrete observable rates as the basis for modeling (see \citet{BB04, ME05}). These observable rates are the basis of liquidly traded inflation swaps, zero coupon inflation-indexed swaps and year-on-year inflation-indexed swaps. While there have been several extensions of these models (e.g. \citet{MM06,MM09}) all of these models suffer the problem that there exist analytical formulas for only one type of swap, but not both. The model in this chapter leads to closed formulas for both types.

Based on the ideas in \citet{KPT11} one can use affine processes to describe analytically highly tractable models. Affine processes are Markov process, where the characteristic function is of exponentially affine form, i.e.
$$ \EV{\e^{u X_t} \vert X_s} = \e^{\phi_ {t-s}(u) + \psi_{t-s}(u) X_s}. $$
The class of affine processes contains a large number of processes, e.g. every Lèvy process is affine. Using the longest-dated nominal zero coupon bond as a numeraire we can model the normalized bond prices as \enquote{exponential martingales} with respect to an affine process $X$. In this type of models we are not only able to price both types of inflation swaps, but can also derive semianalytical formulas for calls and puts on the underlying inflation rates, another liquidly traded inflation derivative.

The structure of this chapter is as follows. The first part describes inflation markets and typical traded derivatives. Afterwards we outline the setup of an inflation market model. The second part introduces the affine inflation market model and derives pricing formulas for the introduced inflation derivatives. The third part provides a concrete model specification including the calibration to actual market data. In section \ref{sec:affineprocess} we collect the properties of affine processes needed in this chapter. Furthermore we specify the affine processes used in the numerical section.

\section{Inflation markets} \label{sec:inflprod}
Denote the time $t$ price of the nominal zero coupon bond with maturity $T$ by $P(t,T)$ and consider an inflation index with time $t$ value $I(t)$. Typically inflation indexes are so-called consumer price indexes (CPI). To shorten notation we will use the term CPI synonymous for inflation index, nevertheless the reader can think of an arbitrary inflation index. The basic mathematical instruments in inflation-linked markets are so called inflation-linked zero coupon bonds (corresponding to zero coupon bonds for nominal interest rate markets).
An inflation-linked zero coupon bond with maturity $T$ is a bond paying $I(T)$ at time $T$. Denote its price by $P_{ILB}(t,T)$.

In actual markets governments issue inflation-linked coupon bonds. Such a bond pays a fixed coupon on the variable basis $I(T_k)/I(T_0)$ at some fixed number of predetermined dates $T_k \leq T$ (typically annually), where $T_0$ is the time of issue. Additionally to the coupons such a bond redeems at maturity $T$ with value $\max\{I(T)/I(T_0),1\}$. Such a bond can therefore be described as a combination of inflation-linked zero coupon bonds plus an included option with payoff $(1-I(T)/I(T_0))_+$. In general these bonds are issued with maturities of several years and inflation is positive. In this case the included option has little influence on the total price, which is why it is market practice to mostly ignore it. In particular if one ignores these options it is possible to strip inflation-linked zero coupon bond prices out of actually traded inflation-linked coupon bonds by the same methods used for nominal quantities.

Consider the quantity 
\begin{equation} 
P_R(t,T) :=  \frac{P_{ILB}(t,T)}{I(t)},\label{eq:realbonddef}
\end{equation}
which is called the price of a real zero-coupon bond. Note that this is not the price\footnote{Here we mean price in terms of money which has to be paid at a transaction. In fact it is essentially this number that is quoted on trading screens. However, in case such a bond is traded, the cash-flow is then the quoted number multiplied by the according index ratio.} of a traded asset, but a theoretical quantity. The usage of the term price is motivated by the fact that this quantity can be viewed as the price of a zero coupon bond in a fictitious economy, where everything is measured in terms of the inflation index $I(t)$\footnote{One could interpret $I(t)$ as a numeraire, but one has to be careful not to use this as a mathematical numeraire, since $I(t)$ is not actually traded.}. Given real zero-coupon bond prices continuously compounded real interest rates are defined by
$R(t,T) := -{ \logn{P_R(t,T)} }/{(T-t)}.$ Accordingly one can define real counterparts to other nominal quantities such as forward interest rates or the short rate. This quantities are sometimes used as a starting point for inflation option pricing models (see e.g. \citet{JY03}, \citet{ME05}).

Next to inflation-linked bond markets there exist several liquidly traded inflation-linked derivatives. First consider the forward price of the inflation index (forward CPI), i.e. the at time $t$ fixed value $\I(t,T)$, which at time $T$ can be exchanged against $I(T)$ without additional costs.  Since $P_{ILB}(t,T)$ is the current price of $I(T)$, 
the time $t$ forward CPI for maturity $T$ is
\begin{equation}
\I(t,T) := \frac{P_{ILB}(t,T)}{P(t,T)} . \label{eq:fwdCPIdef}
\end{equation}
In a zero coupon inflation-indexed swap (ZCIIS) two parties exchange the realized inflation  $\frac{I(T)}{I(t)}$ against a fixed amount $(1+K)^{T-t}$. ZCIIS are mostly traded for full year maturities M. For $T = t + M$ the value of such a payer swap can be expressed as
\begin{equation} \label{eq:ZCIIS} P(t,T) \left( \frac{\I(t,T)}{I(t)} - (1+K)^M \right). \end{equation} The rate $K$, for which equation \eqref{eq:ZCIIS} is zero is then called the ZCIIS rate $ZCIIS(t; M)$. These ZCIIS rates are quoted in the market for several full-year maturities. 
\begin{remark} Note that ZCIIS rates and inflation-linked bonds are closely related via \eqref{eq:fwdCPIdef} and  \eqref{eq:ZCIIS}. In reality this relationship is not observed. This is partly due to different creditworthiness of counterparties in bond and swap markets. A more detailed analysis of this difference can be found in \citet{FL10}. For model calibration one has to choose one market, usually the swap market.
\end{remark}

Next to ZCIIS there is a second important type of swap in inflation markets, the year-on-year inflation-indexed swaps (YYIIS). These swaps exchange the annualized inflation against a fixed rate $K$, i.e. consider an annually spaced tenor structure $T_k=t+k, k=0,\dots,M$.
The netted payment of a payer YYIIS at time $T_k$ is $\left( {I(T_k)}/{I(T_{k-1})} -1 \right)  - K.$
Hence the inflation leg consists of payoffs of the form $$\frac{1}{T-S}  \left( \frac{I(T)}{I(S)} -1 \right).$$
Denote the forward value of such a payoff, the annualized forward inflation rate, by $F_I(t,S,T)$. Then the value of a payer YYIIS with maturity $M$ and strike $K$ can be expressed as
$$\sum_{k=1}^M P(t,T_k) (F_I(t,T_{k-1},T_k) -K).$$
The YYIIS rate $YYIIS(t; M)$ is the rate $K$ such that the corresponding YYIIS has zero value. Note that given YYIIS rates for all annual maturities one can calculate annual forward inflation rates $F_I(t,T_{k-1},T_k)$ and vice versa. 

Foward CPIs $\I(t,T)$, respectively forward inflation rates $F_I(t,S,T)$ are the mathematical quantities underlying the market-traded ZCIIS, respectively YYIIS. Inflation market models aim at modeling these quantities. In existing inflation market models either $\I(t,T)$ or $F_I(t,S,T)$ can be expressed by analytical formulas, but not both. Consider a market, where price processes are assumed to be semimartingales on a filtered probability space $(\Omega, \mathcal{A}, (\mathcal{F}_t), \PM)$. Fix a $T$-forward measure $\QM^T$, i.e. a probability measure equivalent to $\PM$ such that asset prices normalized with the numeraire price $P(t,T)$ are $\QM^T$-martingales. Then
\begin{align*}
\I(t,T) & = \EV[\QM^T]{I(T) \vert \F_t} \\
F_I(t,S,T) & =  \E^{\QM^T} \left[  \frac{1}{T-S}  \left( \frac{I(T)}{I(S)} -1 \right)  \bigg \vert \F_t \right]. \label{eq:fwdinfl} \end{align*}
Calculating the expectations of the inflation index, as well as the fraction of the inflation index at two different times proves difficult. In the model of this chapter both are exponentially affine in the underlying driving stochastic process. For an affine process (see section \ref{sec:affineprocess}) such expectations can be calculated and we are able to give semianalytical formulas for many standard options like caps and floors of forward inflation rates
.

\subsection{The inflation market model}
We now introduce the general setup of an (inflation) market model. 
Consider a tenor structure $0 < T_1 < \dots < T_N =: T$ and a market consisting of zero coupon bonds with maturities $T_k$ and prices $P(t,T_k)$. The price processes $(P(t,T_k))_{0 \leq t \leq T_k}$ are assumed to be positive semimartingales on a filtered probability space $(\Omega, \mathcal{A}, \F, \PM)$ (here $\F = (\mathcal{F}_t)_{0 \leq t \leq T}$), which satisfy $P(T_k,T_k) = 1$ almost surely.  If there exists an equivalent probability measure $\QM^{T}$ such that the normalized bond price processes $P(\cdot,T_k) / P(\cdot,T)$ are martingales\footnote{One can extend bond price processes to $[0,T]$ by setting $P(t,T_k) := \frac{P(t,T)}{P(T_k,T)}$ for $t > T_k$, so that $P(\cdot,T_k) / P(\cdot,T)$ is a martingale on $[0,T]$ if and only if it is a martingale on $[0,T_k]$. Economically this can be interpreted as immediately investing the payoff of a zero coupon bond into the longest-running zero coupon bond.
}, the market is arbitrage-free. This setup describes the class of interest rate market models like the classical LIBOR market model (\citet{BGM97}) and its extensions. 

To extend this setup to inflation markets consider an inflation index $I$, where we assume w.l.o.g. that $I(0)=1$. Assume there exist inflation-linked zero-coupon bonds with the same maturities\footnote{The assumption that for each zero coupon maturity there is a ILB with the same maturity is used only for notional convenience. 
} $T_1, \dots, T_N$ and price processes $P_{ILB}(t,T_k)_{0 \leq t \leq T_k}$, all of which are positive semimartingales\footnote{The inflation index is only described through the bond prices $P_{ILB}$. I.e. the distribution of $I(t)$ is only given at times $T_k$, where it coincides with the distribution of $P_{ILB}(T_k,T_k)$.}. If there exists an equivalent probability measure $\QM^T$ such that all normalized price processes
\begin{equation} \label{eq:martingaleassets}
\left( \frac{ P(t,T_k)}{P(t,T)} \right)_{0 \leq t \leq T_k}, \qquad \left(  \frac{ P_{ILB}(t,T_k)}{P(t,T)} \right)_{0 \leq t \leq T_k}
\end{equation}
are $\QM^T$-martingales, the extended market model is arbitrage-free. For given $\QM^T$ define the $T_k$-forward measures $\QM^{T_k}$ by
\begin{equation} \frac{\dd{\QM^{T_k}}}{\dd{\QM^{T}}} =  \frac{1}{P(T_k,T)} \frac{P(0,T)}{P(0,T_k)}.  \label{eq:measurechange}
\end{equation}
Under $\QM^{T_k}$ the forward interest rate
$$F^k(t) := \frac{1}{\Delta_k} \left(\frac{P(t,T_{k-1})}{P(t,T_k)} -1 \right), \qquad \Delta_k := T_k - T_{k-1}, $$
the earlier introduced forward CPI $$\I(t,T_k) = \frac{P_{ILB}(t,T_k)}{ P(t,T_k)},$$ 
and for $j < k$ the forward inflation rates $F_I(t,T_j,T_k)$ given by
$$1 + (T_k-T_j) F_I(t,T_j,T_k) = \EV[\QM^{T_k}]{\frac{I(T_k)}{I(T_j)} \vert \F_t}$$
are all martingales. Modeling (some of) these martingales is the starting point of inflation market models in the literature (see e.g. \citet{ME05}). In contrast we start by modeling the normalized bond prices in \eqref{eq:martingaleassets} and derive the above quantities thereof.

\section{The Affine inflation market model} \label{sec:affineCPI}
Let $(X_t)_{0 \leq t \leq T}$ with $X_0 = x$ be an analytic affine process with state space $\R^m_{\geq 0} \times \R^n$, $m > 0$, $n \geq 0$ on the probability space $(\Omega,\mathcal{A},\F,\QM^T)$ and define for $k=1,\dots,N$
\begin{equation} \label{eq:affineFwdCPImodel}
\begin{aligned}
 \frac{ P(t,T_k)}{P(t,T)} & := M_t^{u_k}, \qquad && u_k \in (\R_{\geq 0}^m \times \{0\}^n ) \cap \, \mathcal{V},  \\
 \frac{ P_{ILB}(t,T_k)}{P(t,T)} & := M_t^{v_k}, && v_k \in \R^{m+n} \cap \; \mathcal{V},
\end{aligned}
\end{equation}
where
\begin{equation} M_t^u := \EV[\QM^T]{\e^{u \cdot X_T} \vert \F_t} = \ex{\phi_{T-t}(u) + \psi_{T-t}(u) \cdot X_t},  \qquad u \in \mathcal{V}, \label{eq:affinemartingale} \end{equation}
with $\mathcal{V}$ defined in \eqref{eq:momset}. 
The processes $M_t^u$ are $\QM^T$-martingales by the definition of an affine process. Hence this model is arbitrage-free. 
Note that in \eqref{eq:affineFwdCPImodel} the parts of $u_k$ corresponding to real-valued components of $X$ are chosen to be zero. For a decreasing sequence $u_1 \geq \dots \geq u_N \geq 0$ one then has $M_t^{u_{k-1}} \geq M_t^{u_k}$, so that forward interest rates $F^k(t)$ are guaranteed to be nonnegative 
for all $k$. Contrary to interest rates\footnote{Although interest rates are currently negative in certain countries, interest rates are still bounded below by the costs of physically keeping money. We can incorporate bounds different from $0$ by setting $\frac{P(t,T_k)}{P(t,T)} := c_k M_t^{u_k}$. } inflation rates are not required to be nonnegative, which is why we do not restrict $v_k$ in \eqref{eq:affineFwdCPImodel}.
The values of $u_k$ and $v_k$ should be calibrated to fit the initial term structures, i.e. $M_0^{u_k} = P(0,T_k) / P(0,T)$ and $M_0^{v_k} = P_{ILB}(0,T_k) / P(0,T).$ By Lemma \ref{lem:ufitting} it follows that parameters $u_k$ fitting a current term structure with nonnegative forward interest rates can always be chosen to be decreasing. 
For multidimensional affine processes such sequences are far from unique. Concrete specifications how to choose $u_k$ and $v_k$ will be presented in section \ref{sec:numeric}.

The big advantage of this setup is that \eqref{eq:measurechange} in this case reads 
\begin{equation} \frac{\dd{\QM^{T_k}}}{\dd{\QM^{T}}} =  \frac{M_{T_k}^{u_k}}{M_{0}^{u_k}} = \frac{1}{M_{0}^{u_k}} \ex{\phi_{T-T_k}(u_k) + \psi_{T-T_k}(u_k)\T X_{T_k}}
\end{equation}
which is exponentially affine in $X$. In particular it is easy to check (see \citet{KPT11}) that for $0 \leq s \leq r$ and $\psi_{T-r}(u_k) + w \in \mathcal{V}$
\begin{equation}
\begin{aligned}
\EV[{\QM^{T_k}}]{\e^{w \cdot X_r} \vert \F_s}  =&  \ex{ \phi_{r-s}(\psi_{T-r}(u_k) + w) - \phi_{r-s}(\psi_{T-r}(u_k)) } \\
& \ex{\big( \psi_{r-s}(\psi_{T-r}(u_k) + w) - \psi_{r-s}(\psi_{T-r}(u_k)) \big) \cdot  X_s}. \label{eq:measurechar}
\end{aligned}
\end{equation}
Hence the moment generating function of $X$ is also known\footnote{This also shows that $X$ is a time-inhomogeneous affine process under $\QM^{T_k}$.} under different measures $\QM^{T_k}$.
Together with the exponential affine form of basic quantities  this is the reason why this model is analytically highly tractable. For example the forward rates $F^k$ satisfy 
$$(1+\Delta_k F^k(t)) = \frac{M_t^{u_{k-1}}}{M_t^{u_k}} = \e^{A(t,u_{k-1},u_k)+B(t,u_{k-1},u_k) \cdot X_t},$$ with
\begin{equation} \label{eq:AB}
\begin{aligned}
A(t,v,u) & := \phi_{T-t}(v) - \phi_{T-t}(u), \\
B(t,v,u) & := \psi_{T-t}(v) - \psi_{T-t}(u).
\end{aligned}
\end{equation}
Hence the $\QM^{T_k}$-extended moment generating function of $\logn{1+\Delta_k F^k(t)}$ can be calculated explicitly using \eqref{eq:measurechar}. 
The price of a caplet then follows using a Fourier-inversion formula (see \citet{KPT11}). Swaptions can also be dealt with (\citet{KPT11,GPSS14}) and so the most common interest rates derivatives can be calculated efficiently. We can use similar methods for inflation derivatives.

\subsection{Forward CPI and CPI options}
%
As mentioned before, the main advantage of this model is that for several important quantities the moment generating function is known under all forward measures $\QM^{T_k}$. Start by looking at the forward CPI
\begin{equation} \label{eq:fwdCPIaffine}
\I(t,T_k) = \frac{ P_{ILB}(t,T_k)}{P(t,T_k)} =  \frac{ P_{ILB}(t,T_k)}{P(t,T)}  \frac{ P(t,T)}{P(t,T_k)} = \frac{M_t^{v_k}}{M_t^{u_k}} =  \e^{A(t,v_k,u_k) + B(t,v_k,u_k) \cdot X_t},
\end{equation}
with $A$ and $B$ defined in \eqref{eq:AB}. Hence the forward CPI is of exponential affine form and therefore the $\QM^{T_k}$-moment generating function of its logartihm can be calculated using \eqref{eq:measurechar}. In particular, setting $A_I^k := A(T_k,v_k,u_k),  B_I^k := B(T_k,v_k,u_k)$ and using $I(T_k) = \I(T_k, T_k)$ one has
\begin{align*}
\mathcal{M}^{\QM^{T_k}}_{\logn{I(T_k)} \vert \F_s}(z) \defeq & \; \EV[\QM^{T_k}]{I(T_k)^z \vert \F_s} =  \EV[\QM^{T_k}]{\ex{z A_I^k + z B_I^k\T X_{T_k} }\vert \F_s} \\
 = & \ex{z  A_I^k + \phi_{T_k -s}(\psi_{T-T_k}(u_k) + z B_I^k) - \phi_{T_k -s}(\psi_{T-T_k}(u_k)) } \\
& \ex{\big(\psi_{T_k -s}(\psi_{T-T_k}(u_k) + z B_I^k)-\psi_{T_k -s}(\psi_{T-T_k}(u_k)) \big)\T X_s} \\
 = &\ex{z \phi_{T-T_k}(v_k) + (1-z)  \phi_{T-T_k}(u_k)} \\
&   \ex{ \phi_{T_k -s}\big(z \psi_{T-T_k}(v_k) + (1-z) \psi_{T-T_k}(u_k)\big)} \\
& \ex{ \psi_{T_k -s}\big( z \psi_{T-T_k}(v_k) + (1-z) \psi_{T-T_k}(u_k) \big)\T X_s} / M_s^{u_k}.
\end{align*}
Here the last equality follows using \eqref{eq:semiflowinfl}. Note that this function is well defined and analytic in $z$ if $z \psi_{T-T_k}(v_k) + (1-z) \psi_{T-T_k}(u_k) \in  \mathrm{int}(\mathcal{V})$.

Given the moment generating function of $\logn{I(T_k)}$ CPI calls and puts can be calculated using the following well-known Fourier inversion formula (see e.g. \citet{EGP10}). 
If $R \in (1,\infty)$ such that $\mathcal{M}_{X \vert \F}(R) < \infty$, then
\begin{equation}  \label{eq:FourierCall}
\begin{aligned}
\EV{(e^X - K)_+ \vert \F} & = \frac{K}{\pi} \int_0^{\infty} \mathrm{Re} \left(  { \mathcal{M}_{X \vert \F}(\i u + R)} \frac{K^{-(\i u+R)}}{(\i u +R)(\i u +R -1)} 
\right) \dd{u}.
\end{aligned}
\end{equation}
Thus the price of a forward CPI call with maturity $T_k$ and payoff $(I(T_k) - K)_+$ is
\begin{align*}
\mathrm{CPI}&\mathrm{Call}(t,T_k,K) \\
&  = \frac{{K} P(t,T_k)}{\pi} \int_0^{\infty} \mathrm{Re} \left(  { \mathcal{M}^{\QM^{T_k}}_{\logn{I(T_k)} \vert \F_s}(\i u + R)} \frac{{K}^{-(\i u+R)}}{(\i u +R)(\i u +R -1)} \right) \dd{u},
\end{align*}
where $R>1$ is chosen to satisfy $R \psi_{T-T_k}(v_k) + (1-R) \psi_{T-T_k}(u_k) \in \mathrm{int}(\mathcal{V})$.
\begin{remark}
In section \ref{sec:inflprod} it is mentioned that ILBs usually come with an included option guaranteeing a redemption of at least the original nominal amount. For an ILB issued at $S$ and with maturity $T_k$ this translates into an option $(1 - I(T_k)/I(S))_+$ which corresponds to $1/I(S)$ CPI puts with strike $I(S)$. 
\end{remark}
 
\subsection{Forward Inflation and inflation caplets and floorlets}
Typically inflation market models are not able to handle both forward CPI and forward inflation products analytically. With the approach presented here forward inflation rates are of a similar form as forward CPIs. The annualized inflation $F_I(T_k,T_{k-j},T_k)$ satisfies
\begin{align}
1+ (T_k-T_{k-j}) F_I(T_k,T_{k-j},T_k) = \frac{I(T_k)}{I(T_{k-j})} = \e^{A_I^k + B_I^k \cdot X_{T_k} - A_I^{k-j} - B_I^{{k-j}} \cdot X_{T_{k-j}} } =: \e^{Y^k}. \label{eq:Yk}
\end{align}

The following Lemma gives the moment generating functions for random variables of this type. 
\begin{lemma} \label{lem:daffine}
Let $s \leq r \leq t \leq T$ and $$\psi_{T-t}(u_k)+w \in \mathcal{V} \text{ and } \psi_{t-r}(\psi_{T-t}(u_k)+ w) -  \psi_{T-r}(u_k) + u \in \mathcal{V}.$$ Then
\begin{equation}
\begin{aligned}
& \EV[\QM^{T_k}]{\e^{u\T X_r + w\T X_t} \vert \F_s }  =  \ex{\big(\psi_{r-s}(\psi_{t-r}(\psi_{T-t}(u_k)+w) + u) - \psi_{T-s}(u_k) \big)\T X_s} \\
& \qquad  \ex{\phi_{t-r}(\psi_{T-t}(u_k)+w) + \phi_{r-s}(\psi_{t-r}(\psi_{T-t}(u_k)+w)+ u) - \phi_{t-s}(\psi_{T-t}(u_k)) } .
\end{aligned}
\end{equation}
\end{lemma}
\begin{proof}
Using the tower property and applying \eqref{eq:measurechar} twice it follows
\begin{align*}
\E^{\QM^{T_k}}  \Big[ & \e^{u\T X_r + w\T X_t} \big \vert \F_s \Big] =  \EV[\QM^{T_k}]{\EV[\QM^{T_k}]{\e^{w\T X_t} \big \vert \F_r}  \e^{u\T X_r} \big \vert \F_s} \\
= & \E^{\QM^{T_k}} \Big[ \ex{ \phi_{t-r}(\psi_{T-t}(u_k)+w) -  \phi_{t-r}(\psi_{T-t}(u_k))} \\
&  \ex{\big( \psi_{t-r}(\psi_{T-t}(u_k)+w) -  \psi_{t-r}(\psi_{T-t}(u_k)) + u \big)\T X_r} \big \vert \F_s \Big] \\
=  & \ex{ \phi_{t-r}(\psi_{T-t}(u_k)+w) -  \phi_{t-r}(\psi_{T-t}(u_k)) } \\
& \ex{\phi_{r-s}(\psi_{t-r}(\psi_{T-t}(u_k)+w) + u) - \phi_{r-s}(\psi_{T-r}(u_k)) } \\
& \ex{\big(\psi_{r-s}(\psi_{t-r}(\psi_{T-t}(u_k)+w) + u) - \psi_{r-s}(\psi_{T-r}(u_k)) \big)\T X_s} \\
= & \ex{ \phi_{t-r}(\psi_{T-t}(u_k)+w) + \phi_{r-s}(\psi_{t-r}(\psi_{T-t}(u_k)+w)+ u) - \phi_{t-s}(\psi_{T-t}(u_k)) } \\
& \ex{\big(\psi_{r-s}(\psi_{t-r}(\psi_{T-t}(u_k)+w) + u) - \psi_{T-s}(u_k) \big)\T X_s}.
\end{align*}
Here we used the semiflow property \eqref{eq:semiflowinfl} to simplify the expression.
\end{proof} 

By Lemma \ref{lem:daffine} the $\QM^{T_k}$-moment generating function of $Y^k$ defined in \eqref{eq:Yk} is
\begin{align*}
\mathcal{M}^{\QM^{T_k}}_{Y^k \vert \F_s}(z) 
 = & \EV[\QM^{T_k}]{ \ex{z A_I^k + z  {B_I^k} \cdot X_{T_k} - z A_I^{k-j} - z  {B_I^{k-j}} \cdot X_{T_{k-j}} } \vert \F_s } \\
= &  \ex{z A_I^k - z A_I^{k-j} +  \phi_{T_k-T_{k-j}}(\psi_{T-T_k}(u_k)+z B_I^k) } \\
& \ex{\phi_{T_{k-j}-s}(\psi_{T_k-T_{k-j}}(\psi_{T-T_k}(u_k)+z B_I^k) - z B_I^{k-j}) - \phi_{T_k-s}(\psi_{T-T_k}(u_k)) } \\
& \ex{\Transp{\psi_{T_{k-j}-s}(\psi_{T_k-T_{k-j}}(\psi_{T-T_k}(u_k)+ z B_I^k) -  z B_I^{k-j}) - \psi_{T-s}(u_k)} X_s},
\end{align*}
which is well-defined if \begin{equation} \left \{ \psi_{T-T_k}(u_k)+z B_I^k, \psi_{T_k-T_{k-j}}(\psi_{T-T_k}(u_k)+z B_I^k) -  \psi_{T-T_{k-j}}(u_k) - z B_I^{k-j} \right \} \subset \mathcal{V}. \end{equation}
The forward inflation rate is then given by $1+(T_k-T_{k-j}) F_I(t,T_{k-j},T_k) = \mathcal{M}^{\QM^{T_k}}_{Y^k \vert \F_t}(1)$.
So for $\psi_{T-T_{k-j}}(v_{k})-B_I^{k-j} \in \mathcal{V}$ it is 
\begin{align*}
1+(T_k-T_{k-j}) & F_I(t,T_{k-j},T_k) =
 \ex{\Transp{\psi_{T_{k-j}-t}(\psi_{T-T_{k-j}}(v_{k})-B_I^{k-j})-\psi_{T-t}(u_k)}X_t} \\
& \ex{\phi_{T-T_{k-j}}(u_{k-j}) + \phi_{T_{k-j}-t}(\psi_{T-T_{k-j}}(v_{k})-B_I^{k-j} )+ \phi_{T-t}(u_k)}.
\end{align*}
Furthermore the payoff of an inflation caplet with strike $K$ is 
\begin{align*}
(T_k - T_{k-j}) (F_I(T_k,T_{k-j},T_k)-K)_+ =  \left( \frac{I(T_k)}{I(T_{k-j})} - \tilde{K} \right)_+
\end{align*}
where $\tilde{K} = 1 + (T_k - T_{k-j}) K$. With the Fourier inversion formula \eqref{eq:FourierCall} one can calculate the price of an inflation caplet. In particular, we have for $R>1$ 
\begin{align*}
\mathrm{InflCpl}& (t,T_{k-j},T_k,K) 
\\& = \frac{\tilde{K} P(t,T_k)}{\pi} \int_0^{\infty} \mathrm{Re} \left(  { \mathcal{M}^{\QM^{T_k}}_{Y^k \vert \F_t}(\i u + R)} \frac{\tilde{K}^{-(\i u+R)}}{(\i u +R)(\i u +R -1)} \right) \dd{u},
\end{align*}
provided that $\psi_{T-T_k}(u_k)+R B_I^k \in \mathrm{int}(\mathcal{V})$ and $\psi_{T_k-T_{k-j}}(\psi_{T-T_k}(u_k)+R B_I^k) -  \psi_{T-T_{k-j}}(u_k) - R B_I^{k-j} \in \mathrm{int}(\mathcal{V})$.

\subsection{Correlation} \label{sec:correlation}
So far we considered the pricing of typical market traded options. Another important aspect is the correlation structure. The relevant quantities
\begin{equation} \label{eq:logquants}
\begin{aligned}
\logn{1 + \Delta_k F_n^k(t)}  & =  \logn{\frac{M_t^{u_{k-1}}}{M_t^{u_k}}} ={A(t,u_{k-1},u_k) + B(t,u_{k-1},u_k) X_t}, \\
\logn{\I^j(t)} & =  \logn{ \frac{M_t^{v_j}}{M_t^{u_j}} } ={A(t,v_j,u_j) + B(t,v_j,u_j) X_t}. 
\\ \logn{1 + \Delta_k F_I(t,T_{k-j},T_k)} & = \text{const} + 
\big(\psi_{T_{k-j}-t}(\psi_{T-T_k}(v_k)  - B_I^{k-j}) - \psi_{T-t}(u_k) \big)\T X_t. 
\end{aligned}
\end{equation}
are all affine transformation of $X_t$ and the correlation for two such terms is 
\begin{align*}
\Cor [A_t + B_t \cdot X_t,\tilde{A}_t + \tilde{B}_t \cdot X_t] & = \frac{\Var{B_t \cdot X_t,\tilde{B}_t \cdot X_t }}{\sqrt{\Var{B_t \cdot X_t}}\sqrt{\Var{\tilde{B}_t \cdot X_t }}}.
\end{align*}
For independent components of $X_t$ this simplifies to\footnote{Up to some technical conditions the variance of a one-dimensional affine process is
$$\Var{X_t^i} =  \left. \frac{\partial^2}{\partial^2 u} \right \vert_{u=0}  (\phi_t^i(u) + \psi_t^i(u) X_0^i) .$$}
\begin{align} \label{eq:corrstructure}
 \frac{\sum_{i=1}^d B_t^i \tilde{B}_t^i \Var{X_t^i}}{\sqrt{\sum_{i=1}^d (B_t^i)^2 \Var{X_t^i}}\sqrt{\sum_{i=1}^d (\tilde{B}_t^i)^2 \Var{X_t^i} }}.
\end{align}
Hence correlations strongly depend on $B(t,u_{k-1},u_k) = \psi_{T-t}(u_{k-1}) - \psi_{T-t}(u_k)$ and $B(t,v_j,u_j) = \psi_{T-t}(v_j) - \psi_{T-t}(u_j)$, respectively the structure of the $v_k$ and $u_k$. The exact correlation depends on the used measure (e.g. $\QM^{T_k}, \PM$), but choosing $v_k$ and $u_k$ cleverly, one can guarantee that the correlation structure, i.e. the correlation signs, stay the same. Concrete specifications for meaningful correlation structures will be given in the next section. Similar observations can also be made for instantaneous correlations of the corresponding quantities in the case of continuous affine processes (see \citet{GPSS14} for the general idea). 

\section{Implementation example} \label{sec:numeric}  \label{sec:calibrationexample}
We design the structure of the affine inflation market model in such a way that the calibration can be separated into the calibration to nominal market data and the calibration to inflation market data afterwards. The method used to calibrate to nominal market data is based on the ideas in \citet{GPSS14}. There they fit a multiple curve affine LIBOR market model by using a common driving process $X^0$ plus additional driving processes $X^1, \dots, X^M$, all of which are independent, affine and nonnegative. 
For calibration they use caplets with full-year maturities and underlying forwards of tenors less than a year. Using one individual driving process for each year one can then use an iterative procedure to calibrate to market data. Their approach can also be used in this setup. 

In particular consider a semiannual tenor structure $T_k = k/2, k=1,\dots,N$, $N$ even, and a driving affine process consisting of $M+1 = N/2 +1$ components $X^0, X^1, \dots ,X^M$, all of which are independent analytic affine processes with functions $\phi^i$ and $\psi^i$, $i=0,\dots,M$. Then by \eqref{eq:affinecombphipsi}
\begin{align*}
\phi_t(u) & = \sum_{i=0}^M \phi_t^i(u^i) \\
\psi_t(u) & = (\psi_t^0(u^0),\psi_t^1(u^1), \dots, \psi_t^M(u^M)),
\end{align*}
where $u^i, i=0, \dots, M$ denotes the corresponding component of $u \in \R^{M+1}$. To describe an affine inflation market model we also have to specify the vectors $u_k$. The structure of the vectors should be so that the following points are satisfied.
\begin{itemize}
\item Forward interest rates are nonnegative. This is the case if $0 \leq u_k \leq u_{k-1}$. 
\item The model matches the initial interest rate term structure. This basically fixes one component of each vector $u_k$.
\item Calibration to market data is possible by using an iterative procedure. 
\item The model has a meaningful correlation structure.
\end{itemize}
\begin{table}
\begin{center}
\begin{tabular}{l | C{2.5em}:C{2.5em}C{2.5em}C{2.5em}C{2.5em}:C{3.4 em}C{3.4 em}C{2.5em}C{2.5em}  }
& $X^0$ & $X^1$ & $X^2$ & \dots & $X^{M}$ & $X^{M+1}$ & $X^{M+2}$ & \dots & $X^{2M}$ \\ \hdashline
$u_1$ & $\tilde{u}_1$ & $\overline{u}_1$ & $\overline{u}_3$ & $\dots$ & $\overline{u}_{N-1}$ & 0 & $0$ & $\cdots$ & $0$ \\
$u_2$ & $\tilde{u}_2$ & $\overline{u}_2$ & $\overline{u}_3$ & $\dots$ & $\overline{u}_{N-1}$ & 0 & $0$ & $\cdots$ & $0$ \\
$u_3$ & $\tilde{u}_3$ & 0 & $\overline{u}_3$ & $\dots$ & $\overline{u}_{N-1}$ & $0$ & 0  & $\cdots$ & $0$ \\
$u_4$ & $\tilde{u}_4$ & 0 & $\overline{u}_4$ & $\dots$ & $\overline{u}_{N-1}$ & $0$ & 0  & $\cdots$ & $0$ \\
\vdots & \vdots & \vdots & \vdots & $\ddots$ & \vdots & \vdots & \vdots & $\ddots$ & \vdots \\
$u_{N-2}$ & $\tilde{u}_{N-2}$ & 0 & 0 & $\dots$ & $\overline{u}_{N-1}$ & $0$ &0 & $\cdots$ & 0 \\
$u_{N-1}$ & $\tilde{u}_{N-1}$ & 0 & 0 & $\dots$ & $\overline{u}_{N-1}$ & $0$ &0 & $\cdots$ & 0 \\
$u_{N}$& $\tilde{u}_{N}$ & 0 & 0 & $\dots$ & $\overline{u}_{N}$ & $0$ &0 & $\cdots$ & 0 \\
\end{tabular}
\end{center}
\captionsetup{singlelinecheck=false, margin = .5 cm}
\caption{Description of the parameter structure $u_k$. Each row corresponds to one vector with the column names denoting the process the position in the vector corresponds to.}
\label{tab:upar}
\end{table}
This can be achieved by choosing the vectors $u_k$ in the following way. They depend on $2N$ real parameters 
$$\tilde{u}_1 \geq \dots \geq \tilde{u}_N \geq 0, \qquad \overline{u}_1 \geq \dots \geq \overline{u}_N \geq 0.$$
For $1 \leq j \leq M$ set (compare table \ref{tab:upar}, ignoring the zero columns in the rightmost side for now)
$$ u_k = \tilde{u}_k \e^0 + \overline{u}_k e^{\ceil{\frac{k}{2}}} + \sum_{l= \ceil{\frac{k}{2}}+1 }^M \overline{u}_{2l-1} e^l$$
where $e^0,e^1,\dots,e^M$ denote the base vectors $(1,0, \dots, 0), \dots, (0,\dots,0,1)$ of $\R^{M+1}$. Note that with this choice the vectors $(u_k)$ are decreasing. Given $\tilde{u}_1,\dots,\tilde{u}_n$ and the processes $X^0,X^1,\dots,X^M$ by Lemma \ref{lem:ufitting} the parameters $\overline{u}_1,\dots,\overline{u}_N$ are determined by fitting the current term structure. I.e. we require that $$\frac{P(0,T_k)}{P(0,T)} =  \EV[\QM^T]{\e^{u_k\T X_T}} = 
\EV[\QM^T]{\e^{\tilde{u}_k X^0_T}} \EV[\QM^T]{\e^{\overline{u}_k X^{\ceil{\frac{k}{2}}}_T}}  \prod_{i= \ceil{\frac{k}{2}}+1}^M \EV[\QM^T]{\e^{\overline{u}_{2l-1} X^l_T}},$$
so that the parameters $\overline{u}_1,\dots,\overline{u}_N$ can be calculated using backwards iteration.
Furthermore the semiannual forward interest rates with full-year maturities $F^{2k}$ only depend on $u_{2k-1},u_{2k}$ and the processes $X^0, X^k, \dots, X^{M}$. Hence if $X^0$ and $\tilde{u}_1, \dots, \tilde{u}_N$ are already specified, one can fit $X^{M}$ to caplets on the forward interest rate $F^{N}$ and then go backwards to iteratively fit the processes $X^k$ to caplets on forward interest rates $F^{2k}$. Hence if $X^0$ and  $\tilde{u}_1, \dots, \tilde{u}_N$ are fixed, all the remaining parameters can be calibrated to the yield curve and caplet prices. As stated in \citet{GPSS14} and confirmed by our numerical tests the concrete choice of $X^0$ and $\tilde{u}_k$ (within a meaningful range) has no qualitative impact on the resulting calibration quality. Henceforth $X^0$ is fixed as a CIR process (specifications of this process can be found equation \eqref{eq:CIR}). So far we have said nothing about resulting correlations. This is where the choice of the parameters $\tilde{u}_1, \dots, \tilde{u}_n$ comes in.
The relevant functions for correlations of forward interest rates $F^{2k}$ are
\begin{equation*}
B(t,u_{2k-1},u_{2k}) = ( \psi_{T-t}^0(\tilde{u}_{2k-1}) - \psi_{T-t}^0(\tilde{u}_{2k}), 0, \dots, 0, \psi_{T-t}^{k}(\overline{u}_{2k-1})- \psi_{T-t}^{k}(\overline{u}_{2k}), 0, \dots, 0).
\end{equation*}
These functions are \enquote{orthogonal} except for the first component\footnote{Although only stated for even forward rates for notional simplicity this is true for all forward rates.}. Hence by equation \eqref{eq:corrstructure} the correlation structure mainly depends on the sequence $(\tilde{u}_k)$.
Since $X^0$ is nonnegative the function $\psi^0_t(u)$ is increasing in $u$ (see \citet{KPT11}). With $(\tilde{u}_k)$ being a decreasing sequence this results in nonnegative correlations. Setting $\tilde{u}_k = \tilde{u}$ for all $k$ results in zero correlation\footnote{Negative correlations in this setup are only possible if the sequence $(\tilde{u}_k)$ is not decreasing which means that forward interest rates can become negative.}. The magnitude of correlations depends on how much of the variance of forward bond prices is explained by $\tilde{u}_k X^0$. Here we choose the factors $\tilde{u}_k$ so that $\EV[\QM^T]{\e^{2 \tilde{u}_k X_T^c}} = P(0,T_k)/P(0,T)$. The idea behind this choice is that approximately half of the variance should be explained by the common factor. Alternative if one has additional information on correlations (e.g. through market data), this could be incorporated in $\tilde{u}_k$. 

For calibration we used market data from September, 29th 2011. The yield curve is bootstrapped from LIBOR and swap rates and is displayed in figure \ref{fig:yieldcurve}. The used caplet implied volatilities for 6-month forward interest rates with strikes $1\%$ to $6\%$ and maturities from 1 to 10 years are bootstrapped from cap data. The resulting implied volatilities can be found in figure \ref{fig:capletfit}. 
\begin{figure}
\centering
\includegraphics[width=0.8 \textwidth]{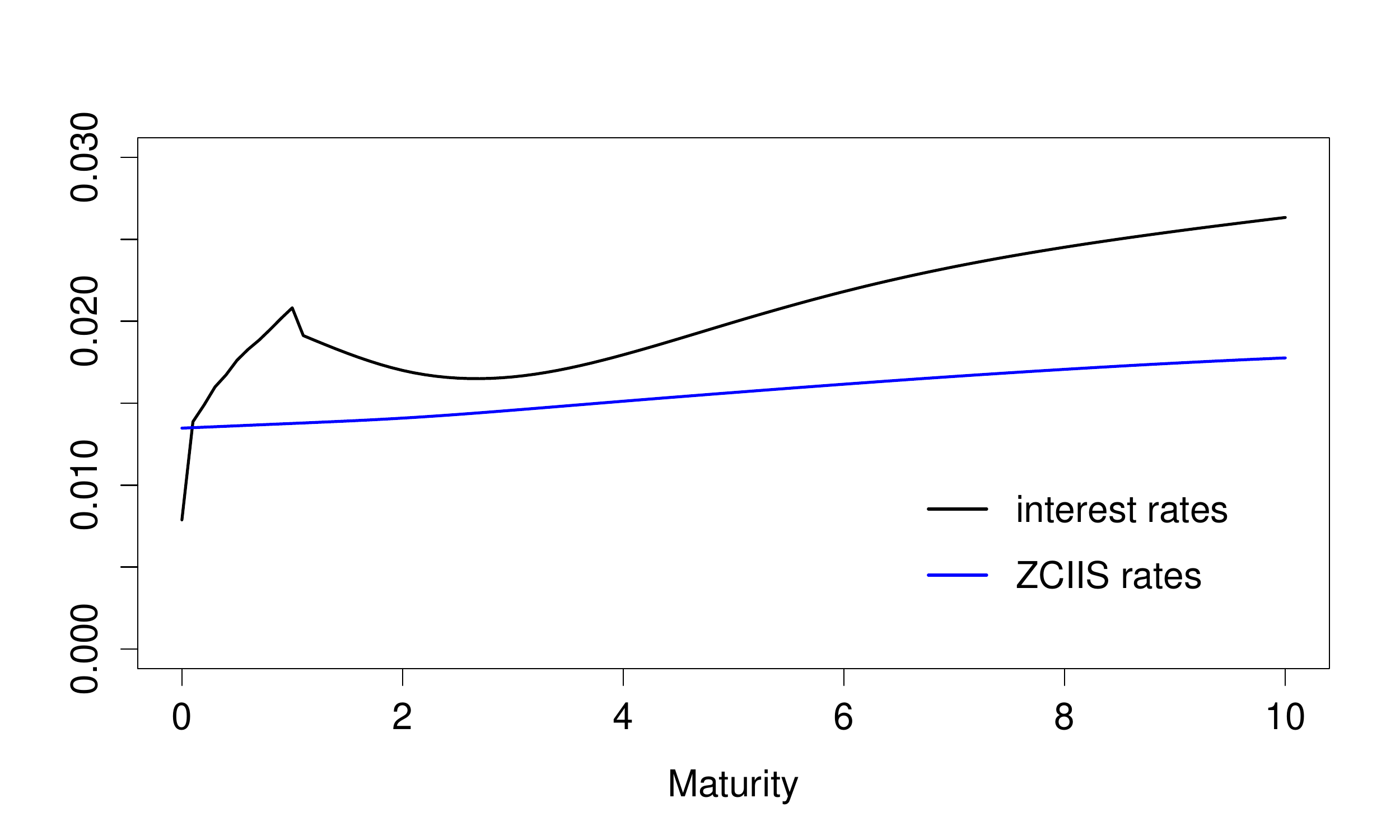}
\caption{Yield curve and ZCIIS curve from September, 29th 2011}
\label{fig:yieldcurve} \label{fig:ZCIIScurve}
\end{figure}
We fixed the time horizon $T=10$ and chose the parameters of the common CIR process as $\lambda=0.026,\theta=0.65,\eta=0.5,x=3.45$. The $M$ individual driving processes are chosen to be CIR processes with added jumps (see equation \eqref{eq:CIRGamOU}). Their parameters and the parameters $\overline{u}_k$ are calibrated with the mentioned recursive method. In each step the parameters are chosen so that the mean squared errors of implied volatilities are minimized. In contradiction with \citet{GPSS14} we were not able to produce a similar calibration quality as described in their paper\footnote{Several requests for clarification with the authors have resulted in the answer that their results are currently not in a state to be shared.}. The resulting calibration can be found in figure \ref{fig:capletfit}. Especially for long dated caplet volatilities the pronounced skew could not be reproduced. Nevertheless the model provides a reasonable fit of the caplet volatility surface, especially since the focus is on inflation derivatives. 
\begin{figure}
\centering
\includegraphics[width=0.55 \textwidth]{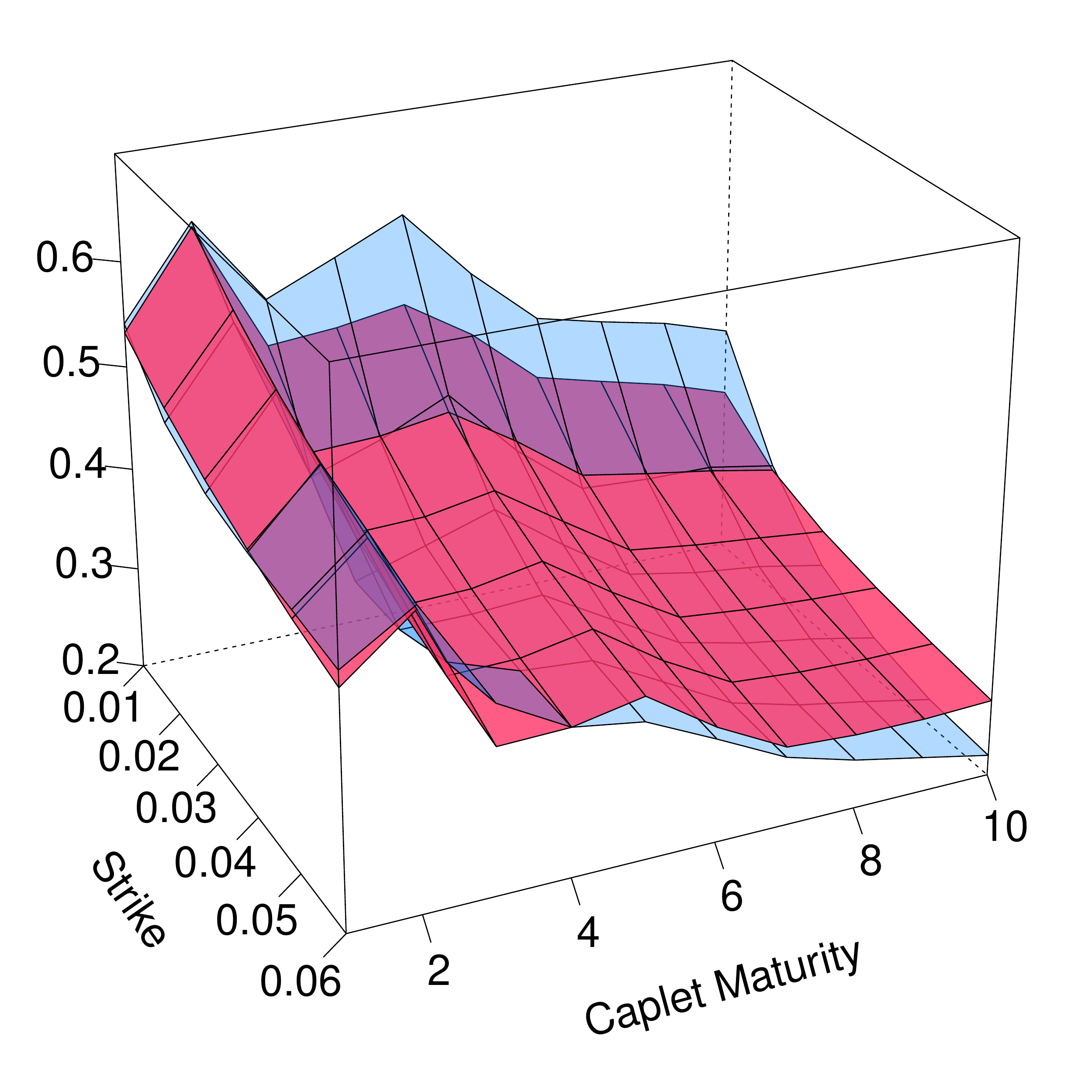}
\captionsetup{singlelinecheck=false, margin = 1.5 cm}
\caption{Market and model caplet implied volatilities of 6-month forward interest rates with strikes $1\%$ to $6\%$ and maturities from 1 to 10 years. Market volatilities are in a transparent blue, while model volatilities are displayed in red.}
\label{fig:capletfit}
\end{figure}

The next step is to extend the calibration to inflation markets. Additionally to the $M + 1$ driving processes used for modeling interest rates we use another $M$ independent analytic affine processes (one for each year) driving inflation-related quantities. Since the individual processes are independent by \eqref{eq:affinecombphipsi} the setup without inflation can be embedded by setting the additional components of $u_k$ equal to zero (see table \ref{tab:upar}). So from now on we assume that the processes $X^0,X^1,\dots, X^M$ and the vectors $u_1, \dots, u_N$ are already fixed.

The choice of the inflation parameters $v_k$ focuses on the same aspects aspects as the choice of $u_k$ without the restriction that inflation rates have to be nonnegative. We again assume that the vectors $v_k$ are determined by $2N$ parameters\footnote{We actually only require $N$ parameters, since the odd rows in table \ref{tab:uinfpar} do not contribute for the considered annual inflation rates. For notional simplicity we nevertheless consider $2N$ parameters.}
$\tilde{v}_1, \dots, \tilde{v}_N$, $\overline{v}_1, \dots, \overline{v}_N.$
In particular, we choose (see also table \ref{tab:uinfpar})
\begin{table}
\begin{center}
\begin{tabular}{l | C{2.5em}:C{2.5em}C{2.5em}C{2.5em}C{2.5em}:C{3.4 em}C{3.4 em}C{2.5em}C{2.5em} }
& $X^0$ & $X^1$ & $X^2$ & \dots & $X^{M}$ & $X^{M+1}$ & $X^{M+2}$ & \dots & $X^{2M}$ \\ \hdashline
\rowcolor{gray}
$v_1$ & $\tilde{v}_1$ & $\overline{u}_1$ & $\overline{u}_3$ & $\dots$ & $\overline{u}_{N-1}$ & $\overline{v}_1$ & $0$ & $\cdots$ & $0$ \\
$v_2$ & $\tilde{v}_2$ & $\overline{u}_2$ & $\overline{u}_3$ & $\dots$ & $\overline{u}_{N-1}$ & $\overline{v}_2$ & $0$ & $\cdots$ & $0$ \\
\rowcolor{gray} $v_3$ & $\tilde{v}_3$ & 0 & $\overline{u}_3$ & $\dots$ & $\overline{u}_{N-1}$ & $0$ & $\overline{v}_3$  & $\cdots$ & $0$ \\
$v_4$ & $\tilde{v}_4$ & 0 & $\overline{u}_4$ & $\dots$ & $\overline{u}_{N-1}$ & $0$ & $\overline{v}_4$  & $\cdots$ & $0$ \\
\vdots & \vdots & \vdots & \vdots & $\ddots$ & \vdots & \vdots & \vdots & $\ddots$ & \vdots \\
$v_{N-2}$ & $\tilde{v}_{N-2}$ & 0 & 0 & $\dots$ & $\overline{u}_{N-1}$ & $0$ &0 & $\cdots$ & 0 \\\rowcolor{gray} $v_{N-1}$ & $\tilde{v}_{N-1}$ & 0 & 0 & $\dots$ & $\overline{u}_{N-1}$ & $0$ &0 & $\cdots$ & $\overline{v}_{N-1}$ \\
$v_{N}$& $\tilde{v}_{N}$ & 0 & 0 & $\dots$ & $\overline{u}_{N}$ & $0$ &0 & $\cdots$ & $\overline{v}_{N}$ \\
\end{tabular}
\end{center}
\captionsetup{singlelinecheck=false, margin = .5 cm}
\caption{Description of the inflation parameter structure $v_k$. Each row corresponds to one vector with the column names denoting the process the position in the vector corresponds to. Note that for annual inflation rate option pricing the vectors $v_j$ with $j$ odd do not matter.}
\label{tab:uinfpar}
\end{table}
$$ v_k = \tilde{v}_k \e^0 + \overline{u}_k e^{\ceil{\frac{k}{2}}} + \sum_{l= \ceil{\frac{k}{2}}+1 }^M \overline{u}_{2l-1} e^l + \overline{v}_k e^{M+\ceil{\frac{k}{2}}}.$$
Choosing the nominal components $v_k^i = u_k^i$ for $i = 1, \dots, M$ has the advantage that forward CPIs and therefore also forward inflation rates do not depend on the nominal processes $X^1, \dots, X^M$. Together with the choice of $v_k$ with respect to the inflation processes $X^{M+1}, \dots X^{2M}$ this implies that the forward CPI $\I(t,T_{2k})$ depends only on $X^0$ and $X^{M+k}$. In particular the function $B(t,v_{2k},u_{2k})$ corresponding to the forward CPI $\I(t,T_{2k})$ defined in \eqref{eq:fwdCPIaffine} is
\begin{align*}
B(t,v_{2k},u_{2k}) & = ( \psi_{T-t}^0(\tilde{v}_{2k}) - \psi_{T-t}^0(\tilde{u}_{2k}), 0, \dots, 0, \psi_{T-t}^{M+k}(\overline{v}_{2k}), 0, \dots, 0).
\end{align*}
From this it follows that for different forward CPIs these functions are \enquote{orthogonal} except for the first component. Furthermore except for the first component they are also \enquote{orthogonal} to the functions $B(t,u_{j-1},u_j)$ relevant for forward interest rates. Hence the correlation structure mainly depends on $\tilde{u}_k$ and $\tilde{v}_k$. Since $\psi_t^0$ is monotonically increasing, one should choose $\tilde{v}_{k} >  \tilde{u}_{k}$ if the corresponding forward CPI should be positively correlated with nominal interest rates or $\tilde{v}_{k} < \tilde{u}_{k}$ if it should be negatively correlated\footnote{We decided to introduce correlation for the inflation part only via the common factor process. We could also have changed the parameters $\overline{u}_k$ in order to introduce correlation. In this case even more complicated correlation patterns could be created.}. Also note that two annual forward CPIs are positively correlated if $\sign{\tilde{v}_k - \tilde{u}_k} = \sign{\tilde{v}_j - \tilde{u}_j}$. This therefore gives us some criteria how to determine the parameters $\tilde{u}_k$ from correlation assumptions and from now on we assume that the parameters $\tilde{v}_1, \dots, \tilde{v}_N$ are given. 
Assuming also a fixed process $X^{M+k}$ we would like to determine $\overline{v}_k$ from the current term structure $ P_{ILB}(0,T_k) / P(0,T)$. We differentiate between two cases. First consider a nonnegative affine process $X^{M+k}$. 
In this case by Lemma \ref{lem:ufitting} there exists a unique $\overline{v}_k$, so that $M_0^{v_k} = P_{ILB}(0,T_k) / P(0,T)$. For $\tilde{v}_k \approx \tilde{u}_k$ this typically results in $\overline{v}_k > 0$. In this case zero coupon inflation for $[t,T_k]$ is always positive. To avoid this alternatively consider an affine process, which takes negative and positive values. In this case $M_0^v$ is not necessarily increasing in  $\overline{v}_k$. However, by Lemma \ref{lem:ufitting} it is still a convex function in $\overline{v}_k$. This means that there are at most two possible choices for $\overline{v}_k$, in which case we need to pick one. For the results in this chapter we only used results where one choice was smaller than zero and the other one larger than zero, which we then picked.
To determine the processes $X^{M+1},\dots,X^{2M}$ notice that the annual forward inflation rate $F_I(t,T_{2k-2},T_{2k})$ depends only the processes $X^0, X^{M+k-1},X^k$ and $F_I(t,T_0,T_2)$ depends only on $X^0, X^{M+1}$. By starting with inflation options on $F_I(T_2,T_0,T_2)$ one can calibrate the parameters of $X^{M+1}$. Then one can iteratively calibrate the parameters of $X^{M+k}$ using inflation options on $F_I(T_{2k},T_{2(k-1)},T_{2k})$. 

For the calibration example we used ZCIIS rates from $1$ to $10$ years (see figure \ref{fig:ZCIIScurve}) and inflation options for strikes ranging from $-2\%$ to $6\%$, the prices of which are displayed in figure \ref{fig:inflfit}.
We chose $\tilde{v}_k = \tilde{u}_k (1+c k)$ with $c \approx 0.08$, so that $c k$  is between $1$ and $1.15$. This means that correlation are positive as usually observed. For the processes $X^{M+1}, \dots, X^{2M}$ we used Ornstein-Uhlenbeck processes with added jumps (see equation \eqref{eq:DGamOUBM}). 
The parameters of these processes and the sequence $(\overline{v}_k)$ are then calibrated as described. In each step the mean squared errors of option prices are minimized. The resulting fit is displayed in figure \ref{fig:inflfit} and this shows that the calibration is very accurate.
\begin{figure}
\centering
\includegraphics[width=0.55 \textwidth]{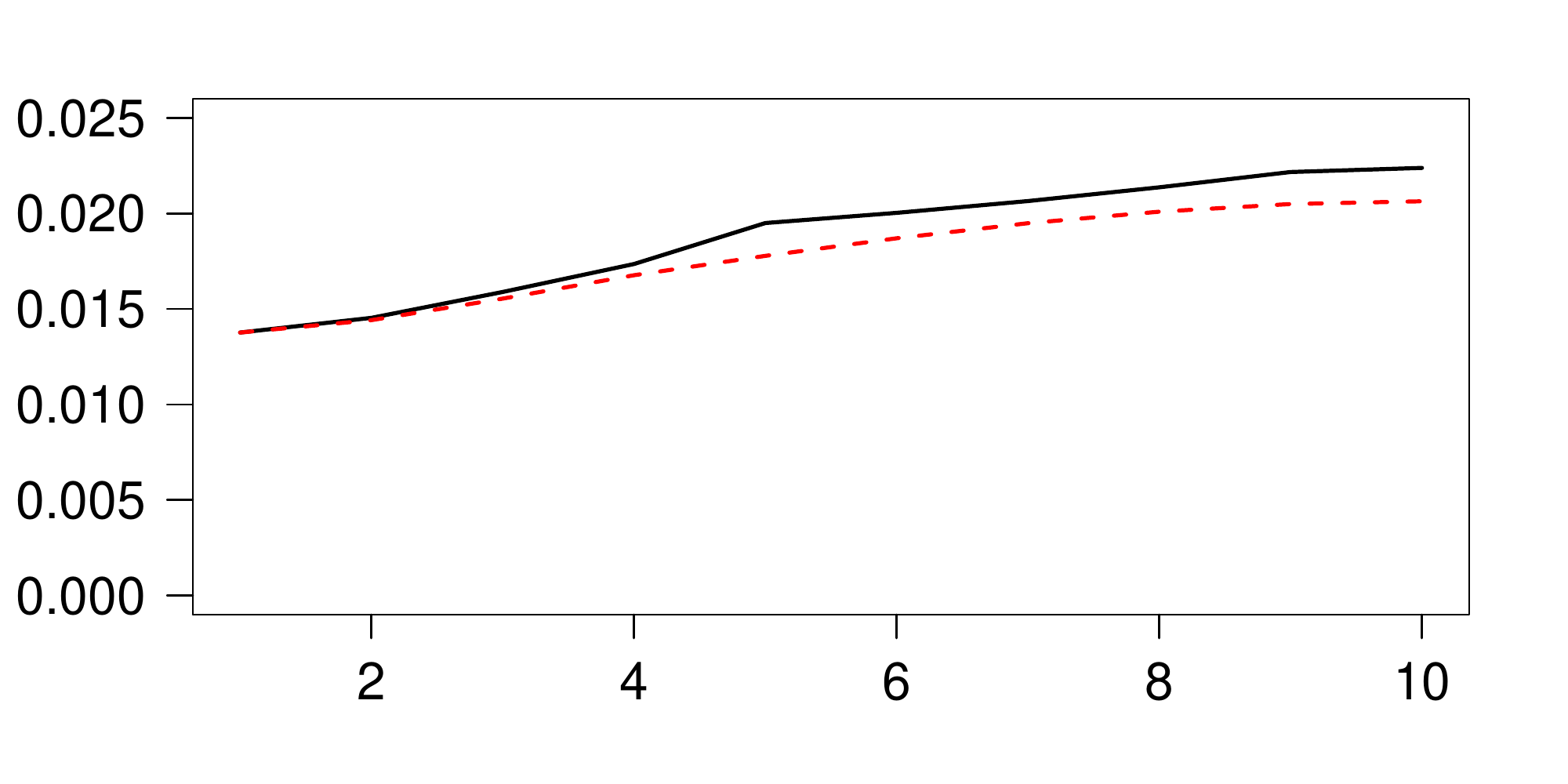}
\captionsetup{singlelinecheck=false, margin = 1.5 cm}
\caption{\small Linear interpolated annual forward inflation rates $F_I(0,T_{2(k-1)},T_{2k})$ (black) and its approximation $\I(t,T_{2k})/\I(t,T_{2(k-1)})-1$ (dashed red) for maturities from 1 to 10 years.} 
\label{fig:forwardinflation}
\end{figure}
\begin{figure}
\centering
\includegraphics[width=0.52 \textwidth]{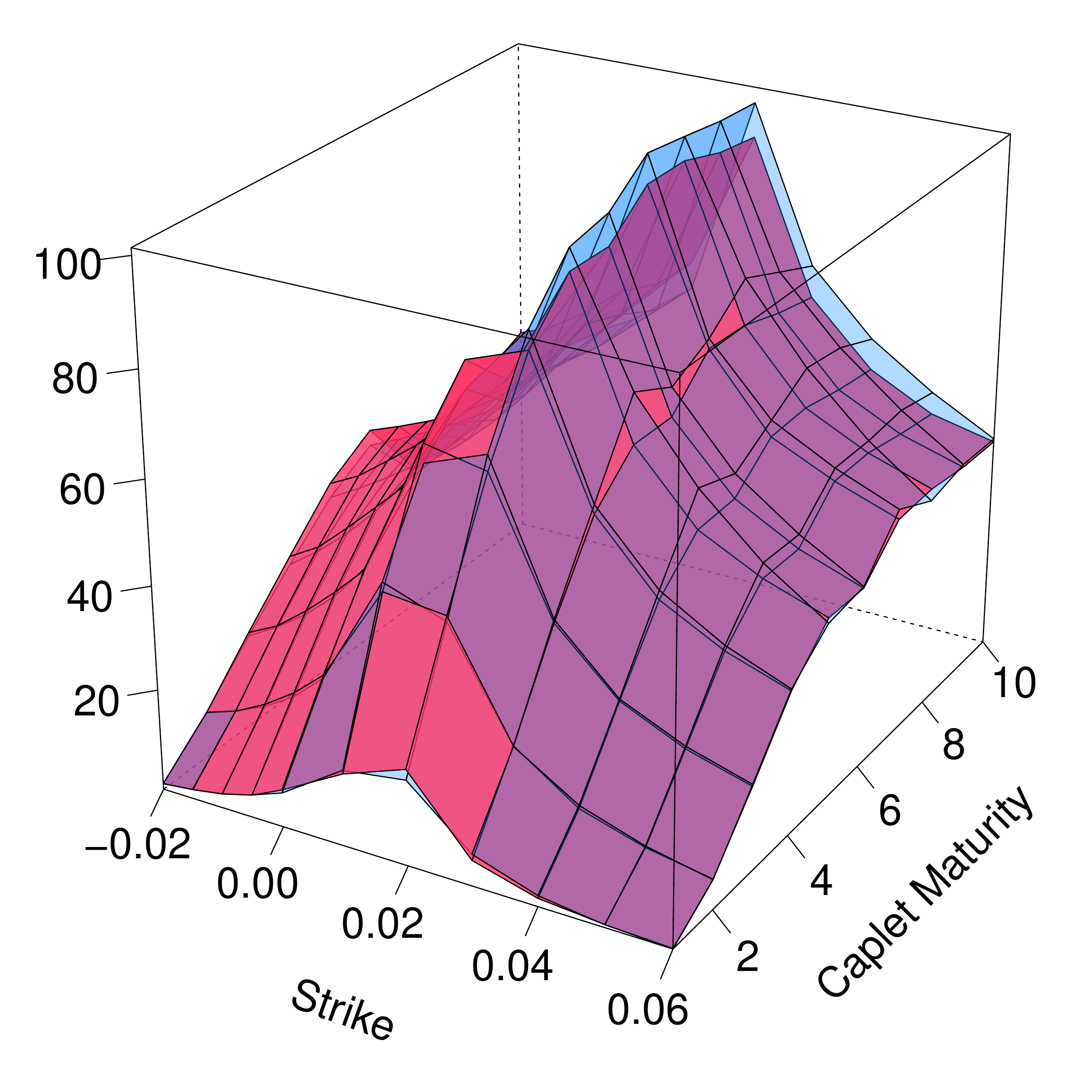}
\captionsetup{singlelinecheck=false, margin = 1.5 cm}
\caption{\small Market and model caplet/floorlet prices in basis points for annual forward inflation with strikes $-2\%$ to $6\%$ and maturities from 1 to 10 years. Market prices are in a transparent blue, while model prices are displayed in red. For strikes between $-2\%$ and $1\%$ prices are quoted for floorlets, for strikes between $2\%$ and $6\%$ prices are quoted for caplets. Market prices are bootstrapped from corresponding cap/floor data.}
\label{fig:inflfit}
\end{figure}
\begin{figure}
\centering
\includegraphics[width=0.52 \textwidth]{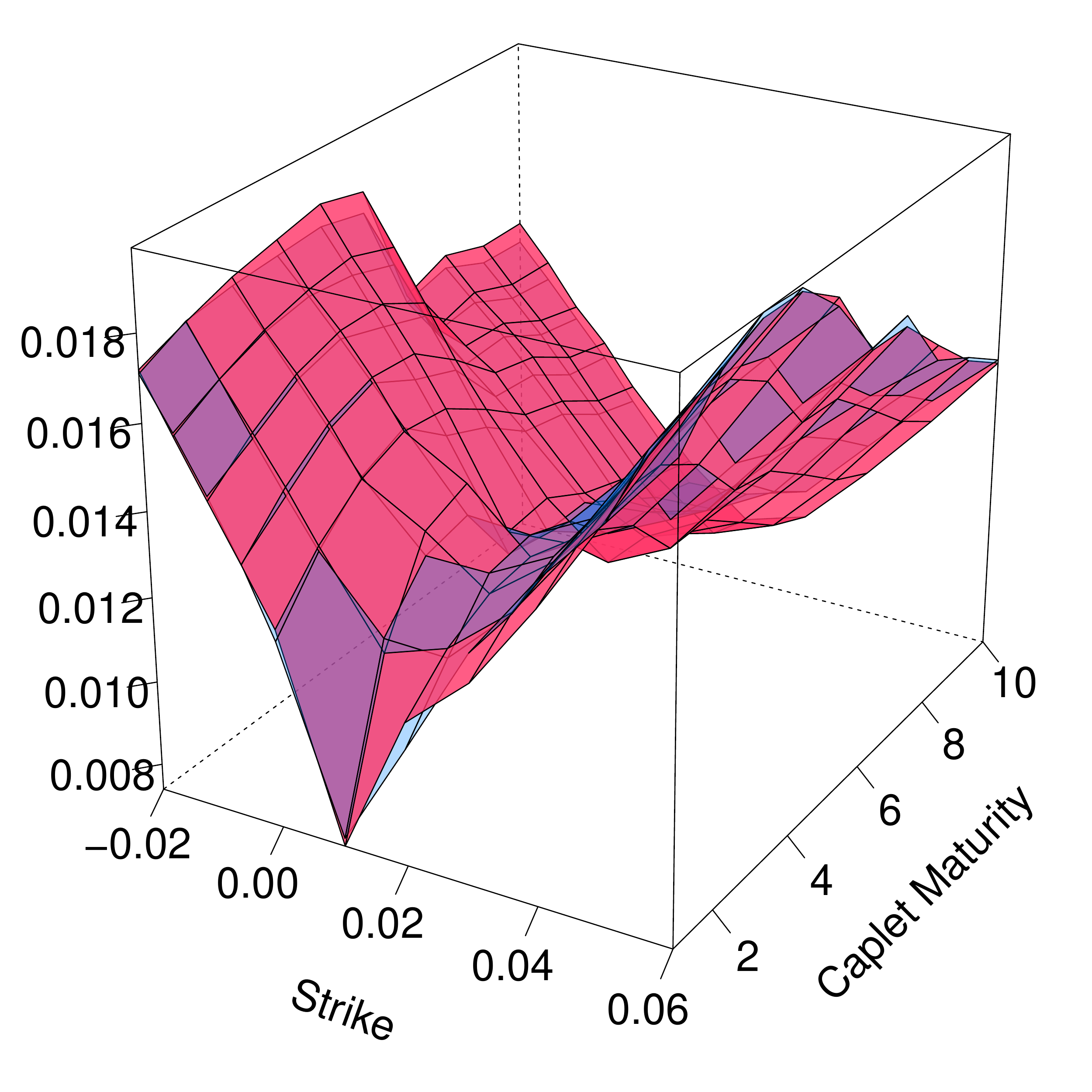}
\captionsetup{singlelinecheck=false, margin = 1.5 cm}
\caption{\small Market and model implied volatilities for annual forward inflation with strikes $-2\%$ to $6\%$ and maturities from 1 to 10 years. Market volatilities are in a transparent blue, while model volatilities are displayed in red.
}
\label{fig:inflfitvol}
\end{figure}

We would also like to display the fit in terms of (shifted lognormal) implied volatilities. Typically one only has quotes on ZCIIS rates,  which do not directly translate to forward inflation rates. However, annual forward inflation rates can be approximated by $F_I(t,T_{k-j},T_k) \approx \I(t,T_k) / \I(t,T_{k-j})-1$ (see figure \ref{fig:forwardinflation}). It is market practice to use this as forward value in the shifted Black formula to calculate approximate market volatilities. The resulting fit in terms of shifted implied volatilities is then displayed in figure \ref{fig:inflfitvol}. This shows that implied volatilities are closely reproduced by the model across all maturities and smiles and that the model is very capable in fitting the observed market data.

\subsection*{Conclusion}
We have introduced a highly tractable inflation market model, where we are able to derive analytical formulas for both types of inflation-indexed swaps. Furthermore inflation caps and floors as well as CPI caps and floors can be calculated with a one-dimensional Fourier inversion formula. Hence prices for liquidly traded inflation derivatives can be calculated quickly and accurately. Additionally the proposed model is able to price classical interest rate derivatives like caps and floors. Using these formulas we are able to calibrate the model to market data. The calibration example shows that the model can be calibrated to inflation market data very accurately.

\section[Appendix]{Appendix on affine processes} \label{sec:affineprocess}
Let  $X = (X_t)_{0 \leq t \leq T}$ be a homogeneous Markov process with values in $D = \R^m_{\geq 0} \times \R^n$ on a measurable space $(\Omega,\mathcal{A})$ with filtration  $(\F_t)_{0 \leq t \leq T}$, with regards to which $X$ is adapted. Denote by $\PM^x$ and $\EV[x]{\cdot}$ the corresponding probability and expectation when $X_0 = x$. 
X is said to be an affine process, if its characteristic function has the form
\begin{equation*} 
\EV[x]{ \e^{u\T X_{t}}}  = \ex{\phi_{t}(u) + \psi_{t}(u)\T x}, \quad u \in \i \R^d, x \in D,
\end{equation*}
where $\phi: [0,T]\times \i \R^d  \rightarrow \C$ and $\psi: [0,T] \times \i \R^d  \rightarrow \C^d$ and $\i \R^d = \{u \in \C^d: \mathrm{Re}(u) = 0 \}.$ 
By homogeneity and the Markov property the conditional characteristic function satisfies
\begin{equation}  \label{eq:affineMomentGeninfl}
 \EV[x]{e^{u\T X_{t}} \vert \F_s} = \ex{\phi_{t-s}(u) + \psi_{t-s}(u)\T X_s}. 
\end{equation}
Accordingly affine processes can also be defined for inhomogeneous Markov processes (see \citet{FI05}). In this case the affine property reads
\begin{equation*} \label{eq:affineMomentGeninhominfl}
\EV[x]{ \e^{u\T X_t} \vert \F_s} = \ex{\phi_{s,t}(u) + \psi_{s,t}(u)\T X_s}, \quad u \in \i \R^d, x \in D, 
\end{equation*}
with $\phi_{s,t}: \i \R^d  \rightarrow \C$ and $\psi_{s,t}: \i \R^d  \rightarrow \C^d$ for $0 \leq s \leq t$.

$X$ is called an analytic affine process, if $X$ is stochastically continuous and the interior of the set\footnote{$\mathcal{V}$ can be described as the (convex) set, where the extended moment generating function of $X_t$ is defined for all times $t$ and all starting values $x$. By Lemma 4.2 in \citet{KM11} the set $\mathcal{V}$ is in fact equal to the seemingly smaller set 
$\left \{ u \in \C^d:   \exists x \in \mathrm{int}(D): \EV[x]{ \e^{\mathrm{Re}(u)\T X_T}} < \infty \right \}.$}
\begin{align}
\mathcal{V} & := \left \{ u \in \C^d: \sup_{0 \leq s \leq T}  \EV[x]{ \e^{\mathrm{Re}(u)\T X_s}} < \infty \quad  \forall x \in D \right \}, \label{eq:momset}
\end{align}
contains 0
. In this case the functions $\phi$ and $\psi$ have continuous extensions to $\mathcal{V}$, which are analytic in the interior, such that \eqref{eq:affineMomentGeninfl} holds for all $u \in \mathcal{V}$ (see \citet{KR08}).

The class of affine processes includes Brownian motion and more generally all Lévy processes. Since Lévy processes have stationary independent increments, it follows that $\psi_t(u) = u$, while $\phi_t(u) = t \kappa(u)$, where $\kappa$ is the cumulant generating function of the Lévy process. Ornstein-Uhlenbeck processes are further important examples of affine processes. The affine processes used in this work are described at the end of this section. 

The standard reference for affine processes is \citet{DFS03}. There they give a characterization of affine processes, where $\phi$ and $\psi$ are specified as solutions of a system of differential equations
. 
To motivate this consider an affine process $X$. By the tower property for conditional expectations it holds for all $x \in D$
$$ \EV[x]{e^{u X_{t+s}}} = \EV[x]{\EV[x]{e^{u X_{t+s}}\vert \F_s} } =  \EV[x]{e^{\phi_t(u) + \psi_t(u)\T X_s} }. $$
Using equation \eqref{eq:affineMomentGeninfl} it follows that $\phi$ and $\psi$ satisfy the so-called semi-flow equations
\begin{equation}
\begin{aligned}
\phi_{t+s}(u) & = \phi_t(u) + \phi_s(\psi_t(u)),  \qquad & \phi_0(u) & = 0, \\
\psi_{t+s}(u) & = \psi_s(\psi_t(u)),  & \psi_0(u)  & = u.
\end{aligned} \label{eq:semiflowinfl}
\end{equation}
For a stochastically continuous affine process $X$ it was shown in \citet{KST11} that the functions 
\begin{equation*}
F(u) := \left. \frac{\partial}{\partial t} \phi_t(u) \right \vert_{t=0^+}, \qquad R(u) := \left. \frac{\partial}{\partial t} \psi_t(u) \right \vert_{t=0^+}
\end{equation*}
exist\footnote{This was also shown for affine processes with general state spaces in \citet{KST11b} and \citet{CT13}.}.
Rewriting \eqref{eq:semiflowinfl} in terms of difference quotients and letting  $s \rightarrow 0$ we get that $\phi$ and $\psi$ satisfy generalized Riccati equations
\begin{equation} \label{eq:ricattiinfl}
\begin{aligned}
 \frac{\partial}{\partial t} \phi_t(u) & = F(\psi_t(u)), \qquad \phi_0(u) = 0, \\
 \frac{\partial}{\partial t} \psi_t(u) & = R(\psi_t(u)), \qquad \phi_0(u) = u. 
\end{aligned}
\end{equation}
The functions $F$ and $R$ have a specific form of Levy-Khintchine type as first described in \citet{DFS03}. There it is also shown that for every $F$ and $R$ of this form \eqref{eq:ricattiinfl} has a unique solution. Specifying the functions $F$ and $R$ is an alternative way to specify an affine process. \citet{KM11} give conditions on $F$ and $R$ under which a solution of \eqref{eq:ricattiinfl} defines an analytic affine process. Note that in order to evaluate $\phi$ and $\psi$ one would like to have closed form solutions to the system \eqref{eq:ricattiinfl}, which in general is not the case.

Coupling independent affine processes is very tractable. For two independent affine processes $X$ and $Y$ and all starting values $x, y$ one obtains
\begin{equation} \label{eq:affinecomb}
\EV[(x,y)]{\e^{(u_X, u_Y) \cdot (X_t,Y_t)}} = \EV[(x,y)]{\e^{u_X\T X_t} \e^{u_Y\T Y_t} } = \EV[(x,y)]{\e^{u_X \cdot X_t} } \EV[(x,y)]{\e^{u_Y \cdot Y_t} }.\end{equation}
Hence $(X,Y)$ is an affine process with
\begin{equation} \label{eq:affinecombphipsi}
\begin{aligned}
\phi^{(X,Y)}_t(u_X,u_Y) & = \phi_t^X (u_X)  + \phi_t^Y (u_Y),  \\
\psi^{(X,Y)}_t(u_X,u_Y) & = (\psi_t^X(u_X), \psi_t^Y(u_Y)) .
\end{aligned}
\end{equation}
This fact together with the following Lemma is used in section \ref{sec:numeric}.
\begin{lemma} \label{lem:ufitting}
Let $X$ be an analytic affine processes comprised of $m+n$ independent affine processes, where the first $m$ are nonnegative. For $(u_1, \dots, u_k, v, u_{k+1}, \dots, u_n) \in \mathrm{int}(\mathcal{V}) \cap \R^{m+n}$ define the function
$$f^k(v) := \EV[x]{\e^{(u_1,\dots,u_{k-1},v,u_{k+1},\dots,u_{m+n}) \cdot X_t}}.$$
Then $f^k$ is monotonically increasing if $k \leq m$ and $f^k$ is convex for all $k$. 
\end{lemma}
\begin{proof}
For each $x \in D$ the term inside the expectation is convex in $v$ and monotonically increasing in $v$ if $k \leq m$. This then also holds after taking the expectation. 
\end{proof}

The last part of this section describes the affine processes used in this chapter. One classical example is the CIR process, which is the unique solution of
\begin{equation} \label{eq:CIR}
\dd{X_t} = - \lambda (X_t - \theta) \dd{t} + 2\eta \sqrt{X_t} \dd{W_t}, \qquad X_0 = x.
\end{equation}
For this process the functions $\phi$ and $\psi$ are defined for $\mathrm{Re}(u) < \frac{\lambda}{2 \eta^2} (1-\e^{-\lambda t})^{-1}$,
\begin{equation*}
\begin{aligned}
\phi_t(u) & = - \frac{\lambda \theta}{2 \eta^2} \logn{1- \frac{2\eta^2}{\lambda} (1-\e^{-\lambda t})  u}, \\
\psi_t(u) & = \frac{\e^{-\lambda t} u}{1- \frac{2\eta^2}{\lambda} (1-\e^{-\lambda t})  u}. 
\end{aligned}
\end{equation*}
The CIR process almost surely stays nonnegative. It is strictly positive if $\frac{\lambda \theta}{2} > \eta^2$. One can add jumps to this process by adding the differential of a compound Poisson process $L_t$  to the dynamics of $X$. 
\begin{equation} \label{eq:CIRGamOU}
\dd{X_t} = - \lambda (X_t - \theta) \dd{t} + 2\eta \sqrt{X_t} \dd{W_t} + \dd{L_t}, \qquad X_0 = x.
\end{equation}
If $L_t$ has exponentially distributed jumps with expectation values $\frac{1}{\alpha}$ arriving at rate $\lambda \beta$ the functions $\phi$ and $\psi$ are (see \citet{GP13})
\begin{equation*}
\begin{aligned}
\phi_t(u)  = & - \frac{\lambda \theta}{2 \eta^2} \logn{1- \frac{2\eta^2}{\lambda} (1-\e^{-\lambda t})  u} \\
& \ - \frac{\lambda \beta}{\lambda - 2 \eta^2 \alpha} \logn{\frac{\alpha-u}{\alpha-u\left(\e^{-\lambda t}+(1-\e^{-\lambda t})\frac{1}{\lambda} 2 \eta^2 \alpha\right)}}, \\
\psi_t(u) = & \frac{\e^{-\lambda t} u}{1- \frac{2\eta^2}{\lambda} (1-\e^{-\lambda t})  u},
\end{aligned}
\end{equation*}
where $$\mathrm{Re}(u) < \min\left\{ \frac{\lambda}{2 \eta^2} (1-\e^{-\lambda t})^{-1}, \alpha \left(\e^{-\lambda t}+(1-\e^{-\lambda t})\frac{1}{\lambda} 2 \eta^2 \alpha\right)^{-1} , \alpha \right\}.$$ Since $L_t$ has only positive jumps this process also stays nonnegative. 
As a third example consider the real-valued affine process defined by
\begin{equation} \label{eq:DGamOUBM}
\dd{X_t} = - \lambda (X_t - \theta) \dd{t} + \sigma \dd{W_t} + \dd{\tilde{L}_t}, \qquad X_0 = x,
\end{equation}
where $\tilde{L}_t$ is a compound Poisson process with positive jumps with mean $\frac{1}{\alpha^+}$ arriving at rate $\lambda \beta^+$ and negative jumps with mean $\frac{1}{\alpha^-}$ arriving at rate $\lambda \beta^-$. The functions $\phi$ and $\psi$ in this case read  (see \citet{MW14})
\begin{align*}
\phi_t(u) & = \frac{\sigma^2 u^2}{4 \lambda} (1-\e^{-2 \lambda t}) + \theta u  (1-\e^{-\lambda t})  + \frac{\beta^+ + \beta^-}{2} \logn{\frac{(\alpha^+-\e^{-\lambda t}u)(\alpha^-+\e^{-\lambda t}u)}{(\alpha^+-u)(\alpha^-+u)}} \\
& + \frac{\beta^+ - \beta^-}{2 } \logn{\frac{(\alpha^+-\e^{-\lambda t}u)(\alpha^-+u)}{(\alpha^+-u)(\alpha^-+\e^{-\lambda t}u)}}, \\
\psi_t(u) & = \e^{-\lambda t} u,
\end{align*}
for $- \alpha^- < \mathrm{Re}(u) < \alpha^+$. 